\documentclass{LMCS}

\def\dOi{11(3:7)2015}
\lmcsheading%
{\dOi}
{1--33}
{}
{}
{May\phantom.~17, 2014}
{Sep.~\phantom03, 2015}
{}

\ACMCCS{[{\bf Theory of computation}]: Logic---Proof theory\,/\,Modal and temporal logics}

\usepackage{color} 
\usepackage{caption} 
\usepackage{subcaption} 
\usepackage{graphics}
\usepackage{amssymb} 
\usepackage{stmaryrd}
\usepackage{amsmath}
\usepackage{latexsym}
\usepackage{amsfonts}
\usepackage{ifsym}
\usepackage{xspace}
\usepackage{colonequals}
\usepackage{wasysym}
\usepackage{hyperref}

\usepackage{listings}

\usepackage[curve,matrix,arrow]{xy}
\usepackage[noxy]{virginialake}
\vlnosmallleftlabels
\newcommand{\hide}[1]{} 
\newcommand{\RInt}{R}

\theoremstyle{theorem}
\newtheorem{theorem}{Theorem}[section]
\newtheorem{lemma}[theorem]{Lemma} \newtheorem{proposition}[theorem]{Proposition}

\newtheorem{conjecture}[theorem]{Conjecture}

\theoremstyle{definition}

\newtheorem{example}[theorem]{Example}
\newtheorem{definition}[theorem]{Definition}

\newtheorem{observation}[theorem]{Observation}
\newtheorem{remark}[theorem]{Remark}
\newtheorem{fakt}[theorem]{Fact}

\newcommand{\refereedone}[1]{}
\newcommand{\ryutadone}[1]{}
\newcommand{\anupamdone}[1]{}
\newcommand{\lutzdone}[1]{}

\def\cM{{\mathcal M}}

\def\Nat{\mathbb{N}}

\def\K{\mathsf{K}}

\def\set#1{\{#1\}}

\def\cons#1{\{#1\}}
\def\conhole      {\cons{\enspace}}%

\def\tuple#1{\langle#1\rangle}

\def\grammareq {\mathrel{\raise.4pt\hbox{::}{=}}}%

\newcommand{\dotto}[1][]{\mathrel{\!\xy\ar@{.>}^-{#1}(5,0)\endxy\!}}
\newcommand{\solto}[1][]{\mathrel{\!\xy\ar@{->}^-{#1}(5,0)\endxy\!}}
\newcommand{\longsolto}[1][]{\mathrel{\!\xy\ar@{->}^-{#1}(11,0)\endxy\!}}
\newcommand{\longdotto}[1][]{\mathrel{\!\xy\ar@{.>}^-{#1}(11,0)\endxy\!}}
\newcommand{\xldotto}[2][]{\mathrel{\!\xy\ar@{.>}^-{#1}(#2,0)\endxy\!}}

\def\cutr{\mathsf{cut}}
\def\idr{\mathsf{id}}

\def\weakr{\mathsf{weak}}
\def\conr{\mathsf{cont}}

\def\rr{\mathsf{r}}

\def\cand{\wedge}
\def\cor{\vee}

\newbox\cutbox
\newdimen\cutwd
\newdimen\cutht
\newdimen\cutdp
\def\ccut{%
  \setbox\cutbox\hbox{$\lozenge$}
  \cutwd=\wd\cutbox
  \cutht=\ht\cutbox
  \cutdp=\dp\cutbox
  \setbox\cutbox\hbox to\cutwd{\hss\vrule width.3pt height\cutht depth\cutdp\hss}
  \mathbin{\lozenge\hskip-\cutwd\copy\cutbox}}
\def\scriptcut{%
  \setbox\cutbox\hbox{$\scriptstyle\lozenge$}
  \cutwd=\wd\cutbox
  \cutht=\ht\cutbox
  \cutdp=\dp\cutbox
  \setbox\cutbox\hbox to\cutwd{\hss\vrule width.3pt height\cutht depth\cutdp\hss}
  \mathord{\lozenge\hskip-\cutwd\copy\cutbox}}

\def\vccut{%
  \setbox\cutbox\hbox{$\lozenge$}
  \cutwd=\wd\cutbox
  \cutht=\ht\cutbox
  \cutdp=\dp\cutbox
  \setbox\cutbox\hbox to\cutwd{\hss\hskip.3pt\vrule width.3pt height\cutht depth\cutdp\hss}
  \mathbin{\lozenge\hskip-\cutwd\copy\cutbox}}

\def\conldel {\{}%
\def\conrdel {\}}%
\def\lrgldel {\mathchoice{(}{(}{\langle}{\langle}}%
\def\lrgrdel {\mathchoice{)}{)}{\rangle}{\rangle}}%
\def\aprldel {\mathchoice
    {\mathopen {\setbox0=\hbox{$\displaystyle     \lrgldel$}\hbox to\wd0
                         {\hfil$\displaystyle     (       $\hfil}}}%
    {\mathopen {\setbox0=\hbox{$\textstyle        \lrgldel$}\hbox to\wd0
                         {\hfil$\textstyle        (        $\hfil}}}%
    {\mathopen {\setbox0=\hbox{$\scriptstyle      \lrgldel$}\hbox to\wd0
                         {\hfil$\scriptstyle      (        $\hfil}}}%
    {\mathopen {\setbox0=\hbox{$\scriptscriptstyle\lrgldel$}\hbox to\wd0
                         {\hfil$\scriptscriptstyle(        $\hfil}}}}%
\def\aprrdel {\mathchoice
    {\mathclose{\setbox0=\hbox{$\displaystyle     \lrgrdel$}\hbox to\wd0
                         {\hfil$\displaystyle     )       $\hfil}}}%
    {\mathclose{\setbox0=\hbox{$\textstyle        \lrgrdel$}\hbox to\wd0
                         {\hfil$\textstyle        )        $\hfil}}}%
    {\mathclose{\setbox0=\hbox{$\scriptstyle      \lrgrdel$}\hbox to\wd0
                         {\hfil$\scriptstyle      )        $\hfil}}}%
    {\mathclose{\setbox0=\hbox{$\scriptscriptstyle\lrgrdel$}\hbox to\wd0
                         {\hfil$\scriptscriptstyle)        $\hfil}}}}%
\def\seqldel {\mathchoice
    {\mathopen {\setbox0=\hbox{$\displaystyle     \lrgldel$}\hbox to\wd0
                         {\hfil$\displaystyle     \langle  $\hfil}}}%
    {\mathopen {\setbox0=\hbox{$\textstyle        \lrgldel$}\hbox to\wd0
                         {\hfil$\textstyle        \langle  $\hfil}}}%
    {\mathopen {\setbox0=\hbox{$\scriptstyle      \lrgldel$}\hbox to\wd0
                         {\hfil$\scriptstyle      \langle  $\hfil}}}%
    {\mathopen {\setbox0=\hbox{$\scriptscriptstyle\lrgldel$}\hbox to\wd0
                         {\hfil$\scriptscriptstyle\langle  $\hfil}}}}%
\def\seqrdel {\mathchoice
    {\mathclose{\setbox0=\hbox{$\displaystyle     \lrgrdel$}\hbox to\wd0
                         {\hfil$\displaystyle     \rangle  $\hfil}}}%
    {\mathclose{\setbox0=\hbox{$\textstyle        \lrgrdel$}\hbox to\wd0
                         {\hfil$\textstyle        \rangle  $\hfil}}}%
    {\mathclose{\setbox0=\hbox{$\scriptstyle      \lrgrdel$}\hbox to\wd0
                         {\hfil$\scriptstyle      \rangle  $\hfil}}}%
    {\mathclose{\setbox0=\hbox{$\scriptscriptstyle\lrgrdel$}\hbox to\wd0
                         {\hfil$\scriptscriptstyle\rangle  $\hfil}}}}%
\def\parldel {\mathchoice
    {\mathopen {\setbox0=\hbox{$\displaystyle     \lrgldel$}\hbox to\wd0
                         {\hfil$\displaystyle     [       $\hfil}}}%
    {\mathopen {\setbox0=\hbox{$\textstyle        \lrgldel$}\hbox to\wd0
                         {\hfil$\textstyle        [        $\hfil}}}%
    {\mathopen {\setbox0=\hbox{$\scriptstyle      \lrgldel$}\hbox to\wd0
                         {\hfil$\scriptstyle      [        $\hfil}}}%
    {\mathopen {\setbox0=\hbox{$\scriptscriptstyle\lrgldel$}\hbox to\wd0
                         {\hfil$\scriptscriptstyle[        $\hfil}}}}%
\def\parrdel {\mathchoice
    {\mathclose{\setbox0=\hbox{$\displaystyle     \lrgrdel$}\hbox to\wd0
                         {\hfil$\displaystyle     ]       $\hfil}}}%
    {\mathclose{\setbox0=\hbox{$\textstyle        \lrgrdel$}\hbox to\wd0
                         {\hfil$\textstyle        ]        $\hfil}}}%
    {\mathclose{\setbox0=\hbox{$\scriptstyle      \lrgrdel$}\hbox to\wd0
                         {\hfil$\scriptstyle      ]        $\hfil}}}%
    {\mathclose{\setbox0=\hbox{$\scriptscriptstyle\lrgrdel$}\hbox to\wd0
                         {\hfil$\scriptscriptstyle]        $\hfil}}}}%

\def\eightpoint{\small}                         
\def\pluldel {\mathchoice
   {\mathopen {\setbox0=\hbox{$\displaystyle     \lrgldel$}\hbox to\wd0
                        {\hfil$\displaystyle     [       $\hfil}%
                        \kern-\wd0\hbox to\wd0
                        {\hss$\vcenter{\hbox{\eightpoint$\scriptscriptstyle\bullet$}}$\hss}}}%
   {\mathopen {\setbox0=\hbox{$\textstyle        \lrgldel$}\hbox to\wd0
                        {\hfil$\textstyle        [       $\hfil}%
                        \kern-\wd0\hbox to\wd0
                        {\hss$\vcenter{\hbox{\eightpoint$\scriptscriptstyle\bullet$}}$\hss}}}%
   {\mathopen {\setbox0=\hbox{$\scriptstyle      \lrgldel$}\hbox to\wd0
                        {\hfil$\scriptstyle      [       $\hfil}%
                        \kern-\wd0\hbox to\wd0
                        {\hss$\raise.1ex\hbox{\eightpoint$\scriptscriptstyle\bullet$}$\hss}}}%
   {\mathopen {\setbox0=\hbox{$\scriptscriptstyle\lrgldel$}\hbox to\wd0
                        {\hfil$\scriptscriptstyle[       $\hfil}%
                        \kern-\wd0\hbox to\wd0
                        {\hss$\raise.03ex\hbox{\eightpoint$\scriptscriptstyle\bullet$}$\hss}}}}%
\def\plurdel {\mathchoice
   {\mathclose{\setbox0=\hbox{$\displaystyle     \lrgldel$}\hbox to\wd0
                        {\hfil$\displaystyle     ]       $\hfil}%
                        \kern-\wd0\hbox to\wd0
                        {\hss$\vcenter{\hbox{\eightpoint$\scriptscriptstyle\bullet$}}$\hss}}}%
   {\mathclose{\setbox0=\hbox{$\textstyle        \lrgldel$}\hbox to\wd0
                        {\hfil$\textstyle        ]       $\hfil}%
                        \kern-\wd0\hbox to\wd0
                        {\hss$\vcenter{\hbox{\eightpoint$\scriptscriptstyle\bullet$}}$\hss}}}%
   {\mathclose{\setbox0=\hbox{$\scriptstyle      \lrgldel$}\hbox to\wd0
                        {\hfil$\scriptstyle      ]       $\hfil}%
                        \kern-\wd0\hbox to\wd0
                        {\hss$\raise.1ex\hbox{\eightpoint$\scriptscriptstyle\bullet$}$\hss}}}%
   {\mathclose{\setbox0=\hbox{$\scriptscriptstyle\lrgldel$}\hbox to\wd0
                        {\hfil$\scriptscriptstyle]       $\hfil}%
                        \kern-\wd0\hbox to\wd0
                        {\hss$\raise.03ex\hbox{\eightpoint$\scriptscriptstyle\bullet$}$\hss}}}}%
\def\witldel {\mathchoice
   {\mathopen {\setbox0=\hbox{$\displaystyle     \lrgldel$}\hbox to\wd0
                        {\hfil$\displaystyle     (       $\hfil}%
                        \kern-\wd0\hbox to\wd0
                        {\hss$\vcenter{\hbox{\eightpoint$\scriptscriptstyle\bullet\mkern3.2mu$}}$\hss}}}%
   {\mathopen {\setbox0=\hbox{$\textstyle        \lrgldel$}\hbox to\wd0
                        {\hfil$\textstyle        (       $\hfil}%
                        \kern-\wd0\hbox to\wd0
                        {\hss$\vcenter{\hbox{\eightpoint$\scriptscriptstyle\bullet\mkern3.2mu$}}$\hss}}}%
   {\mathopen {\setbox0=\hbox{$\scriptstyle      \lrgldel$}\hbox to\wd0
                        {\hfil$\scriptstyle      (       $\hfil}%
                        \kern-\wd0\hbox to\wd0
                        {\hss$\raise.1ex\hbox{\eightpoint$\scriptscriptstyle\bullet\mkern3.2mu$}$\hss}}}%
   {\mathopen {\setbox0=\hbox{$\scriptscriptstyle\lrgldel$}\hbox to\wd0
                        {\hfil$\scriptscriptstyle(       $\hfil}%
                        \kern-\wd0\hbox to\wd0
                        {\hss$\raise.03ex\hbox{\eightpoint$\scriptscriptstyle\bullet\mkern3.2mu$}$\hss}}}}%
\def\witrdel {\mathchoice
   {\mathclose{\setbox0=\hbox{$\displaystyle     \lrgldel$}\hbox to\wd0
                        {\hfil$\displaystyle     )       $\hfil}%
                        \kern-\wd0\hbox to\wd0
                        {\hss$\vcenter{\hbox{\eightpoint$\scriptscriptstyle\mkern3.2mu\bullet$}}$\hss}}}%
   {\mathclose{\setbox0=\hbox{$\textstyle        \lrgldel$}\hbox to\wd0
                        {\hfil$\textstyle        )       $\hfil}%
                        \kern-\wd0\hbox to\wd0
                        {\hss$\vcenter{\hbox{\eightpoint$\scriptscriptstyle\mkern3.2mu\bullet$}}$\hss}}}%
   {\mathclose{\setbox0=\hbox{$\scriptstyle      \lrgldel$}\hbox to\wd0
                        {\hfil$\scriptstyle      )       $\hfil}%
                        \kern-\wd0\hbox to\wd0
                        {\hss$\raise.1ex\hbox{\eightpoint$\scriptscriptstyle\mkern3.2mu\bullet$}$\hss}}}%
   {\mathclose{\setbox0=\hbox{$\scriptscriptstyle\lrgldel$}\hbox to\wd0
                        {\hfil$\scriptscriptstyle)       $\hfil}%
                        \kern-\wd0\hbox to\wd0
                        {\hss$\raise.03ex\hbox{\eightpoint$\scriptscriptstyle\mkern3.2mu\bullet$}$\hss}}}}%
\newbox\ldelbox
\setbox\ldelbox=\hbox{$\lrgldel$}

\newbox\rdelbox
\setbox\rdelbox=\hbox{$\lrgrdel$}

\def\pars #1{\parldel #1\parrdel}%
\def\cons #1{\conldel #1\conrdel}%
\def\quadfs {\rlap{\rm\quad.}}%
\def\qquato {\qquad\to\qquad}%
\def\qualto {\quad\leadsto\quad}%

\def\quand {\quad\mbox{and}\quad}%
\def\qquand {\qquad\mbox{and}\qquad}%
\def\qqquand {\quad\qquad\mbox{and}\qquad\quad}%
\def\qquor {\qquad\mbox{or}\qquad}%
\def\qquiff {\qquad\mbox{iff}\qquad}%
\def\clap#1{\hbox to 0pt{\hss#1\hss}}
\def\sqlap#1{\hbox to .5em{\hss#1\hss}}
\def\qlap#1{\hbox to 1em{\hss#1\hss}}
\def\qqlap#1{\hbox to 2em{\hss#1\hss}}
\def\qqqlap#1{\hbox to 3em{\hss#1\hss}}
\def\qqqqlap#1{\hbox to 4em{\hss#1\hss}}
\def\qqqqqlap#1{\hbox to 5em{\hss#1\hss}}
\def\qqqqqqlap#1{\hbox to 6em{\hss#1\hss}}
\def\qqqqqqqlap#1{\hbox to 7em{\hss#1\hss}}
\def\qqqqqqqqlap#1{\hbox to 8em{\hss#1\hss}}
\def\qqqqqqqqqlap#1{\hbox to 9em{\hss#1\hss}}
\newcommand{\wlap}[2][10ex]{\hbox to#1{\hss#2\hss}}

\def\clapm#1{\clap{$#1$}}

\newcommand{\wlapm}[2][10ex]{\hbox to#1{\hss$#2$\hss}}
\def\rlapm#1{\hbox to 0pt{$#1$\hss}}
\def\llapm#1{\hbox to 0pt{\hss$#1$}}

\def\qqquad{\quad\qquad}
\def\qqqquad{\qquad\qquad}
\def\qqqqquad{\qqquad\qquad}

\newcommand{\vclap}[2][0pt]{\hbox to #1{\hss#2\hss}}
\newcommand{\vclapm}[2][0pt]{\hbox to #1{\hss$#2$\hss}}

\def\proofadjust{\vadjust{\nobreak\vskip-2.7ex\nobreak}}

\def\interdisplayskip{.5ex}
\newskip\mydisplaywidth
\newcommand{\twolinedisplay}[3][10pt]{%
  \mydisplaywidth=\displaywidth
  \advance\mydisplaywidth-#1
  \begin{array}{c}
    \clap{\hbox to\mydisplaywidth{$\displaystyle#2$\hss}}\\[\interdisplayskip]
    \clap{\hbox to\mydisplaywidth{\hss$\displaystyle#3$}}
  \end{array}
}

\newdimen\mqdim
\vlnosmallleftlabels
\newcommand{\vlcin}[5]{\vliq{{\scriptstyle #1*}#2}{#3}{#4}{#5}}
\newcommand{\vlstr}[3]{\vltr{#1}{#2}{\vlshy{}}{#3}{\vlshy{}}}
\newcommand{\vlhtr}[2]{\vlstr{#1}{#2}{\vlshy{\hskip1.5em}}}
\newcommand{\vlshy}[1]{\vlhyaux{$#1$}}

\def\wbox{\square}
\def\wbox{\boxempty}

\def\wdia{\lozenge}

\def\implies{\mathbin{\supset}}
\def\cand{\mathbin{\wedge}}
\def\cor{\mathbin{\vee}}

\newcommand{\wbr}[1]{\pars{#1}}
\newcommand{\wbrn}[2]{\llbracket#1\rrbracket^{#2}}
\def\mpr{\mathsf{mp}}
\def\necr{\mathsf{nec}}
\def\boxnecr{\dotrule{nec}}

\newcommand{\Invr}{\mathsf{Inv}}
\renewcommand{\weakr}{\mathsf{w}}
\renewcommand{\conr}{\mathsf{c}}

\def\sysS{\mathsf{S}}

\def\xax{\mathsf{x}}
\def\yax{\mathsf{y}}
\def\Xax{\mathsf{X}}
\def\Yax{\mathsf{Y}}

\def\kax#1{\mathsf{k_{#1}}}
\def\kaxx{\mathsf{k}}
\def\dax{\mathsf{d}}
\def\tax{\mathsf{t}}
\def\bax{\mathsf{b}}
\def\vax{\mathsf{4}}
\def\svax{\mathsf{s4}}
\def\fax{\mathsf{5}}

\def\minibr{\hbox{$\scriptscriptstyle[\mkern2mu]$}}
\def\dotrule#1{\mathsf{#1}^{\minibr}}

\def\Xdot{\dotrule{X}}
\def\Ydot{\dotrule{Y}}
\def\sYdot{\dotrule{Y}_{\mathsf{s}}}

\def\ydot{\dotrule{y}}
\def\ddot{\dotrule{d}}
\def\tdot{\dotrule{t}}
\def\bdot{\dotrule{b}}
\def\sbdot{\dotrule{sb}}
\def\vdot{\dotrule{4}}

\def\fdot{\dotrule{5}}
\def\sfdot{\dotrule{s5}}
\def\sbfdot{\dotrule{sb5}}
\def\sfbdot{\dotrule{s5b}}

\def\boxmedr{\dotrule{m}}

\def\intrule#1{\mathsf{#1}\strut}
\def\dint{\intrule{d}}
\def\tint{\intrule{t}}
\def\bint{\intrule{b}}
\def\vint{\intrule{4}}

\newlength{\hhatheight}
\newlength{\hhhatheight}

\def\lmark{\bullet}
\def\rmark{\circ}

\def\lrmark{\raise.2ex\hbox{\tiny\LEFTcircle}}

\def\lef#1{#1^\lmark}
\def\rig#1{#1^\rmark}

\def\plef#1{#1^\bullet}
\def\prig#1{#1^\circ}

\def\rigrule#1{\mathsf{#1}^{\rmark_{\rlap{\phantom{x}}}}}
\def\lefrule#1{\mathsf{#1}^{\lmark_{\rlap{\phantom{x}}}}}

\def\Xlefrig{\mathsf{X}^{\lrmark}}
\def\sXlefrig{\mathsf{X}^{\lrmark}_{\mathsf{s}}}

\def\drig{\rigrule{d}}
\def\trig{\rigrule{t}}
\def\brig{\rigrule{b}}
\def\vrig{\rigrule{4}}
\def\vrigp{\rigrule{4'}}
\def\svrig{\rigrule{s4}}
\def\svdrig{\rigrule{s4}_\wdia}
\def\frig{\rigrule{5}}

\def\dlef{\lefrule{d}}
\def\tlef{\lefrule{t}}
\def\blef{\lefrule{b}}
\def\vlef{\lefrule{4}}
\def\vlefp{\lefrule{4'}}
\def\svlef{\lefrule{s4}}
\def\svblef{\lefrule{s4}_\wbox}
\def\flef{\lefrule{5}}

\def\wdiacutr{\wdia\mathsf{cut}}
\def\wboxcutr{\wbox\mathsf{cut}}
\def\Cutr{\mathsf{Cut}}

\newcommand{\Sfive}{\mathsf{S5}}
\newcommand{\Sfour}{\mathsf{S4}}

\newcommand{\Kfour}{\mathsf{K4}}

\newcommand{\CK}{\mathsf{CK}}
\newcommand{\CT}{\mathsf{CT}}
\newcommand{\NCK}{\mathsf{NCK}}
\newcommand{\NCKp}{\mathsf{NCK'}}
\newcommand{\HCK}{\mathsf{HCK}}
\newcommand{\CSfive}{\mathsf{CS5}}
\newcommand{\CSfour}{\mathsf{CS4}}
\newcommand{\CKfour}{\mathsf{CK4}}
\newcommand{\CKfourfive}{\mathsf{CK45}}
\newcommand{\CD}{\mathsf{CD}}
\newcommand{\CDfour}{\mathsf{CD4}}
\newcommand{\CDfourfive}{\mathsf{CD45}}

\newcommand{\CKfive}{\mathsf{CK5}} 
\newcommand{\CKBfive}{\mathsf{CKB5}} 
\newcommand{\CDfive}{\mathsf{CD5}} 
\newcommand{\CDB}{\mathsf{CDB}} 
\newcommand{\CTB}{\mathsf{CTB}}  
\newcommand{\CKB}{\mathsf{CKB}}

\newcommand{\Gamcon}[1]{\Gamma\cons{#1}}
\newcommand{\ddGamcon}[1]{\ddown\Gamma\cons{#1}}
\newcommand{\dGamcon}[1]{\down\Gamma\cons{#1}}

\newcommand{\conempty}{\cons{\emptyset}}

\newcommand{\ddGamma}{\ddown\Gamma}

\newcommand{\ddDelta}{\ddown\Delta}

\def\lhs#1{\ddown{#1}}

\newcommand{\thecons}[1]{#1}
\newcommand{\nocons}[1]{#1}
\newcommand{\Gcons}[1]{\Gamma\cons{#1}}

\newcommand{\Dcons}[1]{\Delta\cons{#1}}
\newcommand{\Scons}[1]{\Sigma\cons{#1}}
\newcommand{\Tcons}[1]{\Theta\cons{#1}}
\newcommand{\ddGcons}[1]{\ddown\Gamma\cons{#1}}

\newcommand{\ddDcons}[1]{\ddown\Delta\cons{#1}}
\newcommand{\ddScons}[1]{\ddown\Sigma\cons{#1}}
\newcommand{\ddTcons}[1]{\ddown\Theta\cons{#1}}
\newcommand{\lefGcons}[1]{\lef\Gamma\cons{#1}}

\newcommand{\lefDcons}[1]{\lef\Delta\cons{#1}}

\newcommand{\lefTcons}[1]{\lef\Theta\cons{#1}}

\newcommand{\colGcons}[1]{\thecons{\Gamma\cons{\nocons{#1}}}}
\newcommand{\colcons}[2]{\thecons{#1\cons{\nocons{#2}}}}

\newcommand{\vlinfG}[4]{%
  \vlinf{#1}{#2}{\colGcons{#3}}{\colGcons{#4}}}

\def\DD{\mathcal{D}}

\def\depth#1{\mathit{depth}(#1)}
\def\height#1{\mathit{h}(#1)}

\def\cv#1{\boldsymbol{v}(#1)}
\def\rcv#1{\boldsymbol{v}_{\mathbf{r}}(#1)}

\def\formula#1{\mathit{fm}(#1)}
\def\down#1{{#1^{\Downarrow}}}
\def\ddown#1{#1^{\Downarrow}}

\newbox\botbox
\setbox\botbox\hbox{$\bot$}
\newdimen\botwd
\botwd=\wd\botbox
\newcommand\downsymbol{\rlap{\copy\botbox}\hbox to\botwd{\hss$\mathord{\Downarrow}$\hss}}

\newbox\vdotbox
\setbox\vdotbox=\vbox{\hbox{$\vdots$}\vskip-1ex\hbox{$\vdots$}\vskip1ex}

\newcommand{\cutredcase}[3]{{#2\!\!\stackrel{(#1)}{\qualto}\!\!#3}}
\newcommand{\cutredcasea}[4][2ex]{%
\clapm{\begin{array}{c}
   #3\stackrel{(#2)}{\qualto}\hskip13em\\[#1]     
   \hskip13em #4\hss
\end{array}}}
\newcommand{\cutredcaseb}[4][1ex]{%
\clapm{\begin{array}{c}
   #3\stackrel{(#2)}{\qualto}\hskip3em\\[#1]     
   #4\hss
\end{array}}}
\newcommand{\cutredcasec}[4][2ex]{%
\begin{array}{c}
   #3\stackrel{(#2)}{\qualto}\hfill\\[#1]     
   \hfill #4
\end{array}}

\newcommand{\svlderivation}[1]{\hbox{\small$\vlderivation{#1}$}}
\begin{document}
\title[On Nested Sequents for Constructive Modal Logics]
        {On Nested Sequents for Constructive Modal Logics}

\author[R.~Arisaka]{Ryuta Arisaka}	%required
\address{INRIA, 1 rue Honor\'e d'Estienne d'Orves, \\
Campus de l'\'Ecole Polytechnique, B\^atiment Alan Turing\\
91120 Palaiseau, France}	%required
%\email{author1@email1}  %optional
%\thanks{Supported by ANR grant STRUCTURAL}	%optional

\author{Anupam Das}	%required
%\address{INRIA, 1 rue Honor\'e d'Estienne d'Orves, \\
%Campus de l'\'Ecole Polytechnique, B\^atiment Alan Turing\\
%91120 Palaiseau, France}	%required
%\email{author1@email1}  %optional

\author{Lutz Stra\ss burger}	%required
%\address{INRIA, 1 rue Honor\'e d'Estienne d'Orves, \\
%Campus de l'\'Ecole Polytechnique, B\^atiment Alan Turing\\
%91120 Palaiseau, France}	%required
%\email{http://www.lix.polytechnique.fr/Labo/Lutz.Strassburger/}  %optional

%% \def\specialheadingsodd{\it\today\quad---\quad Submitted \hfil}
%% \thispagestyle{specialheadings}
  
\begin{abstract}
  We present deductive systems for various modal logics that can be
  obtained from the constructive variant of the normal modal logic~$\CK$
  by adding combinations of the axioms $\dax$, $\tax$, $\bax$, $\vax$, and~$\fax$. This includes
  the constructive variants of the standard modal logics $\Kfour$, $\Sfour$,
  and~$\Sfive$. We use for our presentation the formalism of nested sequents
  and give a syntactic proof of cut elimination.
\end{abstract}

  %% \begin{keyword}
  %% \end{keyword}
%% \end{frontmatter}

\maketitle
%%%%%%%%%%%%%%%%%%%%%%%%%%%%%%%%%%%%%%%%%%%%%%%%%%%%%%%%%%%%%%%%%%%
%%%%%%%%%%%%%%%%%%%%%%%%%%%%%%%%%%%%%%%%%%%%%%%%%%%%%%%%%%%%%%%%%%%
%%%%%%%%%%%%%%%%%%%%%%%%%%%%%%%%%%%%%%%%%%%%%%%%%%%%%%%%%%%%%%%%%%%
%\notodo

\section{Introduction}        
The modal logic $\K$ is obtained from classical propositional logic by
incorporating two unary operators, or \emph{modalities}, $\wbox$ and $\wdia$,
and adding the $\kaxx$-axiom, $\wbox (A\supset B) \supset (\wbox A \supset \wbox B)$, to dictate the interaction
between the modalities and propositional connectives. The behavior of the $\wdia$ modality is then determined by enforcing that it is the De Morgan dual of $\wbox$. Along with this axiom there is the 
    \emph{necessitation} rule, saying that if $A$ is a theorem of~$\K$ then so is $\wbox A$. 
Informally, $\wbox$ is often interpreted as ``necessarily'' and $\wdia$ as ``possibly''. Notice that interaction with other propositional connectives is determined by the adequacy of $\{\supset, \bot \}$ in classical logic.

In the intuitionistic setting, however, one must define the behavior of $\wbox$ and $\wdia$ independently, in the absence of De Morgan duality. Consequently, it is not
enough to just add the standard $\kaxx$-axiom, which makes no mention of the $\wdia$-modality, and so some classical consequences of $\kaxx$ must be added to formulate an intuitionistic version of $\K$. To this end there seems to be no canonical choice, and many different intuitionistic versions of $\K$ have been proposed,
e.g.,~\cite{fitch:48,prawitz:65,fischer-servi:84,plotkin:stirling:86,simpson:phd,bierman:paiva:00,pfenning:davies:01}
(for a survey see~\cite{simpson:phd}). However, in the current literature, two variants prevail; the first, known as \emph{intuitionistic} $\K$, adds the
following five axioms, along with the necessitation rule, to intuitionistic propositional logic:
\begin{equation}
  \label{eq:ik}
    \begin{array}{r@{\;}l}
      \kax1\colon&\wbox(A\implies B)\implies(\wbox A\implies\wbox B)\\%[1ex]
      \kax2\colon&\wbox(A\implies B)\implies(\wdia A\implies\wdia B)\\%[1ex]
    \end{array}
    \qqquad
    \begin{array}{r@{\;}l}
      \kax3\colon&\wdia(A\cor B)\implies(\wdia A\cor\wdia B)\\%[1ex]
      \kax4\colon&(\wdia A\implies \wbox B)\implies\wbox(A\implies B)\\%x[1ex]
      \kax5\colon&\wdia\bot\implies\bot
    \end{array}
\end{equation}
It was originally proposed in~\cite{fischer-servi:84,plotkin:stirling:86} and studied in detail in~\cite{simpson:phd}; more recent work can be found in~\cite{galmiche:salhi:10,str:fossacs13,marin:str:aiml14}. 

The second variant, known as \emph{constructive} $\K$, includes only $\kax1$ and $\kax2$, not $\kax3, \kax4 , \kax5$. This choice of axioms dates back to \cite{prawitz:65}\footnote{We point out that some versions of $\K$ intermediate to these two variants have also been considered, for example the variant with  $\kax1,\kax2$, and $\kax5$ in \cite{wijesekera:90}.}, and its proof theory was investigated, for example, in
\cite{bierman:paiva:00,heilala:pientka:07,Mendler11}, while the semantics of it and some extensions was studied in \cite{Fairtlough97} and \cite{Kojima12}.

To gain intuition about the difference between the two variants, %let $\mathcal{P}$ denote propositional letters.
let us have a look at their standard Kripke semantics.
A model of intuitionistic modal logic is
described by a 4-tuple $(W, \le, \RInt, \mathsf{I})$
with
\begin{itemize}
    \item
a non-empty set of \emph{possible worlds} $W$, preordered by $\le$.
\item an \emph{accessibility relation} $\RInt \subseteq W \times W$ satisfying:
  \begin{enumerate}[label=(\roman{*}), leftmargin=*]
  \item\label{l:F1} For any  $w, v, v' \in W$, if $w\RInt v$ and
$v \le v'$, there exists a $w' \in W$ such that $w \le w'$ and
$w'\RInt v'$.
  \item\label{l:F2} For any $w, w', v \in W$, if
$w \le w'$ and $w\RInt v$, there exists a $v' \in W$ such that
$w'\RInt v'$ and $v \le v'$.
  \end{enumerate}
\item a function  $\mathsf{I}\colon W \rightarrow
    2^{\mathcal{A}}$, where $\mathcal{A}=\set{a,b,c,\ldots}$ denotes the set of propositional letters, such that for any
    $w, w' \in W$, if $w \le w'$ then $\mathsf{I}(w) \subseteq
    \mathsf{I}(w')$. 
\end{itemize}
Note that \ref{l:F1} and \ref{l:F2} ensure a form of monotonicity of $\RInt$
over $\le$. In contrast, a model of constructive modal
logic decouples the accessibility relation $R$ from $\le$. It assumes
a set of `fallible' worlds $\dot{\bot}$  as a subset of $W$; such that
$\dot{\bot}$ is closed under $\le$ and $R$, \emph{i.e.} whenever $w
\in \dot{\bot}$ and $wRw_1$ or $w \le w_1$ we also have $w_1 \in
\dot{\bot}$. This is much weaker a condition on $R$ than \ref{l:F1}
and \ref{l:F2}. Also the definition of the forcing relation $\models$
shows subtle differences. For the atoms, the binary connectives and
the $\wbox$-modality, the intuitionistic and constructive semantics definition coincide:
\begin{itemize}
\item $w \models a$ iff 
         $a \in \mathsf{I}(w)$. 
\item $w \models A \wedge B$ iff 
         $w \models A$ and $w \models B$. 
\item $w \models  A \vee B$ iff 
  $w \models A$ or $w \models B$. 
\item
  $w \models A \supset B$ iff
  $\forall w' \in W.\; w \le w'$ and $w' \models A$ imply
  $w' \models B$.
\item $w \models \Box A$ iff
  $\forall w',v' \in W.\;w \le w'$ and $w'Rv'$ imply
    $v' \models A$.
\end{itemize}
In the intuitionistic case, forcing for $\wdia$ and $\bot$ is defined as follows:
\begin{itemize}
\item $w \models \Diamond A$ iff
  $\exists v\in W.\;wRv$ and $v \models A$.
\item $w \not\models \bot$.  
\end{itemize}
Whereas in the constructive case we have:
\begin{itemize}
\item $w \models \Diamond A$ iff
  $\forall w'\in W.\;$ if $w \le w'$ then $\exists v'\in W.\; w'Rv'$ and $v' \models A$.
 \item $w \models \bot$ iff
     $w \in \dot{\bot}$.
 \item $w \models A$ if $w \in \dot{\bot}$.
\end{itemize}
We can now see that there is a countermodel to each
     of $\kax3,\kax4,\kax 5$ in the constructive setting:
     \begin{itemize}
         \item $\kax3$: $W = \{w_0, w_1, u_0, v_1\}, w_0 \le w_1,
             w_0Ru_0, w_1Rv_1, \mathsf{I}(w_0) =
                 \mathsf{I}(w_1) = \emptyset,\\
                 \mathsf{I}(u_0) = a_1,
                 \mathsf{I}(v_1) = a_2$. 
                 We have 
                 $w_0 \models \wdia 
                     (a_1 \vee a_2)$ and
                   $w_0 \not\models 
                      \wdia a_1 \vee \wdia a_2$.  
                  \item $\kax4$: $W = \{w_0, w_1, u_0, u_1\},
                      w_0 \le w_1, u_0 \le u_1, w_0Ru_0, 
                      \mathsf{I}(w_0) = \mathsf{I}(w_1) = \emptyset,
                      \mathsf{I}(u_0) = \{a_2\}, 
                      \mathsf{I}(u_1) = \{a_1, a_2\}$. 
                      We have $w_0 \models 
                      \wdia a_1 \supset \wbox a_2$ 
                      and $w_0 \not\models 
                      \wbox(a_1 \supset a_2)$.  
                 \item $\kax5$: $W = \{w_0,  u_0\},
                     w_0Ru_0, \mathsf{I}(w_0) = \emptyset,
                     u_0 \in \dot{\bot}$. 
                     We have 
                     $w_0 \models \wdia \bot$ 
                     and $w_0 \not\models \bot$. 
     \end{itemize}  
Notice that the countermodels for $\kax3$ and $\kax4$ could not exist
in the presence of  \ref{l:F1} and \ref{l:F2} above, and the
countermodel for $\kax5$ relies on the availability of the set
$\dot{\bot}$ of fallible worlds. 

We refer the reader
to~\cite{Mendler11} for a more thorough semantic analysis of the
differences between intuitionistic~$\K$ and constructive~$\K$.

This work is concerned with the proof theory of constructive $\K$, denoted $\CK$ and its various extensions with other common modal axioms. As for the classical and intuitionistic variants, we consider the five axioms below:
\begin{equation}
  \label{eq:ax}
  \begin{array}{r@{\;}l}
    \dax\colon&\wbox A\implies \wdia A\\%[\arrskip]
    \tax\colon&
    (A\implies\wdia A)\cand(\wbox A\implies A)\\%[\arrskip]
    \bax\colon&
    (A\implies\wbox\wdia A)\cand(\wdia\wbox A\implies A)\\%[\arrskip]
    \end{array}
    \qquad
    \begin{array}{r@{\;}l}
    \vax\colon&
    (\wdia\wdia A\implies\wdia A)\cand(\wbox A\implies\wbox\wbox A)\\%[\arrskip]
    \fax\colon&
    (\wdia A\implies\wbox\wdia A)\cand(\wdia\wbox A\implies\wbox A)
  \end{array}\quad
\end{equation}
\emph{A priori}, this gives us 32 different
logics, but as in classical (or intuitionistic) modal logic some of them  coincide, so that
we obtain only~15 distinct logics.\footnote{That there are at least 15 is inherited from the classical setting (or similarly the intuitionistic setting), and verifying that the classical equivalences hold is by inspection of the classical proofs.}  These are depicted in Figure~\ref{fig:cube-cons}, where we use the same names as those standard in the classical setting~\cite{garson:stanford}, prefixed by `$\mathsf{C}$'.

\begin{figure}[!t]
  %%\myfigskips
\[
  \xymatrix@R-3ex{ {} & *{\circ} \save
    []+<-2ex,+1.5ex>*\txt{\scriptsize \sf CS4}\restore \ar@{-}[rrrr]
    \ar@{-}[dd] \ar@{-}[ld] & {} & {} & {} & *{\circ}
    \save []+<1.5ex,+1.5ex>*\txt{\scriptsize \sf CS5}\restore \ar@{-}[ld] \ar@{-}[ddddd] \\
    *{\circ} \save []+<-1.5ex,+1.5ex>*\txt{\scriptsize \sf CT}\restore
    \ar@{-}[rrrr] \ar@{-}[ddd] & {} & {} & {} &
    *{\circ} \save []+<-1.5ex,+1.5ex>*\txt{\scriptsize \sf CTB}\restore \ar@{-}[ddd] & {} \\
    {} & *{\circ} \save []+<-2ex,+1.4ex>*\txt{\scriptsize \sf
      CD4}\restore \ar@{-}[ddd] \ar@{-}[rr] & {} & *{\circ} \save
    []+<2.5ex,-1.5ex>*\txt{\scriptsize \sf CD45}\restore
    \ar@{-}@/_/[uurr] & {} & {} \\
    {} & {} & *{\circ} \save []+<2ex,-1.5ex>*\txt{\scriptsize \sf
      CD5}\restore
    \ar@{-}[ur] & {} & {} & {} \\
    *{\circ} \save []+<-2ex,0ex>*\txt{\scriptsize \sf CD}\restore
    \ar@{-}[ddd] \ar@{-}[ruu] \ar@{-}[urr] \ar@{-}[rrrr]& {} & {} & {}
    & *{\circ} \save []+<2.5ex,0ex>*\txt{\scriptsize \sf CDB}\restore
    \ar@{-}[ddd] & {} \\
    {} & *{\circ} \save []+<-2ex,+1.5ex>*\txt{\scriptsize \sf
      CK4}\restore \ar@{-}[ldd] \ar@{-}[rr]& {} & *{\circ} \save
    []+<1.5ex,-1.5ex>*\txt{\scriptsize \sf CK45}\restore \ar@{-}[uuu]
    \ar@{-}[rr] & {} &
    *{\circ} \save []+<2.5ex,-1ex>*\txt{\scriptsize \sf CKB5}\restore \ar@{-}[ldd] \\
    {} & {} & *{\circ} \save []+<+1.5ex,-1.5ex>*\txt{\scriptsize \sf
      CK5}\restore \ar@{-}[ru] \ar@{-}[uuu] & {} & {} & {} \\
    *{\circ} \save []+<-1.5ex,-1.5ex>*\txt{\scriptsize \sf CK}\restore
    \ar@{-}[rrrr] \ar@{-}[rru] & {} & {} & {}
    & *{\circ} \save []+<1.5ex,-1.5ex>*\txt{\scriptsize \sf CKB}\restore & {} \\
  }
  \]
  \caption{The constructive ``modal cube'' %\todo{we should get CK, CK4, CK45, CD, CD4, CD45, CT, CS4, CS5}
  }
  \label{fig:cube-cons}
\end{figure}

While the proof theory of the intuitionistic version of this cube
has been well-studied in labeled systems~\cite{simpson:phd} and
non-labeled
systems~\cite{galmiche:salhi:10,str:fossacs13,marin:str:aiml14}, there
is surprisingly little work on the constructive ``modal cube''. In
fact, to our knowledge, only the logics $\CK$, $\CT$, $\CKfour$, and
$\CSfour$ have received proof theoretic treatment so far, e.g.\ in
\cite{bierman:paiva:00,heilala:pientka:07,Mendler11}.

In this work we attempt to give a unified cut-elimination procedure for all
logics in Figure~\ref{fig:cube-cons}, using the framework of
\emph{nested sequents}
\cite{kashima:nested,GorePT09,brunnler:deepseq,str:fossacs13,Fitting14},
a generalization of Gentzen's sequent calculus which allows sequents to
occur within sequents. This approach has previously been successful
for the classical modal cube in \cite{brunnler:deepseq} and the
intuitionistic modal cube in \cite{str:fossacs13} but, perhaps
surprisingly, the step from intuitionistic to constructive appears
more involved than the one from classical to intuitionistic. 

This is also the reason why, in this paper, we consider only the logics in
the `cube'. We would like to compare the intuitionistic and
constructive cases from the point of view of cut-elimination. Whenever possible,
we aim to point out the differences to the arguments for intuitionistic systems
presented in~\cite{str:fossacs13}.

While the cut-elimination proofs in \cite{brunnler:deepseq} and
\cite{str:fossacs13} are markedly similar, we seem to require a
different method in the constructive setting. The reasons for this are that
certain formulations of some logical rules are no longer sound, and
that we need an explicit contraction rule, along with other structural
rules that further complicate the process of cut-elimination.

Nonetheless we manage to obtain cut-elimination for the logics $\CK$,
$\CKfour$, $\CKfourfive$, $\CD$, $\CDfour$, $\CDfourfive$, $\CT$,
$\CSfour$, and $\CSfive$, and we conjecture that our systems admit cut
for all logics in the cube. 

We are not aware of a similar uniform treatment of constructive modal
logics within other formalisms. However, in hindsight it is
straightforward to translate our results into prefixed tableaux,
using~\cite{Fitting12}, or into a tree-labeled sequent calculus.

\medskip

We point out an interesting observation that the $\bax$-axiom entails $\kax3$ and $\kax5$. 
While this is likely already known to many in the community we could not find this result stated in the literature, and so it is pertinent to raise it here. This arguably questions the ``constructiveness'' of logics including $\bax$, and so the inclusion of such logics in the cube itself, but such considerations are beyond the scope of this work.\footnote{One might argue that this observation is the reason behind Prawitz' statement~\cite{prawitz:65} on $\Sfive$ being inherently non-constructive. However, Prawitz does not explicitly mention $\kax3$ and $\kax5$.}

Several attempts to deal with the proof theory of constructive modal
logic have appeared previously. However, the fundamental data
structures of such calculi all seem to be special cases of nested
sequents.  For example, the 2-sequents of
\cite{Masini92:2sequents:classical,Masini93:2sequents:intuitionistic}
are a form of nested sequent where no tree-branching is allowed.  It
is not clear how the 2-sequent approach, while successful for deontic
logic, could be adapted for the various constructive logics, or even
$\CK$, as pointed out by Wansing in \cite{Wansing94}.  Also the
sequents of \cite{Mendler11,Mendler14} can be seen as a special case
of nested sequents, also where no tree-branching is allowed, but constituting a
richer data structure than 2-sequents because of the inclusion of a `focus'.

    Regarding applications of the family 
of the constructive modal 
logics, the extended Curry-Howard correspondence 
(which, for modal logics, is a relatively 
recent investigation) has been studied for $\CSfour$ 
\cite{alechina:etal:01,Mendler14}. %% with categorical semantics.
The constructive {\small $\Box$} operator here captures staged computation 
\cite{Davies00,Zine99}, and such
logics are also used for the study of contexts 
\cite{Mendler05,Mendler14}. We also point out that there are many logics of interest that are proper extensions of $\CK$ but not of intuitionistic $\K$, e.g.\ $\CSfour$ and $\mathsf{PLL}$; a more detailed discussion of such logics can be found in~\cite{Fairtlough97}.

%%%%%%%%%%%%%%%%%%%%%%%%%%%%%%%%%%%%%%%%%%%%%%%%%%%%%%%%%%%%%%%%%%%%%%%%%%%%%
%%%%%%%%%%%%%%%%%%%%%%%%%%%%%%%%%%%%%%%%%%%%%%%%%%%%%%%%%%%%%%%%%%%%%%%%%%%%%
%%%%%%%%%%%%%%%%%%%%%%%%%%%%%%%%%%%%%%%%%%%%%%%%%%%%%%%%%%%%%%%%%%%%%%%%%%%%%

%%%%%%%%%%%%%%%%%%%%%%%%%%%%%%%%%%%%%%%%%%%%%%%%%%%%%%%%%%%%%%%%%%%%%%%%%%%%%
%%%%%%%%%%%%%%%%%%%%%%%%%%%%%%%%%%%%%%%%%%%%%%%%%%%%%%%%%%%%%%%%%%%%%%%%%%%%%
%%%%%%%%%%%%%%%%%%%%%%%%%%%%%%%%%%%%%%%%%%%%%%%%%%%%%%%%%%%%%%%%%%%%%%%%%%%%%

\section{Preliminaries on Nested Sequents} 
In order to present a nested sequent system for $\CK$, we first need
to define the notion of a nested sequent structure. For this, we
recall the basic notions from~\cite{str:fossacs13}, with slight modifications in notation, tailored to the current setting. Let
$a,b,c,\ldots$ denote propositional variables and define formulas $A,B,C,\ldots$ of
constructive modal logic 
by the following grammar:
\begin{equation*}
%  \label{modform-int}
  A\grammareq 
  a\mid\bot\mid(A\cand A)\mid(A\cor A)\mid(A\implies A)\mid
  \wbox A\mid\wdia A
\end{equation*}
As shorthand we write $\top$ for $\bot\implies\bot$ and we omit parentheses whenever it is not ambiguous. 

A (nested) sequent is a tree whose nodes are multisets of 
formulas tagged with a polarity. There are two polarities,
\emph{input} (intuitively as if on the left of the turnstile in the conventional sequent calculus), denoted by a $\lmark$ superscript, and \emph{output} (intuitively as if on the right of the turnstile in the conventional sequent calculus), denoted by a $\rmark$ superscript. 
Formally we define \emph{LHS sequents}, denoted $\Phi$, and \emph{RHS sequents}, denoted $\Psi$, as follows,
\looseness=-1
\begin{equation}
  \label{eq:ns-cml}
    \Phi\grammareq \emptyset\mid \lef A \mid \wbr{\Phi} \mid \Phi,\Phi
      \qqqquad
  \Psi\grammareq \rig A\mid\wbr{\Phi,\Psi}
  \quad
\end{equation}
and a \emph{full sequent} is a structure of the form $\Phi,\Psi$. We assume that associativity and commutativity of the comma `,' is implicit in our systems, and that $\emptyset$ acts as its unit.

This definition entails that exactly one formula in a full sequent has output
polarity, and all others have input polarity.  
We use capital Greek letters $\Gamma$,
$\Delta, \Sigma, \ldots$ to denote arbitrary sequents, LHS, RHS or full, and may decorate them with a $\lmark$ or $\rmark$ superscript to indicate that they are LHS or RHS, respectively. 

The \emph{corresponding formula} of a sequent is defined inductively as follows,
\begin{align*}
  & \formula{\lef A}= \formula{\rig A} = A \ , \    \formula{\wbr{\Phi}}= \wdia\formula{\Phi} \ , \ 
    \formula{\wbr{\Phi,\Psi}} = \wbox(\formula{\Phi,\Psi}) \\
  & \formula{\emptyset} = \top \ ,\   \formula{\Phi_1 , \Phi_2} = \formula{\Phi_1}\cand\formula{\Phi_2}  \ , \ 
\formula{\Phi,\Psi} = \formula{\Phi}\implies\formula{\Psi}
\end{align*} 
A \emph{context}, denoted
by $\Gamma\conhole$, is a sequent with a hole $\conhole$ taking the place of a subsequent (or, equivalently, a formula);  $\colGcons{\Delta}$ is the sequent obtained
from $\Gamma\conhole$ by replacing the occurrence of $\conhole$ by
$\Delta$. Note that, for this to form a full sequent, $\Gamma\conhole$ and $\Delta$ must have the correct format. 
We distinguish two
kinds of contexts: an \emph{output context} is one that results in a full sequent when its hole is filled with a RHS sequent, and an \emph{input context} analogously for a LHS sequent. This is clarified by the following example, taken
from~\cite{str:fossacs13}.

\begin{example}\label{exa:conhole}
  Let $\Gamma_1\conhole=\lef C,\wbr{\conhole,\wbr{\lef B,\lef C}}$ and
  $\Gamma_2\conhole=\lef C,\wbr{\conhole,\wbr{\lef B,\rig C}}$.
  %%Then $\depth{\Gamma_1\conhole}=\depth{\Gamma_2\conhole}=1$. 
  Now let
  $\Delta_1=\lef A,\wbr{\rig B}$ and $\Delta_2=\lef A,\wbr{\lef B}$.
  Then $\Gamma_1\cons{\Delta_2}$ and $\Gamma_2\cons{\Delta_1}$ are not
  well-formed full sequents, because the former would contain no output
  formula, and the latter would contain two. However, we can form the full sequents,
  \begin{equation*}
    \Gamma_1\cons{\Delta_1}=
    \lef C,\wbr{{\lef A,\wbr{\rig B}},\wbr{\lef B,\lef C}}
    \quand
    \Gamma_2\cons{\Delta_2}=
    \lef C,\wbr{{\lef A,\wbr{\lef B}},\wbr{\lef B,\rig C}}
  \end{equation*}
  whose corresponding formulas, respectively, are:
  \begin{equation*}
    C\implies\wbox(A\cand\wdia(B\cand C)\implies\wbox B)
    \quand
    C\implies\wbox(A\cand\wdia B \implies\wbox(B\implies C))
  \end{equation*}
\end{example}

\begin{observation}\label{obs:context}
  Every output context $\Gamma\conhole$ is of the shape,
  \begin{equation}
    \label{eq:wcontext}
    \lef\Gamma_1,\wbr{\lef\Gamma_2,\wbr{\ldots,\wbr{\lef\Gamma_n,\conhole}\ldots}}
  \end{equation}
  for some $n\ge 0$.  Filling
  the hole of an output context with a RHS or full sequent yields a full
  sequent, and filling it with a LHS sequent yields a
  LHS sequent. 
  Every input context $\Gamma\conhole$ is of the shape,
  \begin{equation}
\Gamma'\cons{\Lambda\conhole,\rig\Pi}
  \end{equation}
  where $\Gamma'\conhole$ and
  $\Lambda\conhole$ are output contexts (i.e., are of the
  shape~\eqref{eq:wcontext} above). Note that $\Gamma'\conhole$ and $\Lambda\conhole$
  and $\Pi$ are uniquely determined by the position of the hole
  $\conhole$ in $\Gamma\conhole$. 
\end{observation}
\vskip.5ex %pexarskip

We can choose to fill the hole of a context $\Gamma\conhole$ with nothing, denoted by
$\Gamma\conempty$,
which means we simply remove the occurrence of $\conhole$. In Example~\ref{exa:conhole} above, $\Gamma_1\conempty=\lef
C,\wbr{\wbr{\lef B,\lef C}}$ is a LHS sequent and $\Gamma_2\conempty=\lef
C,\wbr{\wbr{\lef B,\rig C}}$ is a full sequent. More generally, whenever
$\Gamma\conempty$ is a full sequent, then $\Gamma\conhole$ is an input
context.

\begin{definition}\label{def:pruning}
  For every input context
  $\Gamma\conhole=\Gamma'\cons{\Lambda\conhole,\rig\Pi}$, we define its
  \emph{output pruning} $\ddown{\Gamma}\conhole$ to be the context
  $\Gamma'\cons{\Lambda\conhole}$, i.e., the same context with the
  subtree containing the unique output formula and sharing the same root as $\conhole$ removed. Thus,
  $\ddGcons{\enspace}$ is an output context. 
  
  If $\Gamma\conhole$
  is already an output context then
  $\ddGcons{\enspace}=\Gamma\conhole$.
  For every full sequent
  $\Delta=\lef\Lambda,\rig\Pi$, we define $\ddDelta$ to be its LHS-sequent
  $\Lambda$. For a LHS sequent $\Delta$, we define $\ddDelta=\Delta$.
  
\end{definition}

In Example~\ref{exa:conhole} above, $\ddown{\Gamma_1}\conhole=\lef
C,\wbr{\conhole,\wbr{\lef B,\lef C}}$ and $\ddown{\Gamma_2}\conhole=\lef
C,\wbr{\conhole}$, whereas $\ddown{(\Gamma_1\conempty)}=\lef
C,\wbr{\wbr{\lef B,\lef C}}$ and $\ddown{(\Gamma_2\conempty)}=\lef C$. In particular,  $\ddown{\Gamma_1}\cons{\rig A} = \lef C , \wbr{\rig A , \wbr{ \lef B, \lef C }}$, which is not the same as $\ddown{(\Gamma_1 \cons{ \rig A })} = \lef C$.

%%%%%%%%%%%%%%%%%%%%%%%%%%%%%%%%%%%%%%%%%%%%%%%%%%%%%%%%%%%%%%%%%%%%%%%%%%%%%%%%%%%%%%%%%%%%
%%%%%%%%%%%%%%%%%%%%%%%%%%%%%%%%%%%%%%%%%%%%%%%%%%%%%%%%%%%%%%%%%%%%%%%%%%%%%%%%%%%%%%%%%%%%
%%%%%%%%%%%%%%%%%%%%%%%%%%%%%%%%%%%%%%%%%%%%%%%%%%%%%%%%%%%%%%%%%%%%%%%%%%%%%%%%%%%%%%%%%%%%

\begin{figure}[t!]
  %%\myfigskips
  \def\myskip{-1.5ex}
  \begin{center}%\normalsize
    \begin{tabular}[t]{c@{\qqquad}c@{\qqquad}c}
      $\vlinf{\lef\bot}{}{\colGcons{\lef\bot,\rig\Pi}}{}$ 
      &
      &
      $\vlinf{\idr}{}{\colGcons{\lef a,\rig a}}{}$
      \\ \\[\myskip]
      $\vlinfG{\lef\cand}{}{{\plef{A\cand B}}}{{\lef A,\lef B}}$ 
      &&
      $\vliiinf{\rig\cand}{}{ \colGcons{\prig{A\cand B}}}{ \colGcons{\rig
          A}}{}{\colGcons{\rig B}}$ 
      \\ \\[\myskip] 
      $\vliiinf{\lef\cor}{}{
        \colGcons{\plef{A\cor B},\rig\Pi}}{ \colGcons{\lef A,\rig\Pi}}{}{\colGcons{\lef B, \rig\Pi}}$ 
      &&
      $\vlinfG{\rig\cor}{}{\prig{A\cor B}}{\rig A}$ \quad
      $\vlinfG{\rig\cor}{}{\prig{A\cor B}}{\rig B}$
      \\ \\[\myskip]
      $\vliiinf{\lef\implies}{}{ \colGcons{\plef{A\implies B}}}{
        \colcons{\ddown\Gamma}{\rig A}}{}{\colGcons{\lef
          B}}$ 
      %% $\vliiinf{\lef\implies}{}{ \colGcons{\plef{A\implies B},\rig C}}{
      %%   \colcons{\Gamma}{\plef{A\implies B},\rig A}}{}{\colGcons{\lef
      %%     B,\rig C}}$ 
      && 
      $\vlinfG{\rig\implies}{}{{\prig{A\implies B}}}{{\lef A,\rig
          B}}$ 
      \\ \\[\myskip] 
      $\vlinf{\lef\wbox}{}{ \colGcons{\plef{\wbox
            A},\wbr{\Delta}}}{ \colGcons{\wbr{\lef A,\Delta}}}$
      && 
      $\vlinf{\rig\wbox}{}{\colGcons{\prig{\wbox A}}}{\colGcons{\wbr{\rig
            A}}}$ 
      \\ \\[\myskip] 
      $\vlinf{\lef\wdia}{}{\colGcons{\plef{\wdia
            A}}}{\colGcons{\wbr{\lef A}}}$ 
      &
      \qlap{$\vlinf{\conr}{}{\Gcons{\lef\Delta}}{\Gcons{\lef\Delta,\lef\Delta}}$}
      & 
      $\vlinf{\rig\wdia}{}{
        \colGcons{\prig{\wdia A},\wbr{\Delta}}}{ \colGcons{\wbr{\rig
            A,\Delta}}}$
    %  \\[\myskip]
    \end{tabular}%
    \caption{System $\NCK$}
    \label{fig:NCK}
  \end{center}
\end{figure}

%%%%%%%%%%%%%%%%%%%%%%%%%%%%%%%%%%%%%%%%%%%%%%%%%%%%%%%%%%%%%%%%%%%%%%
%%%%%%%%%%%%%%%%%%%%%%%%%%%%%%%%%%%%%%%%%%%%%%%%%%%%%%%%%%%%%%%%%%%%%%
%%%%%%%%%%%%%%%%%%%%%%%%%%%%%%%%%%%%%%%%%%%%%%%%%%%%%%%%%%%%%%%%%%%%%%

\section{Nested Sequent Systems for $\CK$ and its Variants}

We use the standard notions of \emph{inference rule} and
\emph{derivation} (or \emph{proof}) from usual sequent
calculi; all that changes is the notion of sequent, as
introduced in the previous section. We insist that every sequent in a
derivation is a full sequent.\footnote{In fact
all the inference rules that we discuss in this paper, are such
that it is enough to demand that the conclusion of a derivation is a
full sequent. It then will follow that every sequent occurring in the
derivation is full.} A proof of a formula $A$ is then a
derivation whose conclusion is the (full) sequent $\rig A$. 
We also use the standard notions of
\emph{admissibility} and \emph{derivability} of inference rules (see,
e.g., \cite{buss:98} or \cite{troelstra:schwichtenberg:00}).

Let us now consider the set of inference rules shown in Figure~\ref{fig:NCK}, which we call the system $\NCK$ for $\CK$. These rules are similar to the corresponding rules for intuitionistic modal logic
in~\cite{str:fossacs13} and classical modal logic in~\cite{brunnler:deepseq}, although there are some subtle yet crucial
differences:
\begin{itemize}
\item In~\cite{str:fossacs13} and~\cite{brunnler:deepseq} additive versions of $\lef\implies$ and $\lef\wbox$ were given rather than incorporating an explicit contraction rule in the system. While these were essentially design choices in the previous works, here it is necessary to make contraction explicit since our
  treatment of the $\bax$-axiom does not allow us to show the admissibility of contraction; this is explained further below. Consequently, our cut-elimination proof differs significantly from the ones in \cite{str:fossacs13}
  and~\cite{brunnler:deepseq}.
\item The $\lef\bot$-rule and the $\lef\cor$-rule have a
  restriction on where the output formula occurs in the context: it
  must be in the same subtree of the sequent as the principal formula
  of the rule.  
  The reason for this is the lack of $\kax3$ (for the
  $\lef\cor$-rule) and $\kax5$ (for the $\lef\bot$-rule).
\item In the $\lef\implies$-rule (and also in the $\cutr$-rule
  described below), the `output pruning' is defined
  differently from \cite{str:fossacs13}. There only the unique output formula
  is removed, whereas here the whole subtree containing the output
  formula is removed. The reason for this is the lack of the
  $\kax4$-axiom.
  \item  In~\cite{str:fossacs13} the structural rule
    {\small${\vlinfG{\boxmedr}{}{{\wbr{\Delta_1,\Delta_2}}}{{\wbr{\Delta_1},\wbr{\Delta_2}}}}$}
    is heavily used. However, in the constructive setting, this rule is not available as it is no longer sound: it corresponds to the
    $\kax4$-axiom when the output formula occurs in $\Delta_1$ or~$\Delta_2$.
\end{itemize}
Note that the $\idr$-rule applies only to atomic formulas but, as usual with
sequent-style systems, the general form is derivable and this can be shown by a straightforward induction:

\begin{proposition}
  The rule \hbox{\small$~~\smash{\vlinf{\idr}{}{\colGcons{\lef A,\rig A}}{}}~~$} is
  derivable in\/~$\NCK$.
\end{proposition}

In the course of this paper we make use of the following structural rules:
\begin{equation}
  \label{eq:cutr}
  \vlinf{\boxnecr}{}{\wbr{\Gamma}}{\Gamma}
  \qqquad
  \vlinfG{\weakr}{}{\lef\Delta}{\emptyset}
  \qqquad
  \vliiinf{\cutr}{}{
    \colGcons{\emptyset}}{
    \colcons{\ddGamma}{\rig A}}{}{\colGcons{\lef A}}
\end{equation}
called \emph{necessitation}, \emph{weakening}, 
and \emph{cut}, respectively. These rules are not part of
the system, but we will later see that they are all admissible. Note that in
the weakening rule $\Delta$ must be a LHS sequent, as is the case for the contraction
rule $\conr$, as one might expect in an intuitionistic setting. 
The cut rule makes use of the output
pruning in the same way as the $\lef\implies$-rule.

We now turn to the rules for the axioms in~\eqref{eq:ax}. For
$\dax$, $\tax$ and $\vax$, the corresponding rules are shown in
Figure~\ref{fig:dt4}, and they coincide with those in~\cite{str:fossacs13}.

\begin{figure}[t!]
  %%\myfigskips
  \def\myskip{-1.5ex}
  \begin{center}%\normalsize
    \begin{tabular}[t]{c@{\qquad}c@{\qquad}c@{\qquad}c@{\qquad}c}
      $\vlinfG{\rig\dint}{}{\prig{\wdia A}}{\wbr{\rig A}}$
      &&
      $\vlinfG{\rig\tint}{}{\prig{\wdia A}}{\rig A}$
      &&
      $\vlinf{\rig\vint}{}{\colGcons{\prig{\wdia A},\wbr{\Delta}}}{
        \colGcons{\wbr{\prig{\wdia A},\Delta}}}$
      \\ \\%[\myskip]
      $\vlinfG{\lef\dint}{}{\plef{\wbox A}}{\wbr{\lef A}}$
      &&
      $\vlinfG{\lef\tint}{}{\plef{\wbox A}}{\lef A}$
      &&
      $\vlinf{\lef\vint}{}{\colGcons{\plef{\wbox A},\wbr{\Delta}}}{
        \colGcons{\wbr{\plef{\wbox A},\Delta}}}$
  \end{tabular}
    \caption{Constructive $\rig\wdia$- and $\lef\wbox$-rules for the axioms
      $\dax$, $\tax$, and $\vax$.}
    \label{fig:dt4}
 \end{center}
\end{figure}

For the $\bax$ and $\fax$ axioms, the rules given
in~\cite{str:fossacs13} (themselves adapted from the classical
setting \cite{brunnler:deepseq}) are not sound in the constructive
setting, again due to the lack of~$\kax4$. For $\bax$, one could
restrict the rules of~\cite{str:fossacs13} in the following way,
\begin{equation}
  \label{eq:brig}
  \vlinf{\rig\bint_{\mathit{int}}}{}{
    \colGcons{\wbr{\Delta,\prig{\wdia A}}}}{
    \colGcons{\wbr{\Delta},\rig A}}
  \qualto
  \vlinf{\rig\bint_{\mathit{con}}}{}{
    \colGcons{\wbr{\lef\Delta,\prig{\wdia A}}}}{
    \colGcons{\rig A}}
\end{equation}
%and
\begin{equation}
  \label{eq:blef}
  \vlinf{\lef\bint_{\mathit{int}}}{}{
    \colGcons{\wbr{\Delta,\plef{\wbox A}}}}{
    \colGcons{\wbr{\Delta},\lef A}}
  \qualto
  \vlinf{\lef\bint_{\mathit{con}}}{}{
    \colGcons{\wbr{\lef\Delta,\plef{\wbox A}}}}{
    \colGcons{\lef A}}
\end{equation}
in order to regain soundness. However such a system is not yet complete as, for example, the formula $\wdia(\wbox A\cor
\bot)\implies A$ is no longer provable in the cut-free system.   

To address this problem, we introduce the structural rules in Figure~\ref{fig:CXstr} which were
used during the cut-elimination proofs of \cite{brunnler:deepseq} and~\cite{str:fossacs13}. These rules are identical to the ones in
\cite{brunnler:deepseq} and~\cite{str:fossacs13} for $\dax$, $\tax$,
and $\bax$. For $\vax$, our rule is slightly weaker than the one
in~\cite{str:fossacs13}, again due to the lack of $\kax4$. 
Finally, for $\fax$, the situation is more subtle: again, the general
versions of the logical rules $\flef$ and $\frig$
from~\cite{str:fossacs13} are no longer sound due to the lack of
$\kax4$. These $\flef$ and $\frig$ rules can each be
decomposed into three rules performing `simpler' inference steps, but unfortunately all three of these are unsound. The
first can be made sound by incorporating weakening, as shown for
$\brig$ and~$\blef$ in \eqref{eq:brig} and~\eqref{eq:blef} above, but, as expected, the resulting system is again incomplete.

Perhaps surprisingly, the structural rule $\fdot$ used in
\cite{brunnler:deepseq} and~\cite{str:fossacs13} is also no longer
sound in the constructive setting due to the lack of $\kax4$.
However, that rule (shown on the left below\footnote{In this rule the
  depth of $\Gamma\conhole\cons{\emptyset}$ has to be $>0$. For
  further details we refer the reader to
  \cite{brunnler:deepseq,str:fossacs13,marin:str:aiml14}.}) can also be decomposed into three rules (shown on the right below), of which the
first (shown in Figure~\ref{fig:CXstr}) is sound in the constructive
setting, i.e.\ with respect to $\HCK+\fax$ in the next section. This `decomposition' is
similar to the cases of the rules $\lef{\fax}$ and $\rig{\fax}$
discussed in \cite{str:fossacs13} and~\cite{marin:str:aiml14}.
$$
\vlinf{}{}{\Gamma\cons{\emptyset}\cons{\wbr{\Sigma}}}{\Gamma\cons{\wbr{\Sigma}}\cons{\emptyset}}
\quad\;\equiv\;\quad
\vlinf{}{}{\colGcons{\wbr{\Sigma},\wbr{\Delta}}}{\colGcons{\wbr{\wbr{\Sigma},\Delta}}}
\quad\!+\!\quad
\vlinf{}{}{\colGcons{\wbr{\Delta},\wbr{\wbr{\Sigma},\Theta}}}{\colGcons{\wbr{\Delta,\wbr{\Sigma}},\wbr{\Theta}}}
\quad\!+\!\quad
\vlinf{}{}{\colGcons{\wbr{\Delta,\wbr{\wbr{\Sigma},\Theta}}}}{\colGcons{\wbr{\Delta,\wbr{\Sigma},\wbr{\Theta}}}}
$$

\begin{figure}[t!]
  %%\myfigskips
  \begin{center}
    \begin{tabular}[t]{c@{\qquad}c@{\qquad}c@{\qquad}c@{\qquad}c}
      $\vlinfG{\ddot}{}{\emptyset}{\wbr{\emptyset}}$
      &
      $\vlinfG{\tdot}{}{\Sigma}{\wbr{\Sigma}}$
      &
      $\vlinf{\bdot}{}{
        \colGcons{\Sigma,\wbr{\Delta}}}{
        \colGcons{\wbr{\wbr{\Sigma},\Delta}}}$
      &
      $\vlinf{\vdot}{}{\colGcons{\wbr{\wbr{\Sigma}}}}{
        \colGcons{\wbr{\Sigma}}}$
      &
      $\vlinf{\fdot}{}{
        \colGcons{\wbr{\Sigma},\wbr{\Delta}}}{
        \colGcons{\wbr{\wbr{\Sigma},\Delta}}}$
  \end{tabular}
    \caption{Structural rules for the axioms
      $\dax$, $\tax$, $\bax$, $\vax$, and $\fax$} 
    \label{fig:CXstr}
  \end{center}
\end{figure}

In the remainder of this paper we show soundness and completeness of our systems. For this let us introduce 
the following notation. We use $\Xax$ and $\Yax$ for sets of axioms, i.e.,
$\Xax, \Yax \subseteq \set {\dax, \tax, \bax, \vax, \fax}$, and we 
write $\Xdot$ (or $\Ydot$) to
denote the set of corresponding structural rules shown in
Figure~\ref{fig:CXstr}. If $\Xax\subseteq\set{\dax,\tax,\vax}$, we
write $\Xlefrig$ for the set of corresponding $\lef\wbox$- and $\rig\wdia$-rules shown in
Figure~\ref{fig:dt4}.
Then, we may write $\NCK+\Xlefrig+\Ydot$ to denote $\NCK$ augmented with the rules $\Xlefrig$ and $\Ydot$; in such cases no assumptions on $\Xax$ or $\Yax$ further to those stated are assumed. In particular, their intersection does not need to be empty, nor does one need to be a subset of the other.

%%%%%%%%%%%%%%%%%%%%%%%%%%%%%%%%%%%%%%%%%%%%%%%%%%%%%%%%%%%%%%%%%%%%%%%%
%%%%%%%%%%%%%%%%%%%%%%%%%%%%%%%%%%%%%%%%%%%%%%%%%%%%%%%%%%%%%%%%%%%%%%%%
%%%%%%%%%%%%%%%%%%%%%%%%%%%%%%%%%%%%%%%%%%%%%%%%%%%%%%%%%%%%%%%%%%%%%%%%

\section{Soundness}\label{sec:soundness}

To our knowledge there are no standard Kripke semantics for all the
various constructive modal logics and consideration of this issue is
beyond the scope of this work. Therefore we show soundness of our
rules with respect to the Hilbert system.

For this we define $\HCK$ to be
some complete set of axioms for intuitionistic propositional logic extended by
the axioms $\kax1$ and $\kax2$, shown in~\eqref{eq:ik}, together with
the rules $\mpr$ for \emph{modus ponens} and $\necr$ for \emph{necessitation}:
\begin{equation}
  \label{eq:mp-nec}
  \vliinf{\mpr}{}{B}{A}{A\implies B}
  \qqqqquad
  \vlinf{\necr}{}{\wbox A}{A}
\end{equation}
For a set $\Xax \subseteq \set {\dax, \tax, \bax, \vax, \fax}$ we then
write $\HCK+\Xax$ for the system obtained from $\HCK$ by adding the
axioms in $\Xax$. If $\Xax$ is a singleton $\set{\xax}$, we just write
$\HCK+\xax$.
Soundness can now be stated in the following theorem:
\begin{theorem}[Soundness]\label{thm:sound-con}
  Let $\Xax\subseteq\set{\dax,\tax,\vax}$, let $\Yax\subseteq\set{\dax,\tax,\bax,\vax,\fax}$, and let~
%  \begin{equation*}
    {\small$\upsmash{\vliiinf{\rr}{}{\Gamma}{\Gamma_1}{\ldots}{\Gamma_n}}$}
%  \end{equation*}
  {\rm(}for $n\in\set{0,1,2}${\rm)}
  be an instance of a rule in $\NCK+\weakr+\cutr+\Xlefrig+\Ydot$. Then:
  \begin{enumerate}[label=(\roman{*}), leftmargin=*] %[(i)]
  \item the formula
    $\formula{\Gamma_1}\cand\cdots\cand\formula{\Gamma_n}\implies\formula{\Gamma}$
    is provable in $\HCK+\Xax+\Yax$, and
  \item whenever a sequent $\Gamma$ is provable in $\NCK+\weakr+\cutr+\Xlefrig+\Ydot$, then
    $\formula\Gamma$ is provable in $\HCK+\Xax+\Yax$.
  \end{enumerate}
\end{theorem}

Clearly, (ii) follows immediately from (i) using an induction on the size of the derivation.
To prove (i), we start with the axioms:

\begin{lemma}\label{lem:0-rule-sound}
  Let $\Xax\subseteq\set{\dax,\tax,\bax,\vax,\fax}$, let $\Gamma\conhole$
  be an output context, and $\rig\Pi$ be an RHS-sequent. Then
  $\formula{\Gamma\cons{\lef a,\rig a}}$ and
  $\formula{\Gamma\cons{\lef\bot,\rig\Pi}}$ are provable in
  $\HCK+\Xax$.
\end{lemma}

\begin{proof}
By induction on the structure of $\Gamma\conhole$.
  %%\qed
\end{proof}

For showing soundness of the inference rules with one premise, 
we first have to verify that the deep inference reasoning remains valid in
the constructive setting. 
This is shown in the following three lemmas.

\begin{lemma}\label{lem:pre-context}
  Let $\Xax\subseteq\set{\dax,\tax,\bax,\vax,\fax}$, and let $A$, $B$,
  and $C$ be formulas.
  \begin{enumerate}[label=(\roman{*}), leftmargin=*] 
  \item\label{l:impl} If $A\implies B$ is provable in $\HCK+\Xax$, then so is
    $(C\implies A)\implies(C\implies B)$.
 \item\label{l:impl-inv} If $A\implies B$ is provable in $\HCK+\Xax$, then so is
    $(B\implies C)\implies(A\implies C)$.
  \item\label{l:and} If $A\implies B$ is provable in $\HCK+\Xax$, then so is 
    $(C\cand A)\implies(C\cand B)$.
    \item\label{l:box} If $A\implies B$ is provable in $\HCK+\Xax$, then so is
            $\wbox A\implies\wbox B$.
     \item\label{l:dia} If $A\implies B$ is provable in $\HCK+\Xax$, then so is
        $\wdia A\implies\wdia B$.

  \end{enumerate}
\end{lemma}

\begin{proof}
\ref{l:impl}, \ref{l:impl-inv} and \ref{l:and} follow by completeness of $\HCK$ over intuitionistic logic. \ref{l:box} and \ref{l:dia} follow by necessitation and $\kax1$ or $\kax2$, respectively.
  %%\qed
\end{proof}

\begin{lemma}\label{lem:DI-sound}
  Let $\Xax\subseteq\set{\dax,\tax,\bax,\vax,\fax}$, let $\Delta$ and
  $\Sigma$ be full sequents, and let $\Gamma\conhole$ be an output context.
  If $\formula{\Delta}\implies\formula\Sigma$ is  provable in $\HCK+\Xax$, then so is 
  $\formula{\Gamma\cons{\Delta}}\implies\formula{\Gamma\cons{\Sigma}}$.
\end{lemma}

\begin{proof}%%[of Lemma~\ref{lem:DI-sound}]
  Induction on the structure of $\Gamma\conhole$ (see Observation~\ref{obs:context}), using
  Lemma~\ref{lem:pre-context}.\ref{l:impl} and~\ref{l:box}.
  %%\qed
\end{proof}

\begin{lemma}\label{lem:DI-sound-inv}
  Let $\Xax\subseteq\set{\dax,\tax,\bax,\vax,\fax}$, let $\Delta$ and
  $\Sigma$ be LHS-sequents, and $\Gamma\conhole$ an input context.
  If $\formula\Sigma\implies\formula{\Delta}$ is  provable in $\HCK+\Xax$, then so is 
  $\formula{\Gamma\cons{\Delta}}\implies\formula{\Gamma\cons{\Sigma}}$.
\end{lemma}

\begin{proof} %%[of Lemma~\ref{lem:DI-sound-inv}]
  As in~\cite{str:fossacs13}. By Observation~\ref{obs:context},
  $\Gamma\conhole=\Gamma'\cons{\Lambda\conhole,\Pi}$ for some
  $\Gamma'\conhole$ and $\Lambda\conhole$ and $\Pi$.  By induction on
  $\Lambda\conhole$, using Lemma~\ref{lem:pre-context}.\ref{l:and}
  and~\ref{l:dia}, we get
  $\formula{\Lambda\cons{\Sigma}}\implies\formula{\Lambda\cons{\Delta}}$,
  and from Lemma~\ref{lem:pre-context}.\ref{l:impl-inv} it then
  follows that $(\formula{\Lambda\cons{\Delta}}\implies\formula{\Pi})
  \implies(\formula{\Lambda\cons{\Sigma}}\implies\formula{\Pi})$.  Now
  the statement follows from Lemma~\ref{lem:DI-sound}.  %%\qed
\end{proof}

We can now prove the soundness of rules with one premiss.

\begin{lemma}\label{lem:1-rule-sound}
  Let $\Xax\subseteq\set{\dax,\tax,\bax,\vax,\fax}$, and let~ 
%  \begin{equation*}
    {\small$\upsmash{\vlinf{\rr}{}{\Gamma_2}{\Gamma_1}}$}
%  \end{equation*}
  be an instance of $\weakr$, $\conr$, $\rig\cor$, $\rig\wbox$,
  $\rig\wdia$, $\rig\implies$, $\lef\cand$, $\lef\wdia$, or\/ $\lef\wbox$.
  Then $\formula{\Gamma_1}\implies\formula{\Gamma_2}$ is
  provable in $\HCK+\Xax$.
\end{lemma}

\begin{proof}
  For the rules $\rig\cor$, $\rig\wbox$, $\rig\wdia$, $\rig\implies$
  this follows immediately from Lemma~\ref{lem:DI-sound}, where for
  $\rig\wdia$ we need the $\kax2$-axiom. For the other rules we apply
  Lemma~\ref{lem:DI-sound-inv}.  Note that for the $\lef\wbox$-rule we
  need a case distinction: If the output formula occurs inside
  $\Delta$, then we use $\kax1$ and Lemma~\ref{lem:DI-sound}.  If the
  output formula occurs inside the context $\Gamma\conhole$, then we use
  $\kax2$ and Lemma~\ref{lem:DI-sound-inv}.  %%\qed
\end{proof}

Let us now turn to showing the soundness of the branching rules
$\rig\cand$, $\lef\cor$, $\lef\implies$, and $\cutr$. For this, we
develop appropriate versions of Lemmas \ref{lem:pre-context} and
\ref{lem:DI-sound} that deal with branching behavior. 
Note that, contrary to the intuitionistic case
in~\cite{str:fossacs13}, we do not have such a version of
Lemma~\ref{lem:DI-sound-inv} in the constructive setting. This is due
to the lack of axiom~$\kax3$.

\begin{lemma}\label{lem:pre-context-bin}
  Let $\Xax\subseteq\set{\dax,\tax,\bax,\vax,\fax}$, and let $A$, $B$, $C$,
  and~$D$ be formulas.
  \begin{enumerate}[label=(\roman{*}), leftmargin=*]
  \item\label{l:impl-bin} If $(A\cand B)\implies C$ is provable in $\HCK+\Xax$, then so is
    $((D\implies A)\cand(D\implies B))\implies(D\implies C)$.
      \item\label{l:impl-bin-a} If $(A\cand B)\implies C$ is provable in $\HCK+\Xax$, then so is
        $((D\implies A)\cand(D\cand B))\implies(D\cand C)$.
  \item\label{l:box-bin} If $(A\cand B)\implies C$ is provable in $\HCK+\Xax$, then so is
    $(\wbox A\cand\wbox B)\implies\wbox C$.
  \item\label{l:box-bin-a} If $(A\cand B)\implies C$ is provable in $\HCK+\Xax$, then so is
    $(\wbox A\cand\wdia B)\implies\wdia C$.
  \end{enumerate}
\end{lemma}

\begin{proof}
\ref{l:impl-bin} and \ref{l:impl-bin-a} follow by completeness of $\HCK$ over intuitionistic logic. \ref{l:box-bin} and \ref{l:box-bin-a} follow by necessitation, distributivity of $\wbox$ over $\cand$, and $\kax1$ or $\kax2$ respectively.
  %%\qed
\end{proof}

\begin{lemma}\label{lem:DI-sound-bin}
  Let $\Xax\subseteq\set{\dax,\tax,\bax,\vax,\fax}$, let $\Delta_1$,
  $\Delta_2$, and $\Sigma$ be full sequents, and let $\Gamma\conhole$ be an
  output context.  If
  $(\formula{\Delta_1}\cand\formula{\Delta_2})\implies\formula\Sigma$ is
  provable in $\HCK+\Xax$, then so is
  $(\formula{\Gamma\cons{\Delta_1}}\cand\formula{\Gamma\cons{\Delta_2}})
  \implies\formula{\Gamma\cons{\Sigma}}$.
\end{lemma}

\begin{proof}%%[of Lemma~\ref{lem:DI-sound-bin}]
  Induction on the structure of $\Gamma\conhole$, 
  %%(which is of shape~\eqref{eq:wcontext}), 
  using
  Lemma~\ref{lem:pre-context-bin}.\ref{l:impl-bin} and~\ref{l:box-bin}.
  %%\qed
\end{proof}

\begin{lemma}\label{lem:2-rule-sound}
  Let $\Xax\subseteq\set{\dax,\tax,\bax,\vax,\fax}$, and let~ 
%  \begin{equation*}
  {\small\upsmash{$\vliinf{\rr}{}{\Gamma_3}{\Gamma_1}{\Gamma_2}$}}
%  \end{equation*}
  be an instance of $\rig\cand$, $\lef\cor$, $\lef\implies$, or $\cutr$.  Then
  $(\formula{\Gamma_1}\cand\formula{\Gamma_2})\implies\formula{\Gamma_3}$ is
  provable in $\HCK+\Xax$.
\end{lemma}

\begin{proof}
  For the $\rig\cand$- and $\lef\cor$-rules, this follows immediately
  from Lemma~\ref{lem:DI-sound-bin} and provable formulas $(A\cand
  B)\implies (A\cand B)$ and $((A\implies C)\cand(B\implies C))\implies
  ((A\cor B)\implies C)$, respectively. For
  $\lef\implies$, note that by
  Observation~\ref{obs:context} and Definition~\ref{def:pruning}, the rule is
  of shape
  \begin{align*}
    \vliiinf{\lef{\implies}}{}{\Gamma'\cons{\Lambda\cons{\lef{A\implies
            B}},\rig\Pi}}{
      \Gamma'\cons{\Lambda\cons{\rig A}}}{}{
      \Gamma'\cons{\Lambda\cons{\lef B},\rig\Pi}}
  \end{align*}
  where $\Gamma'\conhole$, $\Lambda\conhole$, and $\Pi\conhole$ are
  output contexts. In particular,
  let $$\Lambda\conhole =
    \Lambda_0,\wbr{\Lambda_1,\wbr{\ldots,\wbr{\Lambda_n,\conhole}\ldots}}\quadfs$$
  Now let $P=\formula{\rig\Pi}$ and $L_i=\formula{\Lambda_i}$ for $i=0\ldots n$, and let
  \begin{align*}
    L_X&=\formula{\Lambda\cons{\rig A}}=
    L_0\implies\wbox(L_1\implies\wbox(L_2\implies\wbox(\cdots
    \implies\wbox(L_n\implies A)\cdots)))
    \\
    L_Y&=\formula{\Lambda\cons{\lef B}} =
    L_0\cand\wdia(L_1\cand\wdia(L_2\cand\wdia(\cdots
    \cand\wdia(L_n\cand B)\cdots)))
    \\
    L_Z&=\formula{\Lambda\cons{\plef{A\implies B}}} =
    L_0\cand\wdia(L_1\cand\wdia(L_2\cand\wdia(\cdots
    \cand\wdia(L_n\cand(A\implies B))\cdots)))
    %\\%[\belowdisplayskip]
  \end{align*}
  To be able to apply Lemma~\ref{lem:DI-sound-bin}, we need to
  show that $(L_X\cand(L_Y\implies P))\implies(L_Z\implies P)$ is
  provable in $\HCK+\Xax$. But this follows from $(L_X\cand L_Z)\implies
  L_Y$, which can be shown provable in $\HCK+\Xax$ using an induction
  on $n$ together with
  Lemma~\ref{lem:pre-context-bin}.\ref{l:impl-bin-a}
  and~\ref{l:box-bin-a}.
  For the $\cutr$-rule we additionally observe that $A\implies A$ is always provable.
  %%\qed
\end{proof}

\begin{remark}\label{rem:NCK-HCK}
  From the lemmas presented so far, we now have that $\NCK+\weakr+\cutr$ is sound with respect to $\HCK$, i.e.\ we have proved already Theorem~\ref{thm:sound-con} in the case of $\Xax=\Yax=\emptyset$. This means that if we have a proof of a formula $A$ in $\NCK + \weakr+\cutr$ in which we allow $\Xax\subseteq\set{\dax,\tax,\bax,\vax,\fax}$ to occur as proper axioms, then we have that $ \bigwedge \Xax \implies A$ is provable in $\HCK$, by purely propositional logic, and therefore $A$ is provable in $\HCK+\Xax$.
\end{remark}

We use the observation in the above remark to prove the following lemma.

\begin{lemma}\label{lem:b5aux}
  Let $S$ and $D$ be arbitrary formulas. Then we have the following:
  \begin{enumerate}[label=(\roman{*}), leftmargin=*]
  \item\label{i:bG} $(S\cand\wdia D)\implies\wdia(\wdia S\cand D)$ is a theorem of $\HCK+\bax$.
  \item\label{i:bS} $\wbox(D\implies\wbox S)\implies(\wdia D\implies S)$ is a
    theorem of $\HCK+\bax$.
  \item\label{i:bD} $\wbox(\wdia S\implies D)\implies(S\implies\wbox D)$ is a
    theorem of $\HCK+\bax$.
  \item\label{i:5G} $(\wdia S\cand\wdia D)\implies \wdia(\wdia S\cand D)$ is a theorem of $\HCK+\fax$.
  \item\label{i:5S} $\wbox(D\implies\wbox S)\implies(\wdia D\implies \wbox S)$ is a
    theorem of $\HCK+\fax$.
  \item\label{i:5D} $\wbox(\wdia S\implies D)\implies(\wdia S\implies\wbox D)$ is a
    theorem of $\HCK+\fax$.
 \end{enumerate}
\end{lemma}

\begin{proof} %[of Lemma~\ref{lem:b5aux}]       
  In the following we show that the formulas in \ref{i:bG}--\ref{i:5D}
  can be proved in $\NCK+\cutr$ extended by $\bax$ or $\fax$, as appropriate, as a proper axiom. Our lemma then follows from
  Remark~\ref{rem:NCK-HCK}.
  \begin{enumerate}[label=(\roman{*}), leftmargin=*]
  \item 
    {\small$\qquad
      \svlderivation{
        \vlin{\rig{\implies}}{}{\rig{(S \cand \wdia D) \implies \wdia (\wdia S \cand D)}}{
          \vlin{\lef\cand}{}{\lef{S \cand \wdia D},\rig{\wdia (\wdia S \cand D)}}{
            \vlin{\lef\wdia}{}{\lef{S}, \wdia \lef{D},\rig{\wdia (\wdia S \cand D)}}{
              \vlin{\rig\wdia}{}{\lef{S}, \wbr{\lef{D}},\rig{\wdia (\wdia S \cand D)}}{
                \vliin{\rig\cand}{}{\lef{S}, \wbr{\lef{D}, \wdia S \cand \rig{D}}}{
                  \vlin{\weakr}{}{\lef{S}, \wbr{\lef D,\rig{\wdia S}}}{
                    \vliin{\cutr}{}{\lef{S}, \wbr{\rig{\wdia S}}}{
                      \vlin{\weakr}{}{\lef{S}, \rig{S \implies \wbox \wdia S}}{
                        \vlin{\bax}{}{\rig{S \implies \wbox \wdia S}}{
                          \vlhy{}}}}{
                      \vliin{\lef\implies}{}{\lef{S \implies \wbox \wdia S},\lef{S}, \wbr{\lef D,\rig{\wdia S}}}{
                        \vlin{\idr}{}{\lef S,\rig S}{
                          \vlhy{}}}{
                        \vlin{\lef\wbox}{}{\lef{\wbox \wdia S},\lef{S}, \wbr{\lef D,\rig{\wdia S}}}{
                          \vlin{\idr}{}{\lef{S}, \wbr{\lef{\wdia S},\lef D,\rig{\wdia S}}}{
                            \vlhy{}}}}}}}{
                  \vlin{\idr}{}{\lef{S}, \wbr{\lef D, \rig D}}{
                    \vlhy{}}}}}}}}
      $}
  \item
    {\small$\qquad
      \svlderivation{
        \vlin{\rig\implies}{}{\rig{\wbox (D \implies \wbox S) \implies (\wdia D \implies S)}}{
          \vlin{\rig\implies}{}{\lef{\wbox (D \implies \wbox S)},\rig{\wdia D \implies S}}{
            \vlin{\lef\wdia}{}{\lef{\wbox (D \implies \wbox S)},\lef{\wdia D},\rig S}{
              \vlin{\lef\wbox}{}{\lef{\wbox (D \implies \wbox S)},\wbr{\lef D},\rig S}{
                \vliin{\lef\implies}{}{\wbr{\lef{D \implies \wbox S},\lef D},\rig S}{
                  \vlin{\idr}{}{\wbr{\lef D,\rig D}}{
                    \vlhy{}}}{
                  \vlin{\weakr}{}{\wbr{\lef{\wbox S},\lef D},\rig S}{
                    \vliin{\cutr}{}{\wbr{\lef{\wbox S}},\rig S}{
                      \vlin{\weakr}{}{\wbr{\lef{\wbox S}},\rig{\wdia\wbox S\implies S}}{
                        \vlin{\bax}{}{\rig{\wdia\wbox S\implies S}}{
                          \vlhy{}}}}{
                      \vliin{\lef\implies}{}{\wbr{\lef{\wbox S}},\lef{\wdia\wbox S\implies S},\rig S}{
                        \vlin{\rig\wdia}{}{\wbr{\lef{\wbox S}},\rig{\wdia\wbox S}}{
                          \vlin{\idr}{}{\wbr{\lef{\wbox S},\rig{\wbox S}}}{
                            \vlhy{}}}}{
                        \vlin{\idr}{}{\wbr{\lef{\wbox S}},\lef{S},\rig S}{
                          \vlhy{}}}}}}}}}}}$}
  \item
    {\small$\qquad
      \svlderivation{
        \vlin{\rig\implies}{}{
          \rig{\wbox(\wdia S \implies D) \implies (S \implies \wbox D)}}{
          \vlin{\rig\implies}{}{
            \lef{\wbox(\wdia S \implies D)},\rig{S \implies \wbox D}}{
            \vlin{\rig\wbox}{}{
              \lef{\wbox(\wdia S \implies D)},\lef{S},\rig{\wbox D}}{
              \vlin{\lef\wbox}{}{
                \lef{\wbox(\wdia S \implies D)},\lef{S},\wbr{\rig D}}{
                \vliin{\lef\implies}{}{
                  \lef{S},\wbr{\lef{\wdia S \implies D},\rig D}}{
                  \vlhtr{\DD_1}{\lef{S},\wbr{\rig{\wdia S}}}}{
                  \vlin{\idr}{}{\lef{S},\wbr{\lef{D},\rig D}}{
                    \vlhy{}}}}}}}}$}
    \qqquad
    (iv)
    {\small$\qquad
      \svlderivation{
        \vlin{\rig\implies}{}{
          \rig{(\wdia S \cand \wdia D) \implies\wdia (\wdia S \cand D)}}{
          \vlin{\lef\cand}{}{
            \lef{\wdia S \cand \wdia D},\rig{\wdia (\wdia S \cand D)}}{
            \vlin{\lef\wdia}{}{
              \lef{\wdia S},\lef{\wdia D},\rig{\wdia (\wdia S \cand D)}}{
              \vlin{\rig\wdia}{}{
                \lef{\wdia S},\wbr{\lef D},\rig{\wdia (\wdia S \cand D)}}{
                \vliin{\rig\cand}{}{
                  \lef{\wdia S},\wbr{\lef D,\rig{\wdia S \cand D}}}{
                  \vlin{\weakr}{}{
                    \lef{\wdia S},\wbr{\lef D,\rig{\wdia S}}}{
                    \vlhtr{\DD_2}{\lef{\wdia S},\wbr{\rig{\wdia S}}}}}{
                  \vlin{\idr}{}{
                    \lef{\wdia S},\wbr{\lef D,\rig D}}{
                    \vlhy{}}}}}}}}$}
    \stepcounter{enumi}
  \item 
    {\small$\qquad
      \svlderivation{
        \vlin{\rig\implies}{}{
          \rig{\wbox (D \implies \wbox S) \implies (\wdia D \implies \wbox S)}}{
          \vlin{\rig\implies}{}{
            \lef{\wbox (D \implies \wbox S)},\rig{\wdia D \implies \wbox S}}{
            \vlin{\lef\wdia}{}{
              \lef{\wbox (D \implies \wbox S)},\lef{\wdia D},\rig{\wbox S}}{
              \vlin{\lef\wbox}{}{
                \lef{\wbox (D \implies \wbox S)},\wbr{\lef D},\rig{\wbox S}}{
                \vliin{\lef\implies}{}{
                \wbr{\lef{D \implies \wbox S},\lef D},\rig{\wbox S}}{
                  \vlin{\idr}{}{
                    \wbr{\rig{D},\lef D}}{
                    \vlhy{}}}{
                  \vlin{\weakr}{}{
                    \wbr{\lef{\wbox S},\lef D},\rig{\wbox S}}{
                    \vlhtr{\DD_3}{\wbr{\lef{\wbox S}},\rig{\wbox S}}}}}}}}}$}
    \qqquad
    (vi)
    {\small$\qquad
      \svlderivation{
        \vlin{\rig\implies}{}{
          \rig{\wbox (\wdia S \implies D) \implies (\wdia S \implies \wbox D)}}{
          \vlin{\rig\implies}{}{
            \lef{\wbox (\wdia S \implies D)},\rig{\wdia S \implies \wbox D}}{
            \vlin{\rig\wbox}{}{
              \lef{\wbox (\wdia S \implies D)},\lef{\wdia S},\rig{\wbox D}}{
              \vlin{\lef\wbox}{}{
                \lef{\wbox (\wdia S \implies D)},\lef{\wdia S},\wbr{\rig D}}{
                \vliin{\lef\implies}{}{
                  \lef{\wdia S},\wbr{ \lef{\wdia S \implies D},\rig D}}{
                  \vlhtr{\DD_2}{\lef{\wdia S},\wbr{ \rig{\wdia S}}}}{
                  \vlin{\idr}{}{
                    \lef{\wdia S},\wbr{ \lef{D},\rig D}}{
                    \vlhy{}}}}}}}}$}
  \end{enumerate}
  where $\DD_1$ is a subderivation of (i), $\DD_2$ is the same as
  $\DD_1$, except that we use $\fax$ instead of $\bax$, and $\DD_3$ is
  a variant of a subderivation of (ii), using $\fax$ instead of
  $\bax$.
\end{proof}

Now we can show soundness of the rules in Figures~\ref{fig:dt4}
and~\ref{fig:CXstr}, which we need to complete the proof of
Theorem~\ref{thm:sound-con}.

\begin{lemma}\label{lem:x-rule-sound}
  Let $\Xax\subseteq\set{\dax,\tax,\vax}$, let $\Yax\subseteq\set{\dax,\tax,\bax,\vax,\fax}$, let
  $\xax\in\Xax$, let $\yax\in\Yax$, and let~
%  \begin{equation*}
    \hbox{$\smash{\vlinf{\rr}{}{\Gamma_2}{\Gamma_1}}$} \strut
%  \end{equation*}
  be an instance of $\rig\xax$ or $\lef\xax$ or $\ydot$. Then
  $\formula{\Gamma_1}\implies\formula{\Gamma_2}$ is provable in $\HCK+\Xax+\Yax$.
\end{lemma}

\begin{proof}
  For $\dlef$, $\drig$, $\tlef$, $\trig$, $\tdot$, and $\vdot$ this
  follows immediately from Lemmas \ref{lem:DI-sound}
  and~\ref{lem:DI-sound-inv} and the corresponding axioms, shown
  in~\eqref{eq:ax}. For $\vlef$ and $\vrig$, observe that these two
  rules can be derived using the rules $\lef\wdia$ and $\rig\wdia$,
  respectively, and
  \begin{equation}
    \label{eq:newfour}
    \vlinf{\vlefp}{}{\Gamma\cons{\lef{\wbox A}}}{\Gamma\cons{\lef{\wbox\wbox A}}}
    \qquand
    \vlinf{\vrigp}{}{\Gamma\cons{\rig{\wdia A}}}{\Gamma\cons{\rig{\wdia\wdia A}}}
  \end{equation}
  respectively. The soundness of the two rules in~\eqref{eq:newfour}
  follows immediately from Lemmas \ref{lem:DI-sound}
  and~\ref{lem:DI-sound-inv} and the $\vax$-axiom.
  For $\ddot$ we need to show that $\top\implies\wdia\top$ is provable,
  which follows from $\top\implies\wbox\top$ and the $\dax$-axiom. For
  $\bdot$, we have to make a case analysis on where the output formula
  is. If it is in $\Gamma\conhole$, soundness of the rule follows from
  Lemma~\ref{lem:b5aux}.\ref{i:bG} and Lemma~\ref{lem:DI-sound-inv}.
  If it is in $\Sigma$, we use Lemma~\ref{lem:b5aux}.\ref{i:bS} and
  Lemma~\ref{lem:DI-sound}, and if it is in $\Delta$, we use
  Lemma~\ref{lem:b5aux}.\ref{i:bD} and Lemma~\ref{lem:DI-sound}. For
  the rule $\fdot$ we proceed similarly, using
  Lemma~\ref{lem:b5aux}.\ref{i:5G}--\ref{i:5D} instead.
  %%\qed
\end{proof}

Now we can put everything together to prove Theorem~\ref{thm:sound-con}.  
\begin{proof}[Proof of Theorem~\ref{thm:sound-con}]
  Point (i) is just Lemmas \ref{lem:0-rule-sound}, \ref{lem:1-rule-sound},
  \ref{lem:2-rule-sound}, and~\ref{lem:x-rule-sound}. Point~(ii) follows
  immediately from~(i) by induction on the size of the derivation.
  %%\qed
\end{proof}

Having established the soundness of our system, we can use it to
make some interesting observations. Surprisingly, the $\bax$-axiom
entails the axioms $\kax3$ and $\kax5$ (shown in~\eqref{eq:ik} in the
introduction), as can be seen by the following two derivations in
$\NCK+\bdot$:
\begin{equation}%\small
  \hskip-1em
  \svlderivation{
    \vlin{\rig\implies}{}{\prig{\wdia(A\cor B)\implies(\wdia A\cor \wdia B)}}{
      \vlin{\lef\wdia}{}{\plef{\wdia(A\cor B)},\prig{\wdia A\cor \wdia B}}{
        \vlin{\bdot}{}{\wbr{\plef{A\cor B}},\prig{\wdia A\cor \wdia B}}{
          \vliin{\lef\cor}{}{\wbr{\plef{A\cor B},\wbr{\prig{\wdia A\cor \wdia B}}}}{
            \vlin{\rig\cor}{}{\wbr{\plef{A},\wbr{\prig{\wdia A\cor \wdia B}}}}{
              \vlin{\bdot}{}{\wbr{\plef{A},\wbr{\prig{\wdia A}}}}{
                \vlin{\rig\wdia}{}{\wbr{\wbr{\wbr{\lef A},\prig{\wdia A}}}}{
                  \vlin{\idr}{}{\wbr{\wbr{\wbr{\lef A, \rig A}}}}{
                    \vlhy{}}}}}}{
            \vlin{\rig\cor}{}{\wbr{\plef{B},\wbr{\prig{\wdia A\cor \wdia B}}}}{
              \vlin{\bdot}{}{\wbr{\plef{B},\wbr{\prig{\wdia B}}}}{
                \vlin{\rig\wdia}{}{\wbr{\wbr{\wbr{\lef B},\prig{\wdia B}}}}{
                  \vlin{\idr}{}{\wbr{\wbr{\wbr{\lef B, \rig B}}}}{
                    \vlhy{}}}}}}}}}}
  \qqquand
    \svlderivation{
    \vlin{\rig\implies}{}{\rig{\wdia\bot\implies\bot}}{
      \vlin{\lef\wdia}{}{\lef{\wdia\bot},\rig\bot}{
        \vlin{\bdot}{}{\wbr{\lef{\bot}},\rig\bot}{
          \vlin{\lef\bot}{}{\wbr{\lef{\bot},\wbr{\rig\bot}}}{
            \vlhy{}}}}}}
\end{equation}
%%\end{equation}
While the proof of $\kax5$ can be easily shown directly in the Hilbert
system, the proof of $\kax3$ in $\HCK+\bax$ is not so simple. From our
cut-elimination result in Section~\ref{sec:cutelim} it will follow
that the $\fax$ axiom alone is not enough to derive $\kax3$ or
$\kax5$. But since $\bax$ is derivable in $\CSfive$, both $\kax3$ or
$\kax5$ are derivable in $\CSfive$.

%%%%%%%%%%%%%%%%%%%%%%%%%%%%%%%%%%%%%%%%%%%%%%%%%%%%%%%%%%%%%%%%%%%%%%%%
%%%%%%%%%%%%%%%%%%%%%%%%%%%%%%%%%%%%%%%%%%%%%%%%%%%%%%%%%%%%%%%%%%%%%%%%
%%%%%%%%%%%%%%%%%%%%%%%%%%%%%%%%%%%%%%%%%%%%%%%%%%%%%%%%%%%%%%%%%%%%%%%%

\section{Completeness}\label{sec:completeness}

\begin{figure}  \[
  \begin{array}{c}
  \begin{array}{ccc}
  \dax\text{ axiom:}
  &
  \tax\text{ axiom:}
  &
  \bax\text{ axiom:}
  \\[-1ex]
%  \noalign{\smallskip}
  \svlderivation{
    \vlin{\rig\implies}{}{\wbox A \implies \wdia \rig{A}}{
    \vlin{\rig\dax}{}{\wbox \lef{A}, \wdia \rig{A}}{
    \vlin{\lef{\wbox}}{}{\wbox \lef{A}, \wbr{\rig{A}}}{
    \vlin{\idr}{}{\wbr{\lef{A}, \rig{A}}}{\vlhy{}}
    }
    }
    }
    }
  \quad&\quad
  \svlderivation{
  \vliin{\rig{\cand}}{}{(A \implies \wdia A) \cand \rig{(\wbox A \implies A)}}{
  \vlin{\rig{\implies}}{}{A \implies \wdia \rig{A}}{
  \vlin{\rig{\tax}}{}{\lef{A}, \wdia \rig{A}}{
  \vlin{\idr}{}{\lef A , \rig A }{\vlhy{}}
  }
  }
  }{
  \vlin{\rig{\implies}}{}{\wbox A \implies \rig{A}}{
  \vlin{\lef{\tax}}{}{\wbox \lef{A}, \rig{A}}{
  \vlin{\idr}{}{\lef A , \rig A }{\vlhy{}}
  }
  }
  }
  }
  \quad & \quad 
  \svlderivation{
  \vliin{\rig{\cand}}{}{(\wdia \wbox A \implies A) \cand \rig{(A \implies \wbox \wdia A)}}{
  \vlin{\rig{\implies}}{}{ \wdia \wbox A \implies \rig{A} }{
  \vlin{\lef{\hat\wdia}}{}{\wdia \wbox \lef{A}, \rig{A}}{
  \vlin{\bdot}{}{\wbr{\wbox \lef{A}}, \rig{A}}{
  \vlin{\lef{\wbox}}{}{\wbr{\wbox \lef{A}, \wbr{\rig{A}}}}{
  \vlin{\idr}{}{\wbr{\wbr{\lef{A}, \rig{A}}}}{\vlhy{}}
  }
  }
  }
  }
  }{
  \vlin{\rig{\implies}}{}{ A \implies \wbox \wdia A }{
  \vlin{\rig{\wbox}}{}{ \lef{A}, \wbox \wdia \rig{A} }{
  \vlin{ \bdot }{}{ \wbr{\wbox \lef{A}}, \rig{A} }{
  \vlin{\rig{\wdia}}{}{ \wbr{\wbr{\lef{A}}, \wdia \rig{A}} }{
  \vlin{\idr}{}{ \wbr{\wbr{\lef{A}, \rig{A}}} }{\vlhy{}}
  }
  }
  }
  }
  }
  }
  \end{array}
  \\
  \noalign{\bigskip}
  \noalign{\smallskip}
  \noalign{\smallskip}
  \begin{array}{cc}
  \vax\text{ axiom:}
  &
  \fax\text{ axiom:}
  \\[-1ex]
%%%%%  \noalign{\smallskip}
  \svlderivation{
  \vliin{\rig\cand}{}{ (\wbox A \implies \wbox \wbox A) \cand \rig{(\wdia \wdia A \implies \wdia A)} }{
  \vlin{\rig{\implies}}{}{ \wbox A \implies \wbox \wbox \rig{A} }{
  \vlin{ \rig{\wbox} }{}{ \wbox \lef{A}, \wbox \wbox \rig{A} }{
  \vlin{\lef{\vax}}{}{ \wbox \lef{A}, \wbr{\wbox \rig{A}} }{
  \vlin{\idr}{}{ \wbr{\wbox \lef{A}, \wbox \rig{A}} }{\vlhy{}}
  }
  }
  }
  }{
  \vlin{ \rig{\implies} }{}{ \wdia \wdia A \implies \wdia \rig{A} }{
  \vlin{ \lef{\hat\wdia} }{}{ \wdia \wdia \lef{A}, \wdia \rig{A} }{
  \vlin{\rig{\vax}}{}{ \wbr{\wdia \lef{A}}, \wdia \rig{A} }{
  \vlin{\idr}{}{ \wbr{\wdia \lef{A}, \wdia \rig{A}} }{\vlhy{}}
  }
  }
  }
  }
  }
  \quad & \quad 
  \svlderivation{
  \vliin{\rig{\cand}}{}{ (\wdia \wbox A \implies \wbox A) \cand \rig{(\wdia A \implies \wbox \wdia A)} }{
  \vlin{ \rig{\implies} }{}{ \wdia \wbox A \implies \wbox \rig{A} }{
  \vlin{ \lef{\hat\wdia} }{}{ \wdia \wbox \lef{A}, \wbox \rig{A} }{
  \vlin{ \rig{\wbox} }{}{ \wbr{\wbox \lef{A}}, \wbox \rig{A} }{
  \vlin{ \fdot }{}{ \wbr{\wbox \lef{A}}, \wbr{\rig{A}} }{
  \vlin{ \lef{\wbox} }{}{ \wbr{\wbox \lef{A}, \wbr{\rig{A}}} }{
  \vlin{\idr}{}{ \wbr{\wbr{\lef{A}, \rig{A}}} }{\vlhy{}}
  }
  }
  }
  }
  }
  }{
  \vlin{ \rig{\implies} }{}{ \wdia A \implies \wbox \wdia \rig{A} }{
  \vlin{\lef{\hat\wdia}  }{}{ \wdia \lef{A}, \wbox \wdia \rig{A} }{
  \vlin{ \rig{\wbox} }{}{ \wbr{\lef{A}}, \wbox \wdia \rig{A} }{
  \vlin{ \fdot }{}{ \wbr{\lef{A}}, \wbr{\wdia \rig{A}} }{
  \vlin{ \rig{\wdia} }{}{ \wbr{\wbr{\lef{A}}, \wdia \rig{A}} }{
  \vlin{\idr}{}{ \wbr{\wbr{\lef{A},\rig{A}}} }{\vlhy{}}
  }
  }
  }
  }
  }
  }
  }
  \end{array}
  \end{array}
  \]
  \caption{Proofs of the axioms $\dax$, $\tax$, $\bax$, $\vax$, and
    $\fax$ in our system
}
  \label{fig:prf-of-ax}
\end{figure}

Completeness is also shown with respect to the Hilbert system. This is
in fact very similar to the completeness proof for intuitionistic modal
logic given in~\cite{str:fossacs13}. 
To simplify our cut-elimination
argument in Section~\ref{sec:cutelim} we will put a restriction on the
$\lef\wdia$-rule: we define the system $\NCKp$ to be $\NCK$ with the
$\lef\wdia$-rule replaced by
\begin{equation}
  \vlinf{\lef{\hat{\wdia}}}{}{\colGcons{\plef{\wdia A},\rig\Pi}}{\colGcons{\wbr{\lef A},\rig\Pi}} 
\end{equation}

\begin{theorem}[Completeness]\label{thm:complete}
  Let $\Xax\subseteq\set{\dax,\tax,\vax}$ and
  $\Yax\subseteq\set{\dax,\tax,\bax,\vax,\fax}$. Then every formula
  that is provable in $\HCK+\Xax+\Yax$ is provable in
  $\NCKp+\Xlefrig+\Ydot+\cutr$.
\end{theorem}

\begin{proof}
  Clearly, all axioms of propositional intuitionistic logic are
  provable in $\NCKp$. The axioms $\kax1$ and $\kax2$ are provable in
  $\NCKp$, by the same derivations as in~\cite{str:fossacs13}, so we
  do not repeat them here. Note that the derivations for
  $\kax3$, $\kax4$, and~$\kax5$ of~\cite{str:fossacs13} are not valid in our setting
  because of the restrictions to the $\lef\cor$-, $\lef\implies$-, and
  $\lef\bot$-rules, respectively.  Figure~\ref{fig:prf-of-ax} shows that each
  axiom $\xax\in\Xax\cup\Yax$ is provable in
  $\NCKp+\Xlefrig+\Ydot$.  Finally the rules $\mpr$ and $\necr$,
  shown in~\eqref{eq:mp-nec}, can be simulated by the rules $\cutr$
  and $\boxnecr$, shown in~\eqref{eq:cutr}, as follows:
  \begin{equation*}
    \svlderivation{
      \vliin{\cutr}{}{\rig B}{
        \vlhy{\rig A}}{
        \vliin{\cutr}{}{\lef A,\rig B}{
          \vlhy{\rig{A\implies B}}}{
          \vliin{\lef\implies}{}{\lef{A\implies B},\lef A,\rig B}{
            \vlin{\idr}{}{\rig A,\lef A}{
              \vlhy{}}}{
            \vlin{\idr}{}{\lef B,\lef A,\rig B}{
              \vlhy{}}}}}}
    \qqqqquad
    \svlderivation{
      \vlin{\rig\wbox}{}{\rig{\wbox A}}{
        \vlin{\boxnecr}{}{\wbr{\rig A}}{
          \vlhy{\rig A}}}}
  \end{equation*}
  From here we appeal to the admissibility of the $\boxnecr$-rule,
  which follows by a straightforward induction on
  the size of a derivation.  %%\qed
\end{proof}

Theorems \ref{thm:sound-con} and~\ref{thm:complete} are enough to give sound and
complete nested sequent systems with cut for any logic in the cube
shown in Figure~\ref{fig:cube-cons}, by simply adding the
corresponding structural rules from Figure~\ref{fig:CXstr}. If one of
the axioms is $\dax$, $\tax$, or $\vax$, then we can use the logical
rules from Figure~\ref{fig:dt4} instead of the structural rule. For
example, for $\CSfour$, we can use $\set{\tdot,\vdot}$ or
$\set{\tlef,\trig,\vdot}$ or $\set{\tdot,\vlef,\vrig}$ or
$\set{\tlef,\trig,\vlef,\vrig}$ or any union of these sets.

In the
next section we show cut-elimination for $\NCKp+\Xlefrig+\Ydot$,
yielding completeness for the cut-free system. However, this is not
achieved for every subset of $\Xlefrig \cup \Ydot$ with
$\Xax\subseteq\set{\dax,\tax,\vax}$ and
$\Yax\subseteq\set{\dax,\tax,\bax,\vax,\fax}$. In fact, it can be
shown that, for example, $\NCKp+\vdot$ is not complete for $\CKfour$. On the other
hand, we have:

\begin{theorem}[Cut-free Completeness]\label{thm:complete-cutfree}
  Let $\Xax\subseteq\set{\dax,\tax,\vax}$ and
  $\Yax\subseteq\set{\dax,\bax,\fax}$, such that
  if $\tax\in\Xax$ and $\fax\in \Yax$ then $\bax\in\Yax$, and
if $\bax\in\Yax$ or $\fax \in \Yax$ then $\vax\in\Xax$. Then every formula that is provable in 
  $\HCK+\Xax+\Yax$ is also provable in
  $\NCKp+\Xlefrig+\Ydot$.
\end{theorem}

Thus, if we want a cut-free system for $\CSfour$, we have to add the
rules $\set{\tlef,\trig,\vlef,\vrig}$ to $\NCKp$.  The proof of Theorem~\ref{thm:complete-cutfree}
relies on the cut-elimination argument presented in the next
section, and can thus be presented only at the end of
Section~\ref{sec:cutelim}. 

Looking back at the cube in Figure~\ref{fig:cube-cons}, we can see
that Theorem~\ref{thm:complete-cutfree} gives us cut-free systems for
the logics $\CK$, $\CKfour$, $\CKfourfive$, $\CD$, $\CDfour$,
$\CDfourfive$, $\CT$, $\CSfour$, and $\CSfive$. The logics for which
our cut-elimination proof does not apply are
%%%that are still missing, are
$\CKB$, $\CKfive$, $\CKBfive$, $\CDfive$, $\CDB$, and $\CTB$.
%%%\todo{check!}

%%%%%%%%%%%%%%%%%%%%%%%%%%%%%%%%%%%%%%%%%%%%%%%%%%%%%%%%%%%%%%%%%%%%%%%%
%%%%%%%%%%%%%%%%%%%%%%%%%%%%%%%%%%%%%%%%%%%%%%%%%%%%%%%%%%%%%%%%%%%%%%%%
%%%%%%%%%%%%%%%%%%%%%%%%%%%%%%%%%%%%%%%%%%%%%%%%%%%%%%%%%%%%%%%%%%%%%%%%

\section{Cut-Elimination}\label{sec:cutelim}

By inspection of the statement of Theorem~\ref{thm:complete-cutfree}, we have that $\tax$ and $\vax$ must be present
as logical rules, and $\bax$ and $\fax$ as structural rules, whereas $\dax$ can be present in either variation.
This is due to the following result, whose proof is straightforward.

\begin{proposition}\label{PropDaxDerivable}
  (i) The rules $\dlef$ and $\drig$ are derivable in
  $\set{\lef\wbox,\ddot}$ and $\set{\rig\wdia,\ddot}$,
  respectively. (ii) The rule $\ddot$ is admissible for any subsystem
  of $\NCKp+\Xlefrig+\Ydot$, provided $\dax\in\Xax$.
\end{proposition}
However, our cut-elimination argument becomes slightly simpler if we
work with $\ddot$ instead of $\dlef$ and $\drig$. To summarize, the
following definition fixes the axiom sets our cut-elimination proof deals with.

\begin{definition}
  Let $\Xax,\Yax\subseteq\set{\dax,\tax,\bax,\vax,\fax}$. We call the
  pair $\tuple{\Xax,\Yax}$ \emph{safe} if
  $\Xax\subseteq\set{\tax,\vax}$ and
  $\Yax\subseteq\set{\dax,\bax,\fax}$, such that
if  $\tax\in \Xax$ and $\fax\in\Yax$ then $\bax\in\Yax$, and
if  $\bax\in\Yax$ or $\fax\in\Yax$ then $\vax\in\Xax$.
\end{definition}

We can now state our cut-elimination result in a concise way:

\begin{theorem}[Cut-Elimination]\label{thm:cut-elim}
  Let $\tuple{\Xax,\Yax}$ be a safe pair of axioms, and let $\DD$ be a
  proof in $\NCKp+\Xlefrig+\Ydot+\cutr$. Then
  there is a proof $\DD'$ of the same conclusion in
  $\NCKp+\Xlefrig+\Ydot$.
\end{theorem}

The rest of this section is dedicated to the proof of Theorem~\ref{thm:cut-elim}.

Since our cut-elimination strategy might seem unorthodox, we first explain some of the problems we encountered. Consider the
following derivation:
\begin{equation}
  \label{eq:weird}
  \svlderivation{
    \vliin{\cutr}{}{\Gcons{\Tcons{\plef{C\cor B}},\rig\Pi}}{
      \vliin{\lef\cor}{}{\Gcons{\Tcons{\plef{C\cor B},\rig{\wdia A}}}}{
        \vlhy{\Gcons{\Tcons{\plef{C},\rig{\wdia A}}}}}{
        \vlhy{\Gcons{\Tcons{\plef{B},\rig{\wdia A}}}}}}{
      \vlin{\lef\wdia}{}{\Gcons{\Tcons{\plef{C\cor B},\lef{\wdia A}},\rig\Pi}}{
        \vlhy{\Gcons{\Tcons{\plef{C\cor B},\wbr{\lef A}},\rig\Pi}}}}}
\end{equation}
We cannot permute the instance of $\lef\cor$ under the $\cutr$ because
in general it is not applicable in $\Gcons{\Tcons{\plef{C\cor
      B}},\rig\Pi}$. On the other hand, we cannot reduce the rank of
the cut along the main connective of the cut formula $\wdia A$,
since there is no invertible rule for $\rig{\wdia A}$, and
different things might happen in the left two branches. Furthermore,
we cannot just impose the same restriction that we impose on the
$\lef\cor$ rule also on the $\cutr$ rule, because then we would not be
able to reduce the cut rank in the ordinary $\rig\wdia$-$\lef\wdia$
cases. The situation in~\eqref{eq:weird} is the reason we work with
the rule $\lef{\hat{\wdia}}$ instead of $\lef\wdia$. Note that
imposing the same restriction on all logical rules would make other
permutation cases difficult.

\medskip

In what follows we will use the shorthand $\wbrn{\Gamma}n$ to denote $\Gamma$ with $n$ pairs of brackets around it, i.e.\ $\overbrace{[\vldots[}^n\ \Gamma\ \overbrace{]\vldots]}^n$. Also, we define the \emph{depth} of a context $\Gamma\conhole$ to be the
number of bracket pairs in whose scope the hole of $\Gamma\conhole$ appears, i.e., the depth of 
$\Delta_0,\wbr{\Delta_1,\wbr{\ldots,\wbr{\Delta_n,\conhole}\ldots}}$
is $n$.

We consider \emph{super rule} variants of the rules $\vlef$,
$\vrig$, $\bdot$, and~$\fdot$, shown in Figure~\ref{fig:s4}, obtained from unboundedly many
applications of the corresponding normal rules in a certain way. For a safe pair
$\tuple{\Xax,\Yax}$ of axioms, we define
\begin{equation*}
  \sXlefrig=
  \begin{cases}
    \Xlefrig &
    \text{ if } \vax\not\in\Xax \\
    (\Xlefrig\setminus\set{\vrig,\vlef})\cup\set{\svrig,\svlef,\svdrig,\svblef} &
    \text{ if } \vax\in\Xax 
  \end{cases}
\end{equation*}
and 
\begin{equation*}
  \sYdot=
  \begin{cases}
    \Ydot & 
    \text{ if } \bax,\fax\notin\Yax \\
    (\Ydot\setminus\set{\bdot})\cup\set{\sbdot} &
    \text{ if } \bax\in\Yax,\fax\notin\Yax \\
    (\Ydot\setminus\set{\fdot})\cup\set{\sfdot} &
    \text{ if } \bax\notin\Yax,\fax\in\Yax \\
    (\Ydot\setminus\set{\bdot,\fdot})\cup\set{\sfbdot,\sbfdot} &
    \text{ if } \bax,\fax \in \Yax \\
  \end{cases}
\end{equation*}

\begin{figure}[!t]
  \begin{equation*}
    \vlinf{\svrig\!}{}{
      \Gcons{\rig{\wdia A},\Dcons{\emptyset}}}{
      \Gcons{\Dcons{\rig{\wdia A}}}}
    \quad
    \vlinf{\svlef\!}{}{
      \Gcons{\lef{\wbox A},\Dcons{\emptyset}}}{
      \Gcons{\Dcons{\lef{\wbox A}}}}
    \quad
    \vlinf{\svdrig\!}{}{
      \Gcons{\rig{\wdia A},\wbr{\Dcons{\emptyset}}}}{
      \Gcons{\wbr{\Dcons{\rig{A}}}}}
    \quad
    \vlinf{\svblef\!}{}{
      \Gcons{\lef{\wbox A},\wbr{\Dcons{\emptyset}}}}{
      \Gcons{\wbr{\Dcons{\lef{A}}}}}
  \end{equation*}
  \begin{equation*}
    \vlinf{\sbdot}{}{
      \colGcons{\Sigma,\Dcons\emptyset}}{
      \colGcons{\Dcons{\wbrn{\Sigma}n}}}
    \quad
    \vlinf{\sfdot}{}{
      \colGcons{\wbr\Sigma,\Dcons\emptyset}}{
      \colGcons{\Dcons{\wbr{\Sigma}}}}
    \quad
    \vlinf{\sfbdot}{}{
      \colGcons{\Sigma,\Dcons\emptyset}}{
      \colGcons{\Dcons{\wbrn{\Sigma}k}}}
    \quad
    \vlinf{\sbfdot}{}{
      \colGcons{\wbr\Sigma,\Dcons\emptyset}}{
      \colGcons{\Dcons{\wbrn{\Sigma}k}}}
  \end{equation*}
  \caption{Super rules for $\vax$, $\bax$, and $\fax$, where $n$ is
the depth of $\Delta\conhole$ and $1\le k\le n$.}
  \label{fig:s4}
\end{figure}

We need these variants in order to obtain height-preserving
admissibility of certain rules. We have the following proposition:

\begin{proposition}\label{prop:super}
  A sequent is provable in $\NCKp+\sXlefrig+\sYdot$ if and only if it
  is provable in $\NCKp+\Xlefrig+\Ydot$.
\end{proposition}

\begin{proof} %%[of Proposition~\ref{prop:super}]
  One direction follows immediately from the observation that $\vlef$
  and $\vrig$ are special cases of $\svlef$ and $\svrig$,
  respectively, and that $\bdot$ is a special case of $\sbdot$ and of
  $\sfbdot$, and that $\fdot$ is a special case of $\sfdot$ and of
  $\sbfdot$. Conversely, $\svlef$ and $\svrig$ are just sequences of
  $\vlef$ and $\vrig$, respectively, and $\svblef$ and $\svdrig$ are
  obtained by composing with $\lef\wbox$ and $\rig\wdia$,
  respectively. Then $\sbdot$ and $\sfdot$ are just sequences of
  $\bdot$ and $\fdot$, respectively, whereas $\sfbdot$ and $\sbfdot$
  use both $\bdot$ and $\fdot$.
  %%\qed
\end{proof}

\begin{lemma}\label{lem:4-over-b}
  Let $\tuple{\Xax,\Yax}$ be a safe pair of axioms. If $\vax\in\Xax$,
  then the rules $\svlef$ and $\svrig$ permute over any $\rr\in\sYdot$.
\end{lemma}

\begin{proof}
  We show here how $\svrig$ permutes over $\sbdot$. There are two nontrivial interactions:
  \begin{equation*}%\small
    \svlderivation{
      \vlin{\svrig}{}{\Gcons{\rig{\wdia A} , \Scons{\emptyset},\Dcons{\emptyset}}}{
        \vlin{\sbdot}{}{\Gcons{\Scons{\rig{\wdia A}},\Dcons{\emptyset}}}{
          \vlhy{\Gcons{\Dcons{\wbrn{\Scons{\rig{\wdia A}}}n}}}}}}
    \qquato
    \svlderivation{
      \vlin{\sbdot}{}{\Gcons{\rig{\wdia A} , \Scons{\emptyset},\Dcons{\emptyset}}}{
        \vlin{\svrig}{}{\Gcons{\rig{\wdia A},\Dcons{\wbrn{\Scons{\emptyset}}n}}}{
          \vlhy{\Gcons{\Dcons{\wbrn{\Scons{\rig{\wdia A}}}n}}}}}}
  \end{equation*}
  and
  \begin{equation*}%\small
    \svlderivation{
      \vlin{\svrig}{}{\Gcons{\rig{\wdia A} , \Sigma,\Dcons{\emptyset}}}{
        \vlin{\sbdot}{}{\Gcons{\Sigma,\Delta'\cons{\emptyset}}}{
          \vlhy{\Gcons{\Delta'\cons{\wbrn{\Sigma}n}}}}}}
    \qquato
    \svlderivation{
      \vlin{\sbdot}{}{\Gcons{\rig{\wdia A} , \Sigma,\Dcons{\emptyset}}}{
        \vlin{\svrig}{}{\Gcons{\rig{\wdia A},\Dcons{\wbrn{\Sigma}n}}}{
          \vlhy{\Gcons{\Delta'\cons{\wbrn{\Sigma}n}}}}}}
  \end{equation*}
  The other cases are similar. %%\qed
\end{proof}

When the rules $\vlef$ and $\vrig$ are present our cut rule, shown in \eqref{eq:cutr}, is not strong enough for our induction to work. Therefore we additionally use the following two rules
\begin{equation}\small
  \label{eq:diaboxcut}
  \vliinf{\wdiacutr}{}{\Gcons{\lef\Theta\cons{\emptyset}}}{
    \ddGamma\cons{\lef\Theta\cons{\prig{\wdia A}}}}{
    \Gcons{\plef{\wdia A},\lef\Theta\cons{\emptyset}}}
%%  \qquad
  %% \vliinf{\wdiacutr}{}{\Gcons{\rig\Theta\cons{\emptyset}}}{
  %%   \Gcons{\prig{\wdia A}}}{
  %%   \Gcons{\plef{\wdia A},\rig\Theta\cons{\emptyset}}}
  \;\quand\;
%\end{equation*}
%and
%\begin{equation*}
  \vliinf{\wboxcutr}{}{\Gcons{\Theta\cons{\emptyset}}}{
    \ddGamma\cons{\prig{\wbox A},\ddown{(\Theta\cons{\emptyset})}}}{
    \Gcons{\Theta\cons{\plef{\wbox A}}}}
  %% \vliinf{\wboxcutr}{}{\Gcons{\lef\Theta\cons{\emptyset}}}{
  %%   \ddGamma\cons{\prig{\wbox A},\lef\Theta\cons{\emptyset}}}{
  %%   \Gcons{\lef\Theta\cons{\plef{\wdia A}}}}
  %% \qquad
  %% \vliinf{\wboxcutr}{}{\Gcons{\rig\Theta\cons{\emptyset}}}{
  %%   \Gcons{\prig{\wbox A}}}{
  %%   \Gcons{\rig\Theta\cons{\plef{\wbox A}}}}
\end{equation} 
which are just combinations of $\cutr$ with towers of $\vrig$ and
$\vlef$, respectively. More precisely: 

\begin{fakt}\label{fact:4cut}
  The rule $\wdiacutr$ is derivable in $\set{\cutr,\svrig}$ and in $\set{\cutr,\vrig}$, and the
  rule $\wboxcutr$ is derivable in $\set{\cutr,\svlef}$ and in $\set{\cutr,\vlef}$.
\end{fakt}

By $\Cutr$, we refer to the set
$\set{\cutr,\wdiacutr,\wboxcutr}$ or $\set{\cutr}$, depending on
whether $\vlef$ and $\vrig$ are present or not, and we write
$\ast\cutr$ for any variant in~$\Cutr$.
Throughout this section we fix the convention that, for any $\ast\cutr$ step, the output cut formula occurs in the left premise, while the input cut formula occurs in the right premise.

\begin{definition}
For a formula $A$ we define $\depth{A}$
  inductively as follows:
  \begin{align*}
   & \depth{a}=\depth{\bot}=1 \qqqqquad
    \depth{\wbox A}=\depth{\wdia A}=\depth{A}+1\\
   & \depth{A\cand B}=\depth{A\cor B}=\depth{A\implies B}=
    \max(\depth{A},\depth{B})+1
  \end{align*} 
  %% \begin{align*}
  %%   \depth{a}=\depth{\bot}&=1\\
  %%   \depth{\wbox A}=\depth{\wdia A}&=\depth{A}+1\\
  %%   \depth{A\cand B}=\depth{A\cor B}=\depth{A\implies B}&=
  %%   \max(\depth{A},\depth{B})+1
  %% \end{align*}
  Given a $\cutr$ step, as shown in~\eqref{eq:cutr}, its
  \emph{cut formula} is $A$, and its \emph{rank} is $\depth{A}$.
\end{definition} 
Since the rules $\wdiacutr$ and $\wboxcutr$ can
be seen as derivations consisting of one instance of $\cutr$ and some
instances of $\vrig$ and $\vlef$, respectively, the definition of rank
also applies to $\wdiacutr$ and $\wboxcutr$ given
in~\eqref{eq:diaboxcut}. We use this convention throughout this
section: whenever we define a notion for an instance of $\cutr$, this
definition also applies for $\wdiacutr$ and $\wboxcutr$ because there
is a unique instance of $\cutr$ contained in them.

\begin{definition}\label{dfn:black-destructing}
  The inference rules $\idr$, $\lef\bot$, $\lef\cand$, $\lef\cor$,
  $\lef\implies$, $\lef\wbox$, $\lef\wdia$, $\tlef$, and $\svblef$ are called
  \emph{black destructing}. The \emph{principal formula} of a black
  destructing rule instance is the input formula singled out in its
  conclusion in Figures \ref{fig:NCK}, \ref{fig:dt4} and~\ref{fig:s4}.
\end{definition}

In other words a rule instance is black destructing if, considered bottom-up, it decomposes
an input formula along its main connective, and that formula is its
principal formula. In particular, note that $\vlef$ and $\svlef$ are not black
destructing.

\begin{definition}\label{dfn:cut-anchored+value}
  An instance of $\cutr$ is \emph{anchored} if the rule immediately
  above it on the right is a black-destructing rule whose principal
  formula is the cut formula. 
  We define the
  \emph{value} of a $\cutr$-instance to be the pair $\tuple{r,s}$, where $r$ is its rank, and $s=0$ if it is anchored and $s=1$ if it is not anchored. The \emph{value} of an instance of $\wdiacutr$ or $\wboxcutr$ is the value of the underlying $\cutr$-instance (if we read the $\wdiacutr$/$\wboxcutr$ as composition of $\cutr$ and $\svrig$/$\svlef$).
Finally, the \emph{cut-value} of a derivation $\DD$, denoted by $\cv{\DD}$ is the
  multiset of the values of its $\cutr$-instances.
\end{definition}

We order cut values lexicographically, i.e., 
$$\tuple{r_1,s_1}<\tuple{r_2,s_2} \qquiff r_1<r_2\qquor r_1=r_2 \quand
s_1<s_2$$ Then, multisets of cut values are ordered via a common
multiset ordering: Given an ordered set $\tuple{V,<}$, let $\cM(V)$ be
the set of multisets of elements of $V$, and let $M_1,M_2\in\cM(V)$ be
two such multisets. We define $M_1\ll M_2$ iff there is a multiset
surjection $f\colon M_1\to M_2$ such that for all $v\in M_2$, we either have $f^{-1}
(v) = \set{v}$ or $\forall u\in f^{-1} (v).\;u<v$. 

\begin{fakt}
  If $<$ is a strict total order, then so is $\ll$. Furthermore, if $<$ is well-founded, then so is $\ll$ \cite{derschowitz:nachum:multisets}.
\end{fakt}

\begin{example}
  If we let $V=\Nat$, with the usual ordering, then we have, for
  example, $\set{1,2,3,4,4,5}\ll\set{2,5,5}$ and
  $\set{1,1,2,2,2,2}\ll\set{1,2,3}$, where in the first case we can
  map $2$ to $2$, and $1,3,4,4$ all to one $5$, and $5$ to the other
  $5$. We could also map $1$ to $2$, and $2,3,4,4$ to one $5$, and $5$
  to the other $5$. These choices will be associated with reductions
  in our later cut-elimination arguments.
\end{example}

This gives us a well-order $\ll$ on the cut-values of a derivation,
and our cut-reduction proceeds by an induction on this well-order. For simplicity, we always
consider a topmost cut. 
There are two main lemmas, one for reducing
anchored cuts (Lemma~\ref{lem:one-step-anchored}), and one for
reducing cuts that are not anchored
(Lemma~\ref{lem:make-anchored}). For both of these lemmas we need, as
is often the case, height preserving admissibility and invertibility of certain
inference rules.

\begin{definition}
  The
  \emph{height} of a derivation $\DD$, denoted by $\height{\DD}$, is
  defined to be the length of a maximal branch in the derivation
  tree.
  We say that a rule $\rr$ with one premise is \emph{height preserving
    admissible} in a system $\sysS$, if for each derivation $\DD$ in
  $\sysS\setminus\set{\rr}$ of $\rr$'s premise there is a derivation $\DD'$ of $\rr$'s
  conclusion in $\sysS\setminus\set{\rr}$.
  Similarly, a rule $\rr$ is \emph{height preserving invertible} in a
  system $\sysS$, if for every derivation of the conclusion of $\rr$
  there are derivations for each of $\rr$'s premises with at most the
  same height. 
\end{definition}

\begin{proposition}\label{prop:admiss}
  Let $\tuple{\Xax,\Yax}$ be a safe pair of axioms. Then all rules in $\Xdot$, as well as 
  the rules $\weakr$ and $\necr$ are
  height preserving admissible for 
  $\NCKp\cup\sXlefrig\cup\sYdot$. 
\end{proposition}

\begin{proof}
For $\weakr$ and $\necr$ this is a straightforward induction on the height of the derivation.
For $\tdot$ and $\vdot$ we permute steps upwards through the proof to show admissibility, preserving height of the other rules in each reduction. Notice that, for either step, any nontrivial overlap with a rule above must have a bracket in its conclusion. 
For $\tdot$ we have the following nontrivial cases:\looseness=-1
\begin{enumerate}
\item\label{bdot-tdot} $\sbdot-\tdot$. The only overlap possible is in the $\Delta$ part of a $\sbdot$-step, so let $\Delta\conhole = \Delta_1 \cons{\wbr{ \Delta_2 \conhole} }$ with $\depth{\Delta_2} = m$ and $\depth{\Delta_1} = n$, and the permutation is as follows:
%\[
%\svlderivation{
%\vlin{\tdot}{}{\Gamma\{\Sigma,\Delta \}}{
%\vlin{\bdot}{}{\Gamma\{ \Sigma , \wbr{ \Delta } \} }{\vlhy{ \wbr{ \wbr{\Sigma} , \Delta }  }}
%}
%}
%\quad\to\quad
%\svlderivation{
%\vlin{\tdot}{}{\Gamma\{\Sigma,\Delta \}}{
%\vlin{\tdot}{}{\Gamma\{ \wbr{\Sigma,\Delta}  \}}{\vlhy{ \wbr{ \wbr{\Sigma} , \Delta }  }}
%}
%}
%\]
\[%\small
\svlderivation{
\vlin{\tdot}{}{\Gcons{\Sigma , \Delta_1 \cons{ \Delta_2 \cons{\emptyset} }  }}{
\vlin{\sbdot}{}{\Gcons{\Sigma , \Delta_1 \cons{[ \Delta_2 \cons{\emptyset}] }  }  }{\vlhy{ \Gcons{ \Delta_1 \cons{ [ \Delta_2\cons{ \wbrn{\Sigma}{m+n+1} }  ] } } }}
}
}
\quad\to\quad
\svlderivation{
\vlin{\sbdot}{}{ \Gcons{\Sigma , \Delta_1 \cons{ \Delta_2 \cons{\emptyset} }  } }{
\vlin{\tdot}{}{ \Gcons{ \Delta_1 \cons{ \Delta_2 \cons{\wbrn{\Sigma}{m+n} } }  } }{
\vlin{\tdot}{}{ \Gcons{ \Delta_1 \cons{ \Delta_2 \cons{\wbrn{\Sigma}{m+n+1} } }   }}{\vlhy{ \Gcons{ \Delta_1 \cons{ [ \Delta_2\cons{ \wbrn{\Sigma}{m+n+1} }  ] } }  }}
}
}
}
\]
and we can apply the induction hypothesis twice.
\item\label{vrig-tdot} $\svdrig-\tdot$.
\[%\small
\svlderivation{
\vlin{\tdot}{}{\Gcons{ \rig{\wdia A} , \Dcons{\emptyset} }}{
\vlin{\svdrig}{}{ \Gcons{ \rig{\wdia A} , \wbr{\Dcons{\emptyset}} } }{\vlhy{ \Gcons{ \wbr{ \Dcons{ \rig A } } } }}
}
}
\quad\to\quad
\svlderivation{
\vlin{\mathsf r}{}{ \Gcons{ \rig{\wdia A} , \Dcons{\emptyset} } }{
\vlin{\tdot}{}{ \Gcons{ \Dcons{ \rig A } } }{\vlhy{  \Gcons{ \wbr{ \Dcons{ \rig A } } }  }}
}
}
\]
where $\mathsf r$ is $\trig$ if the hole of $\Delta\conhole$ has depth $0$ and $\svdrig$ otherwise. 
%\[
%\svlderivation{
%\vlin{\tdot}{}{\Gamma\{ \rig{\wdia A} , \Delta  \}  }{
%\vlin{\vrig}{}{\Gamma\{ \rig{\wdia A} , \wbr{\Delta}  \}  }{\vlhy{ \Gamma\{ \wbr{ \rig{\wdia A } , \Delta } \}  }}
%}
%}
%\quad\to\quad
%\vlinf{\tdot}{}{ \Gamma\{ \rig{\wdia A} , \Delta  \}   }{ \Gamma\{ \wbr{ \rig{\wdia A } , \Delta } \}  }
%\]
\item\label{vlef-tdot} $\svblef-\tdot$. Similar to case \ref{vrig-tdot}.
\item\label{rigdia-tdot} $\rig\wdia-\tdot$. 
\[%\small
\svlderivation{
\vlin{\tdot}{}{\Gamma\{ \rig{ \wdia A } , \Delta  \}  }{
\vlin{\rig\wdia}{}{\Gamma\{ \rig{\wdia A} , \wbr{\Delta} \}  }{\vlhy{ \Gamma\{\wbr{ \rig A , \Delta }  \} }}
}
}
\quad\to\quad
\svlderivation{
\vlin{\trig}{}{ \Gamma\{ \rig{ \wdia A } , \Delta  \}   }{
\vlin{\tdot}{}{\Gamma\{\rig A , \Delta  \} }{\vlhy{ \Gamma\{\wbr{ \rig A , \Delta }  \}  }}
}
}
\]
\item\label{lefbox-tdot} $\lef\wbox - \tdot$. Similar to case \ref{rigdia-tdot}. 
\item\label{sfbdot-tdot} $\sfbdot - \tdot$. Similar to case \ref{bdot-tdot}.
\item\label{sbfdot-tdot} $\sbfdot - \tdot$. One overlap case is similar to case \ref{bdot-tdot}, and the other is given below.
\[%\small
\svlderivation{
\vlin{\tdot}{}{ \Gcons{ \Sigma , \Dcons{ \emptyset } } }{
\vlin{\sbfdot}{}{ \Gcons{ \wbr{\Sigma} , \Dcons{ \emptyset }  }}{\vlhy{ \Gcons{ \Dcons{ \wbrn{ \Sigma }k } } }}
}
}
\quad\to\quad
\vlinf{\sfbdot}{}{ \Gcons{ \Sigma , \Dcons{ \emptyset } } }{ \Gcons{ \Dcons{ \wbrn{ \Sigma }k } } }
\]

\item\label{sfdot-tdot} $\sfdot - \tdot$. One overlap case is similar to case \ref{bdot-tdot}, and the other is similar to \ref{sbfdot-tdot} above.

\end{enumerate} 
And for $\vdot$ we have the following nontrivial cases:
\begin{enumerate} 
\setcounter{enumi}{8}
\item\label{sbdot-vdot} $\sbdot - \vdot$. The only overlap possible is in the $\Delta$ part of a $\sbdot$-step, so let $\Delta\conhole = \Delta_1 \cons{\wbr{ \Delta_2 \conhole  } }$ with $\depth{\Delta_2} = m$ and $\depth{\Delta_1} = n$, and the permutation is as follows:
\[%\small
\svlderivation{
\vlin{\vdot}{}{ \Gamma\{\Sigma , \Delta_1 \{ \wbr{\wbr{ \Delta_2 \{\emptyset\} } } \}  \} }{
\vlin{\sbdot}{}{ \Gamma\{\Sigma , \Delta_1 \{ \wbr{ \Delta_2 \{\emptyset\} } \}  \}}{\vlhy{ \Gamma\{\Delta_1 \{ \wbr{ \Delta_2 \{ \wbrn{\Sigma }{m+n+1} \} } \} \} }}
}
}
\quad \to \quad
\svlderivation{
\vlin{\sbdot}{}{  \Gamma\{\Sigma , \Delta_1 \{ \wbr{\wbr{ \Delta_2 \{\emptyset\} } } \}  \}  }{
\vlin{\vdot}{}{ \Gamma\{\Delta_1 \{ \wbr{\wbr{ \Delta_2 \{ \wbrn{\Sigma }{m+n+2} \} } } \} \}  }{
\vlin{\vdot}{}{  \Gamma\{\Delta_1 \{ \wbr{\wbr{ \Delta_2 \{ \wbrn{\Sigma }{m+n+1	} \} } } \} \}  }{\vlhy{ \Gamma\{\Delta_1 \{ \wbr{ \Delta_2 \{ \wbrn{\Sigma }{m+n+1} \} } \} \} }}
}
}
}
\]
and we can apply the induction hypothesis twice.
\item\label{vrig-vdot} $\svdrig - \vdot$. 
\[%\small
\svlderivation{
\vlin{\vdot}{}{\Gcons{ \rig{\wdia A} ,\wbr{\wbr{  \Dcons{\emptyset} }} }}{
\vlin{\svdrig}{}{ \Gcons{ \rig{\wdia A} , \wbr{\Dcons{\emptyset}} } }{\vlhy{ \Gcons{ \wbr{ \Dcons{ \rig A } } } }}
}
}
\quad\to\quad
\svlderivation{
\vlin{\svdrig}{}{\Gcons{ \rig{\wdia A} ,\wbr{\wbr{  \Dcons{\emptyset} }} }}{ 
  \vlin{\vdot}{}{\Gcons{ \wbr{  \wbr{ \Dcons{ \rig A } }} }}{\vlhy{\Gcons{ \wbr{ \Dcons{ \rig A } } }}}}}
\]
\item\label{vlef-vdot} $\svblef - \vdot$. Similar to case \ref{vrig-vdot}.
\item\label{rigdia-vdot} $\rig\wdia-\vdot$. 
\[%\small
\svlderivation{
\vlin{\vdot}{}{\Gcons{ \rig{\wdia A} , \wbr{\wbr{ \Delta }} }}{
\vlin{\rig\wdia}{}{\Gcons{ \rig{\wdia A} , \wbr{ \Delta } }}{\vlhy{ \Gcons{ \wbr{ \rig A , \Delta } } }}
}
}
\quad\to\quad
\svlderivation{
\vlin{\svdrig}{}{ \Gcons{ \rig{\wdia A} , \wbr{\wbr{ \Delta }} } }{
\vlin{\vdot}{}{  \Gcons{ \wbr{\wbr{ \rig A , \Delta } }}  }{\vlhy{ \Gcons{ \wbr{ \rig A , \Delta } } }}
}
}
\]
\item\label{lefbox-vdot} $\lef\wbox - \vdot$. Similar to case \ref{rigdia-vdot}.

%% \[\small
%% \svlderivation{
%% \vlin{\tdot}{}{\Gamma\{\wbr{\Sigma} \}}{
%% \vlin{\vdot}{}{\Gamma\{\wbr{\wbr{\Sigma}} \} }{\vlhy{\Gamma\{\wbr{\Sigma} \} }}
%% }
%% }
%% \quad \to \quad
%% \svlderivation{
%% \Gamma\{\wbr{\Sigma} \}
%% }
%% \]

\item\label{sfbdot-vdot} $\sfbdot - \vdot$. Similar to case \ref{sbdot-vdot}.
\item\label{sbfdot-vdot} $\sbfdot - \vdot$. One overlap case is similar to \ref{sbdot-vdot} and the other is given below.
\[%\small
\svlderivation{
\vlin{\vdot}{}{ \Gcons{ \wbr{\wbr{\Sigma}}, \Dcons{ \emptyset } } }{
\vlin{\sbfdot}{}{ \Gcons{ \wbr{\Sigma} , \Dcons{ \emptyset }  }}{\vlhy{ \Gcons{ \Dcons{ \wbrn{ \Sigma }k } } }}
}
}
\quad\to\quad
\svlderivation{
\vlin{\sbfdot}{}{ \Gcons{  \wbr{\wbr{\Sigma}} , \Dcons{ \emptyset } } }{
\vlin{\vdot}{}{ \Gcons{ \Dcons{ \wbrn{ \wbr{\Sigma} }k } }}{\vlhy{  \Gcons{ \Dcons{ \wbrn{ \Sigma }k } } }}
}
}
\]

\item\label{sfdot-vdot} $\sfdot - \vdot$. One overlap case is similar to \ref{sbdot-vdot}, and the other is similar to \ref{sbfdot-vdot} above.

\end{enumerate}
Note also that permutations over contraction preserve height, since we can apply the induction hypothesis twice.
%%\qed
\end{proof}

Note that the variants $\sXlefrig$ of $\Xlefrig$ and
$\sYdot$ of $\Ydot$ are needed to make
Proposition~\ref{prop:admiss} work. Without the ``super-rules'' we
would not be able to preserve the height, and consequently would not
be able to proceed by the induction hypothesis when
eliminating $\tdot$ and~$\vdot$ in the cases \ref{bdot-tdot} and~\ref{sbdot-vdot} above.

\begin{proposition}\label{prop:invert}
  The rules $\lef\cand$, $\lef\cor$, $\lef\wdia$, $\rig\cand$,
  $\rig\implies$, $\rig\wbox$, and $\lef\implies$ on the right premise, are 
  height preserving invertible for
  $\NCKp\cup\sXlefrig\cup\sYdot$.
%% \cup\Cutr$.\ryuta{Same. This Cut 
%%       should not be here.} 
\end{proposition}

\begin{proof}
  Straightforward induction on the
height of the derivation. 
\end{proof}

Before we can state our main lemmas, we need to define a restricted
version of Buss' \emph{logical flow-graphs}~\cite{buss:91}.

\newcommand{\flg}[1]{G(#1)}

\begin{definition}
  We define the \emph{(formula) flow-graph} of a derivation $\DD$,
  denoted by $\flg{\DD}$ to be the directed graph whose vertices are
  all input formula occurrences in $\DD$, and whose edges are just between
  two formula occurrences which are the same unaltered occurrence in the
  premise and conclusion of an instance of an inference rule. This
  concerns all formula occurrences in $\Gamma\conhole$, $\Delta$,
  $\Pi$, and $\Sigma$ in the rules in Figures \ref{fig:NCK},
  \ref{fig:dt4}, \ref{fig:CXstr},~\ref{fig:s4} and~$\ast\cutr$, as
  well as the occurrences of $\lef{\wbox A}$ in the $\vlef$ and
  $\svlef$ rules. The edges are always directed from premise to
  conclusion.  The \emph{length} of a path in $\flg{\DD}$ is its
  number of edges. A path $p$ in $\flg{\DD}$ is
  \emph{maximal} if for every path $p'$ in $\flg{\DD}$ with $p\subseteq p'$ we have
  $p=p'$.
\end{definition}
Let us emphasize that there are no edges between a formula occurrence and any of its subformulae that may occur in
$\flg{\DD}$. For example, the principal $\lef{A\cor B}$ in
the conclusion of an $\lef\cor$-rule is connected to neither the $\lef A$
nor the $\lef B$ in the premises. But every formula occurrence in
$\Gcons{\conhole,\rig\Pi}$ in the conclusion is connected via an edge
to the same occurrence in each of the two premises.  Thus, the
flow-graph is essentially a set of trees, where branching occurs in
the branching rules $\lef\cor$, $\lef\implies$, $\rig\cand$, $\ast\cutr$, and in a
contraction because every formula occurrence in $\lef\Delta$ in the
conclusion is connected to each of its copies in the premise.

Recall from Definition~\ref{dfn:black-destructing} the notion of a black-destructing rule, from Definition~\ref{dfn:cut-anchored+value} the notion of an anchored cut, and our convention that an output cut formula is written on the left-hand side of a $\ast\cutr$ step and an input cut formula on the right.

\begin{definition}  
    A \emph{cut-path} in $\flg{\DD}$ is a maximal path that ends at the
  cut formula $\lef A$ in the right-hand side premise of a
  $\ast\cutr$-instance.  
  A cut path is \emph{relevant} if it starts at the
  principal formula of a black destructing rule. Otherwise it is
  called \emph{irrelevant}.  A cut path is \emph{left-free} if it
  never passes through a left-hand side premise of an instance of
  $\ast\cutr$. A derivation $\DD$ is \emph{left-free} if all relevant cut
  paths in $\flg{\DD}$ are left-free.  An \emph{origin} of
  $\flg{\DD}$ is the topmost vertex of a relevant cut path in
  $\flg{\DD}$. An origin is \emph{anchored} if its cut path has
  length~0. A derivation is \emph{anchored} if all its cuts are anchored. 
  An instance of $\ast\cutr$ in $\DD$ is called
  \emph{relevant} if it has at least one relevant cut path. Otherwise
  it is called \emph{irrelevant}.  The \emph{relevant cut-value} of a
  derivation $\DD$, denoted by $\rcv{\DD}$, is the multiset of the
  values of its relevant cuts. 
\end{definition}

 To be clear, irrelevant cut-paths are exactly those that begin in the context of an axiom, i.e.~ in the $\Gamma\conhole$ part of a $\lef\bot$ or $\idr$ step.

Notice that we are using the term `anchored' to describe both cuts, as
in Definition~\ref{dfn:cut-anchored+value}, and origins as in the
definition above (as well as derivations). In particular we point out
that, if a cut is anchored, then it can have only one origin
which is also anchored. Conversely, if all origins are anchored (which
are only defined for relevant cut-paths), there may be some cuts that
are not anchored in the derivation, namely those with only irrelevant
cut-paths.  An anchored derivation, thus, is one all of whose cuts and
origins are anchored, which is not the same as simply having all
origins anchored. This subtlety is important in the proof of
Lemma~\ref{lem:make-anchored} below. But first, let us make an example.

\begin{example}  
    Consider the derivation: \proofadjust
\[
\svlderivation{
  \vliin{\cutr}{}{\lef\bot,\rig{b\cor c}}{
    \vlhtr{}{\lef\bot,\rig A}}{
    \vliin{\cutr}{}{\lef\bot, \lef A , \rig{b\cor c} }{
      \vliin{\cutr}{}{\lef\bot , \lef A, \rig{d\cand b  } }{
        \vlin{\lef\bot}{}{ \lef\bot,\lef A , \rig E}{
          \vlhy{}}}{ 
        \vlin{\lef\bot}{}{ \lef\bot,\lef A , \lef E , \rig{ d\cand b } }{
          \vlhy{}}}}{
      \vlin{\rig\cor}{}{ \lef\bot, \lef A , \lef{ d\cand b },\rig{b\cor c} }{
        \vlin{\lef\cand }{}{ \lef\bot, \lef A , \lef{ d\cand b },\rig{b} }{ 
        \vlin{\idr}{}{\lef\bot, \lef A , \lef d , \lef b , \rig{b}}{
          \vlhy{}}}}}}}
\]
Here the cut-paths for $\lef{A}$ and $\lef{E}$ are irrelevant, 
while the cut-path for $\lef{d \cand b}$ is relevant.  
There are three cut-paths for $\lef{A}$ and, except for the rightmost one, 
they do not satisfy left-freeness, since 
they pass through the left premise of a cut instance. 
The only cut-path for $\lef{d\cand b}$ 
has two vertices and length 1. Therefore, this cut is not anchored.  
But if we permute that cut over the $\rig\cor$-rule instance, we obtain the derivation
\[
\svlderivation{
  \vliin{\cutr}{}{\lef\bot,\rig{b\cor c}}{
    \vlhtr{}{\lef\bot,\rig A}}{
    \vlin{\rig\cor}{}{ \lef\bot, \lef A ,\rig{b\cor c} }{
      \vliin{\cutr}{}{\lef\bot, \lef A , \rig{b} }{
        \vliin{\cutr}{}{ \lef\bot, \lef A , \rig{ d\cand b } }{
          \vlin{\lef\bot}{}{ \lef\bot,\lef A , \rig E}{
            \vlhy{}}}{ 
          \vlin{\lef\bot}{}{ \lef\bot,\lef A , \lef E , \rig{ d\cand b } }{
            \vlhy{}}}}{
        \vlin{\lef\cand }{}{ \lef\bot, \lef A , \lef{ d\cand b },\rig{b} }{ 
          \vlin{\idr}{}{\lef\bot, \lef A , \lef d , \lef b , \rig{b}}{
            \vlhy{}}}}}}}
\]
in which the cut-path for $\lef{d\cand b}$ has length 0, and so this cut is anchored. 
\end{example}

\begin{lemma}\label{lem:irrelevant}
  Let $\tuple{\Xax,\Yax}$ be a safe pair of axioms.  Given a derivation
  $\DD$ in $\NCKp+\sXlefrig+\sYdot+\Cutr$, there is a derivation
  $\DD'$ in $\NCKp+\sXlefrig+\sYdot+\Cutr$ of the same conclusion, such that $\DD'$ has
  no irrelevant cuts, and such that $\rcv{\DD'}\le\rcv{\DD}$.
\end{lemma}

\begin{proof}
  We proceed by induction on the number of irrelevant cuts in~$\DD$.
  Consider the topmost one. We can replace
  $$
  \svlderivation{
    \vliin{\cutr}{}{\Gcons{\emptyset}}{
      \vlhtr{\DD_1}{\ddGamma\cons{\rig A}}}{
      \vlhtr{\DD_2}{\Gcons{\lef A}}}}
  \qquad\mbox{by}\qquad
  \svlderivation{
    \vlhtr{\DD'_2}{\Gcons{\emptyset}}}
  $$ where $\DD'_2$ is obtained from $\DD_2$ by removing the $\lef A$
  occurrence everywhere; this results in a correct derivation since, by irrelevance, $\lef A$ must occur in the context of an axiom. For $\wdiacutr$ and $\wboxcutr$ we proceed
  similarly. %%\qed
\end{proof}

In the following, we use the notation $n\mathord*\rr$ where $n$ is a natural
number and $\rr$ a name of an inference rule. Then $n\mathord*\rr$ simply
stands for $n$ consecutive applications of~$\rr$.

\begin{lemma}\label{lem:make-anchored}
  Let $\tuple{\Xax,\Yax}$ be a safe pair of axioms, and let $\DD$ be a
  left-free derivation in $\NCKp+\sXlefrig+\sYdot+\Cutr$. Then
  there is an anchored derivation $\DD'$ in
  $\NCKp+\sXlefrig+\sYdot+\Cutr$ of the same conclusion, such that for each $\ast\cutr$ in $\DD'$,
  there is a $\ast\cutr$ in $\DD$ of the same rank.
\end{lemma}

\begin{proof}%[Proof of Lemma~\ref{lem:make-anchored}] 
  We proceed by induction on the number of origins in 
   $\flg{\DD}$ that are not anchored. If all origins are 
   anchored, then we remove all irrelevant cuts 
   using Lemma~\ref{lem:irrelevant} and we are done. 
   Otherwise, we pick a topmost origin that is not anchored and proceed
   by an inner induction on the length of 
   its cut-path to show that there is a derivation in which
   the number of non-anchored origins has decreased. Note that, if the length of this cut-path is~$0$, then the cut is already anchored and there is nothing to do. 

Now consider
   the $\ast\cutr$-instance connected to our origin and make 
   a case analysis on the rule instance $\rr$ on 
   the right above~it. 
  \begin{enumerate}
  \item If $\rr$ is one of $\lef\cand$, $\lef\wdia$, $\rig\implies$,
    $\rig\wbox$, we can reduce as follows:
    \begin{equation}%\small
      \cutredcase{\Invr_1}{
        \svlderivation{
          \vliin{\cutr}{}{\Gamma\conempty}{
            \vlhtr{\DD_1}{\dGamcon{\rig A}}}{
            \vlin{\rr}{}{\Gamcon{\lef A}}{
              \vlhtr{\DD_2}{\Gamma_1\cons{\lef A}}}}
      }\qquad}{\qquad
        \svlderivation{
          \vlin{\rr}{}{\Gamma\conempty}{
            \vliin{\cutr}{}{\Gamma_1\conempty}{
              \vlin{\Invr_\rr}{}{\down{\Gamma_1}\cons{\rig A}}{
                \vlhtr{\DD_1}{\down{\Gamma}\cons{\rig A}}}}{
              \vlhtr{\DD_2}{\Gamma_1\cons{\lef A}}}}
      }}
    \end{equation} 
    where the $\Invr_\rr$ is eliminated by Proposition~\ref{prop:invert}.
  \item If $\rr$ is one of $\lef\cor$, $\rig\cand$, we can reduce as follows:
    \begin{equation*}%\small\qquad
      \cutredcase{\rr}{\qquad
        \svlderivation{
          \vliin{\cutr}{}{\Gamma\conempty}{
            \vlhtr{\DD_1}{\dGamcon{\rig A}}}{
            \vliin{\rr}{}{\Gamcon{\lef A}}{
              \vlhtr{\DD_2}{\Gamma_1\cons{\lef A}}}{
              \vlhtr{\DD_3}{\Gamma_2\cons{\lef A}}}}
      }}{
        \svlderivation{
          \vliin{\rr}{}{\Gamma\conempty}{
            \vliin{\cutr}{}{\Gamma_1\conempty}{
              \vlin{\Invr_\rr}{}{\down{\Gamma_1}\cons{\rig A}}{
                \vlhtr{\DD_1}{\down{\Gamma}\cons{\rig A}}}}{
              \vlhtr{\DD_2}{\Gamma_1\cons{\lef A}}}}{
            \vliin{\cutr}{}{\Gamma_2\conempty}{
              \vlin{\Invr_\rr}{}{\down{\Gamma_2}\cons{\rig A}}{
                \vlhtr{\DD_1}{\down{\Gamma}\cons{\rig A}}}}{
              \vlhtr{\DD_3}{\Gamma_2\cons{\lef A}}}}
      }}
    \end{equation*}
    where the $\Invr_\rr$ steps are eliminated by
    Proposition~\ref{prop:invert}. 
  \end{enumerate}
  Note that it can happen in these two cases that the
  $\Invr_\rr$-step is vacuous in the above because
  depending on the position of the output formula in
  $\Gamma\conhole$ it is possible that
  $\ddown{\Gamma_1}\conhole=\ddGamma\conhole$ and
  $\ddown{\Gamma_2}\conhole=\ddGamma\conhole$.
  \begin{enumerate}\setcounter{enumi}{2}
\item If $\rr$ is $\lef\implies$, there are two cases. 
        The first is
    \begin{equation*}%%\small%\footnotesize
      %cutredcaseb
      \cutredcasec
	[8ex]
	{\lef\implies_1}
      {\hfill
        \svlderivation{
          \vliin{\cutr}{}{\Gamcon{\Tcons{\plef{B\implies C}},\lefDcons\emptyset}}{
            \vlhtr{\DD_1}{\dGamcon{\ddown{(\Tcons{\plef{B\implies C}})},\lefDcons{\rig A}}}}{
            \vliin{\lef\implies}{}{\Gamcon{\Tcons{\plef{B\implies C}},\lefDcons{\lef A}}}{
              \vlhtr{\DD_2}{\dGamcon{\ddTcons{\rig B},\lefDcons{\lef A}}}}{
              \vlhtr{\DD_3}{\Gamcon{\Tcons{\lef C},\lefDcons{\lef A}}}}}
      }}{
        \svlderivation{
          \vlin{\conr}{}{\Gamcon{\Tcons{\plef{B\implies C}},\lefDcons\emptyset}}{
            \vliin{\lef\implies}{}{\Gamcon{\lhs{(\Tcons{\plef{B\implies C}})},\Tcons{\plef{B\implies C}},\lefDcons\emptyset}}{
            \vlhtr{\DD_4}{\dGamcon{\lhs{(\Tcons{\plef{B\implies C}})},\ddTcons{\rig B},\lefDcons\emptyset}}
%              \vliin{\cutr}{}{\dGamcon{\lhs{(\Tcons{\plef{B\implies C}})},\ddTcons{\rig B},\lefDcons\emptyset}}{
%                \vlhtr{\DD_1}{\dGamcon{\lhs{(\Tcons{\plef{B\implies C}})},\lefDcons{\rig A}}}}{
%                \vlin{\weakr}{}{\dGamcon{\lhs{(\Tcons{\plef{B\implies C}})},\ddTcons{\rig B},\lefDcons{\lef A}}}{
%                  \vlhtr{\DD_2}{\dGamcon{\ddTcons{\rig B},\lefDcons{\lef A}}}}}
                  }{
                  %\hskip10em
                  \vlhtr{\DD_5}{\Gamcon{\lhs{(\Tcons{\plef{B\implies C}})},\Tcons{\lef C},\lefDcons\emptyset}}
%              \vliin{\cutr}{}{\Gamcon{\lhs{(\Tcons{\plef{B\implies C}})},\Tcons{\lef C},\lefDcons\emptyset}}{
%                \vlin{\weakr}{}{\dGamcon{\lhs{(\Tcons{\plef{B\implies C}})},\lhs{(\Tcons{\plef{C}})},\lefDcons{\rig A}}}{
%                 % \vlin{\Invr}{}{\dGamcon{\lhs{(\Tcons{\plef{C}})},\lefDcons{\rig A}}}{
%                    \vlhtr{\DD_1}{\dGamcon{\lhs{(\Tcons{\plef{B\implies C}})},\lefDcons{\rig A}}}}}{
%                \vlin{\weakr}{}{\Gamcon{\lhs{(\Tcons{\plef{B\implies C}})},\Tcons{\lef C},\lefDcons{\lef A}}}{
%                  \vlhtr{\DD_3}{\Gamcon{\Tcons{\lef C},\lefDcons{\lef A}}}}}
                  }}
      }}%\quad
    \end{equation*} 
    where $\DD_4$ is
    \[
    \qquad
    \svlderivation{
    \vliin{\cutr}{}{\dGamcon{\lhs{(\Tcons{\plef{B\implies C}})},\ddTcons{\rig B},\lefDcons\emptyset}}{
                    \vlhtr{\DD_1}{\dGamcon{\lhs{(\Tcons{\plef{B\implies C}})},\lefDcons{\rig A}}}}{
                    \vlin{\weakr}{}{\dGamcon{\lhs{(\Tcons{\plef{B\implies C}})},\ddTcons{\rig B},\lefDcons{\lef A}}}{
                      \vlhtr{\DD_2}{\dGamcon{\ddTcons{\rig B},\lefDcons{\lef A}}}}}
    }
    \]
    and $\DD_5$ is
    \[
    \qquad
    \svlderivation{
     \vliin{\cutr}{}{\Gamcon{\lhs{(\Tcons{\plef{B\implies C}})},\Tcons{\lef C},\lefDcons\emptyset}}{
                    \vlin{\weakr}{}{\dGamcon{\lhs{(\Tcons{\plef{B\implies C}})},\lhs{(\Tcons{\plef{C}})},\lefDcons{\rig A}}}{
                     % \vlin{\Invr}{}{\dGamcon{\lhs{(\Tcons{\plef{C}})},\lefDcons{\rig A}}}{
                        \vlhtr{\DD_1}{\dGamcon{\lhs{(\Tcons{\plef{B\implies C}})},\lefDcons{\rig A}}}}}{
                    \vlin{\weakr}{}{\Gamcon{\lhs{(\Tcons{\plef{B\implies C}})},\Tcons{\lef C},\lefDcons{\lef A}}}{
                      \vlhtr{\DD_3}{\Gamcon{\Tcons{\lef C},\lefDcons{\lef A}}}}}
    }
    \quadfs
    \]
    From here, $\weakr$ steps are removed by Proposition~\ref{prop:admiss}. Left-freeness is preserved since the derivations initially on the right of the cut, $\DD_2$ and $\DD_3$, remain on the right of all cuts after the transformation. Finally, both of the new cuts have the same rank as the initial cut, satisfying the requirement in the statement of the lemma. (To see that the application of the $\conr$-rule is correct we refer to Observation~\ref{obs:context} and Definition~\ref{def:pruning}.)
    
    Note that this case shows
    that we need an explicit contraction rule. Making contraction
    implicit in the $\lef\implies$-rule (as done
    in~\cite{str:fossacs13}) would not be enough, since we also need
    to duplicate the context $\Theta\conhole$. 
    %% In this
    %% case we used $\Invr+\weakr$ instead of only $\weakr$ on purpose
    %% for maintaining the number of relevance points. \todo{maybe not
    %%  needed anymore. check!!}  
    
    In the case shown above, the output
    formula in the conclusion can be in $\Gamma\conhole$ or
    $\Theta\conhole$.  There is another such case for
    $\lef\implies$ on the right branch, where the output formula in
    the conclusion is in $\Delta\conhole$. That case is simpler, and
    no extra contraction is needed: 
    \begin{equation*}
%    %\small
    %\tiny 
    %\scriptsize %\footnotesize
%      \cutredcasea
      \cutredcasec
      [7.5ex]
      {\lef\implies_2}{
        %\hskip-17em
        \svlderivation{
          \vliin{\hskip5ex\cutr}{}{\lefGcons{\lefTcons{\plef{B\implies C}},\Dcons\emptyset}}{
            \vlhtr{\DD_1}{\lefGcons{\lefTcons{\plef{B\implies C}},\ddDcons{\rig A}}}}{
            \vliin{\lef\implies}{}{\lefGcons{\lefTcons{\plef{B\implies C}},\Dcons{\lef A}}}{
              \vlhtr{\DD_2}{\lefGcons{\lefTcons{\rig B}}}}{
              \vlhtr{\DD_3}{\lefGcons{\lefTcons{\lef C},\Dcons{\lef A}}}}}
      }}{\hskip14ex
        \svlderivation{
          \vliin{\lef\implies}{}{\lefGcons{\lefTcons{\plef{B\implies C}},\Dcons\emptyset}}{
            \vlhtr{\DD_2}{\lefGcons{\lefTcons{\rig B}}}}{
            \vliin{\cutr}{}{\lefGcons{\lefTcons{\plef{C}},\Dcons\emptyset}}{
              \vlin{\Invr_{\lef\implies}}{}{\lefGcons{\lefTcons{\plef{C}},\ddDcons{\rig A}}}{
                \vlhtr{\DD_1}{\lefGcons{\lefTcons{\plef{B\implies C}},\ddDcons{\rig A}}}}}{
              \vlhtr{\DD_3}{\lefGcons{\lefTcons{\lef C},\Dcons{\lef A}}}}}
      }}
    \end{equation*} 
    Here we use invertibility of $\lef\implies$ on the right
    (Proposition~\ref{prop:invert}).
  \item\label{c:r} If $\rr$ is one of $\lef\wbox$, $\rig\wdia$, $\tlef$, $\trig$,
    $\conr$, or one of the $\svax$-rules, working entirely in the
    context of the cut formula $\lef A$, then there are contexts
    $\Gamma' , \Gamma_1'$ such that we can reduce as follows:
    \begin{equation}%\small
      \cutredcase{\rr}{
        \svlderivation{
          \vliin{\cutr}{}{\Gamma\conempty}{
            \vlhtr{\DD_1}{\dGamcon{\rig A}}}{
            \vlin{\rr}{}{\Gamcon{\lef A}}{
              \vlhtr{\DD_2}{\Gamma_1\cons{\lef A}}}}
      }\quad}{\quad
        \svlderivation{
          \vlin{\conr}{}{\Gamma\conempty}{
            \vlin{\rr}{}{\Gamma'\conempty}{
              \vliin{\cutr}{}{\Gamma_1'\conempty}{
                \vlin{\weakr}{}{\down{\Gamma_1'}\cons{\rig A}}{
                  \vlhtr{\DD_1}{\down{\Gamma}\cons{\rig A}}}}{
                \vlin{\weakr}{}{\Gamma_1'\cons{\lef A}}{
                  \vlhtr{\DD_2}{\Gamma_1\cons{\lef A}}}}}}
      }}
    \end{equation}
    where, read top-down, the two $\weakr$-steps weaken as much as is necessary to
    unify the contexts in order to perform a $\cutr$-step. The $\rr$-step then
    acts on the appropriate redex (as determined by the $\rr$-step on the
    left) before contraction is performed to eliminate any formulae
    duplicated as a result of the permutation. The $\weakr$-steps are then removed by
    Proposition~\ref{prop:admiss}. 
  \item If $\rr$ is a $\svlef$ step moving the cut formula (which is of
    shape $\lef{\wbox A}$), then we can inductively apply,
    \begin{equation*}%\small\qquad
      \cutredcasec
      [6.5ex]
      {\wbox\vax}{
        \svlderivation{
          \vliin{\qquad\wboxcutr}{}{\Gcons{\Tcons{\Dcons{\emptyset}}}}{
            \vlhtr{\DD_1}{\ddGcons{\rig{\wbox A},\ddown{(\Tcons{\Dcons{\emptyset}})}}}}{
            \vlin{\vlef}{}{\Gcons{\Tcons{\lef{\wbox A},\Dcons{\emptyset}}}}{
              \vlhtr{\DD_2}{\Gcons{\Tcons{\Dcons{\lef{\wbox A}}}}}}}}
      }{
        \svlderivation{
          \vliin{\hskip18em\wboxcutr}{}{\Gcons{\Tcons{\Dcons{\emptyset}}}}{
            \vlhtr{\DD_1}{\ddGcons{\rig{\wbox A},\ddown{(\Tcons{\Dcons{\emptyset}})}}}}{
            \vlhtr{\DD_2}{\Gcons{\Tcons{\Dcons{\lef{\wbox A}}}}}}}
      }
    \end{equation*}
    by decomposing the instance of $\svlef$ into several $\vlef$ steps. 
  \item If $\rr$ is a $\conr$ step duplicating the cut formula, we can reduce as follows:
    \begin{equation}%\small
      \label{eq:conred-a}
      \cutredcasea{\conr}{
        \svlderivation{
          \vliin{\cutr}{}{\Gcons{\lefDcons{\emptyset}}}{
            \vlhtr{\DD_1}{\ddGcons{\lefDcons{\rig A}}}}{
            \vlin{\conr}{}{\Gcons{\lefDcons{\lef A}}}{
              \vlhtr{\DD_2}{\Gcons{\lefDcons{\lef A},\lefDcons{\lef A}}}}}}
      }{\hskip-8em
        \svlderivation{
          \vlin{\conr}{}{\Gcons{\lefDcons{\emptyset}}}{
            \vliin{\cutr}{}{\Gcons{\lefDcons{\emptyset},\lefDcons{\emptyset}}}{
              \vlin{\weakr}{}{\ddGcons{\lefDcons{\rig A},\lefDcons{\emptyset}}}{
                \vlhtr{\DD_1}{\ddGcons{\lefDcons{\rig A}}}}}{
              \vliin{\cutr}{}{\Gcons{\lefDcons{\lef A},\lefDcons{\emptyset}}}{
                \vlin{\weakr}{}{\ddGcons{\lefDcons{\lef A},\lefDcons{\rig A}}}{
                  \vlhtr{\DD_1}{\ddGcons{\lefDcons{\rig A}}}}}{
                \vlhtr{\DD_2}{\Gcons{\lefDcons{\lef A},\lefDcons{\lef A}}}}}}}
      }
    \end{equation}
    Note that the number of origins is not increased because
    the derivation is left-free. Furthermore, we can choose the order
    of the two new cuts such that the origin we are working
    on belongs to the topmost cut. Thus, we can proceed by the induction hypothesis. 
  \item\label{c:sb} If $\rr$ is a $\sbdot$ or $\sfbdot$ step, such that the
    cut-formula $\lef A$ is inside $\Sigma$. Then there are two
    subcases.
    \begin{enumerate}
    \item If the depth of $\Gamma\conhole$ is $0$, then
      $\colGcons{\Sigma,\Dcons\emptyset}=\Sigma,\Delta'\cons\emptyset$
      for some $\Delta'\conhole$. Thus, without loss of generality, we have
      \begin{equation*}%\small
        \label{eq:sb-cut1}
        \cutredcasec[6.5ex]{\sbdot}{%\qquad
          \svlderivation{
            \vliin{\hskip5em\cutr}{}{\Scons{\emptyset},\Dcons{\emptyset}}{
              \vlhtr{\DD_1}{\ddScons{\rig A},\lhs{(\Dcons{\emptyset})}}}{
              \vlin{\sbdot}{}{\Scons{\lef A},\Dcons{\emptyset}}{
                \vlhtr{\DD_2}{\Dcons{\wbrn{\Scons{\lef A}}n}}}}}
        }{%\hskip9em
          \svlderivation{
            \vlin{\conr}{}{\Scons{\emptyset},\Dcons{\emptyset}}{
              \vlin{\sbdot}{}{\Scons{\emptyset},\lhs{(\Dcons{\emptyset})},\Dcons{\emptyset}}{
                \vliin{\hskip15em\cutr}{}{\Dcons{\wbrn{\Scons{\emptyset},\lhs{(\Dcons{\emptyset})}}n}}{
                  \vlcin{(n+1)}{\weakr}{}{\ddDcons{\wbrn{\ddScons{\rig A},\lhs{(\Dcons{\emptyset})}}n}}{
                    \vlcin{2n}{\necr}{}{\wbrn{\ddScons{\rig A},\lhs{(\Dcons{\emptyset})}}{2n}}{
                      \vlhtr{\DD_1}{\ddScons{\rig A},\lhs{(\Dcons{\emptyset})}}}}}{
                  \vlin{\weakr}{}{\Dcons{\wbrn{\Scons{\lef A},\lhs{(\Dcons{\emptyset})}}n}}{
                    \vlhtr{\DD_2}{\Dcons{\wbrn{\Scons{\lef A}}n}}}}}}}}
      \end{equation*}
      where we use Proposition~\ref{prop:admiss} to remove the $\weakr$- and $\necr$-steps.    
    \item\label{c:sb:deep} If the depth of $\Gamma\conhole$ is $\ge1$, then
      $\colGcons{\Sigma,\Dcons\emptyset}=\Gamma'\cons{\wbr{\Sigma,\Delta'\cons\emptyset}}$
      for some $\Gamma'\conhole$ and $\Delta'\conhole$. Thus, without loss of generality, we have
      \begin{equation}%\footnotesize
        \label{eq:sb-cut}
        \cutredcasec[8ex]{\sbdot}{
          \svlderivation{
            \vliin{\hskip5em\cutr}{}{\Gcons{\wbr{\Scons{\emptyset},\Dcons{\emptyset}}}}{
              \vlhtr{\DD_1}{\ddGcons{\wbr{\ddScons{\rig A},\lhs{(\Dcons{\emptyset})}}}}}{
              \vlin{\sbdot}{}{\Gcons{\wbr{\Scons{\lef A},\Dcons{\emptyset}}}}{
                \vlhtr{\DD_2}{\Gcons{\wbr{\Dcons{\wbrn{\Scons{\lef A}}n}}}}}}}
        }{
          \svlderivation{
            \vlin{\conr}{}{\Gcons{\wbr{\Scons{\emptyset},\Dcons{\emptyset}}}}{
              \vlin{\sbdot}{}{\Gcons{\wbr{\Scons{\emptyset},\lhs{(\Dcons{\emptyset})},\Dcons{\emptyset}}}}{
                \vliin{\hskip8em\cutr}{}{\Gcons{\wbr{\Dcons{\wbrn{\Scons{\emptyset},\lhs{(\Dcons{\emptyset})}}n}}}}{
                  \vlcin{(n+1)}{\weakr}{}{\ddGcons{\wbr{\ddDcons{\wbrn{\ddScons{\rig A},\lhs{(\Dcons{\emptyset})}}n}}}}{
                    \vlcin{2n}{\vdot}{}{\ddGcons{\wbr{\wbrn{\ddScons{\rig A},\lhs{(\Dcons{\emptyset})}}{2n}}}}{
                      \vlhtr{\DD_1}{\ddGcons{\wbr{\ddScons{\rig A},\lhs{(\Dcons{\emptyset})}}}}}}}{
                  \vlin{\weakr}{}{\Gcons{\wbr{\Dcons{\wbrn{\Scons{\lef A},\lhs{(\Dcons{\emptyset})}}n}}}}{
                    \vlhtr{\DD_2}{\Gcons{\wbr{\Dcons{\wbrn{\Scons{\lef A}}n}}}}}}}}}
      }
      \end{equation}
      where we use Proposition~\ref{prop:admiss} to remove the
      $\weakr$- and $\vdot$-steps.  This case is the reason why we
      need the presence of $\vax$ when we have $\bax$ in our logic.
    \end{enumerate}
  \item\label{c:sb:Delta} If $\rr$ is a $\sbdot$ or $\sfbdot$, such that the cut-formula
    $\lef A$ is inside $\Dcons\emptyset$, we have that
    $\Dcons\emptyset=\Delta_1\cons{\Delta_2\cons\emptyset,\Delta_3\cons\emptyset}$
    and we can reduce as follows
    \begin{equation}%\footnotesize
      \cutredcasec[8ex]{\sbdot}{
        \svlderivation{
          \vliin{\hskip3em\cutr}{}{\Gcons{\Sigma,\Delta_1\cons{\Delta_2\cons\emptyset,\Delta_3\cons\emptyset}}}{
            \vlhtr{\DD_1}{\ddGcons{\lhs\Sigma,\ddDelta_1\cons{\ddDelta_2\cons{\rig A},\lhs{(\Delta_3\cons\emptyset)}}}}}{
            \vlin{\sbdot}{}{\Gcons{\Sigma,\Delta_1\cons{\Delta_2\cons{\lef A},\Delta_3\cons\emptyset}}}{
              \vlhtr{\DD_2}{\Gcons{\Delta_1\cons{\Delta_2\cons{\lef A},\Delta_3\cons{\wbrn{\Sigma}n}}}}}}}
      }{
        \svlderivation{
          \vlin{\conr}{}{\Gcons{\Sigma,\Delta_1\cons{\Delta_2\cons\emptyset,\Delta_3\cons\emptyset}}}{
            \vlin{\sbdot}{}{\Gcons{\lhs\Sigma,\Sigma,\Delta_1\cons{\Delta_2\cons\emptyset,\Delta_3\cons\emptyset}}}{
              \vliin{\hskip5.5em\cutr}{}{\Gcons{\lhs\Sigma,\Delta_1\cons{\Delta_2\cons\emptyset,\Delta_3\cons{\wbrn{\Sigma}n}}}}{
                \vlin{\weakr}{}{\ddGcons{\lhs\Sigma,\ddDelta_1\cons{\ddDelta_2\cons{\rig A},\lhs{(\Delta_3\cons{\wbrn{\Sigma}n})}}}}{
                  \vlhtr{\DD_1}{\ddGcons{\lhs\Sigma,\ddDelta_1\cons{\ddDelta_2\cons{\rig A},\lhs{(\Delta_3\cons\emptyset)}}}}}}{
                \vlin{\weakr}{}{\Gcons{\lhs\Sigma,\Delta_1\cons{\Delta_2\cons{\lef A},\Delta_3\cons{\wbrn{\Sigma}n}}}}{
                  \vlhtr{\DD_2}{\Gcons{\Delta_1\cons{\Delta_2\cons{\lef A},\Delta_3\cons{\wbrn{\Sigma}n}}}}}}}}}
      }
    \end{equation}
    where the $\weakr$ on the left is not needed if
    $\lhs{(\Delta_3\cons\emptyset)}=\lhs{(\Delta_3\cons{\wbrn{\Sigma}n})}$.
    Note that this case can be seen as a special case of case~\ref{c:r} above.
  \item If $\rr$ is a $\sfdot$ or $\sbfdot$, such that the cut-formula $\lef A$ is inside
      $\Sigma$, then the situation is similar to case~\ref{c:sb:deep} above:
      \begin{equation}%\footnotesize
        \label{eq:sf-cut}\hskip3em
        \cutredcasec[8ex]{\sfdot}{
          \svlderivation{
            \vliin{\cutr}{}{\Gcons{\wbr{\Scons{\emptyset}},\Dcons{\emptyset}}}{
              \vlhtr{\DD_1}{\ddGcons{\wbr{\ddScons{\rig A}},\lhs{(\Dcons{\emptyset})}}}}{
              \vlin{\sfdot}{}{\Gcons{\wbr{\Scons{\lef A}},\Dcons{\emptyset}}}{
                \vlhtr{\DD_2}{\Gcons{\Dcons{\wbr{\Scons{\lef A}}}}}}}}
        }{
          \svlderivation{
            \vlin{\conr}{}{\Gcons{\wbr{\Scons{\emptyset}},\Dcons{\emptyset}}}{
              \vlin{\sfdot}{}{\Gcons{\wbr{\Scons{\emptyset}},\Dcons{\emptyset},\lhs{(\Dcons{\emptyset})}}}{
                \vliin{\hskip7em\cutr}{}{\Gcons{\Dcons{\wbr{\Scons{\emptyset}}},\lhs{(\Dcons{\emptyset})}}}{
                  \vlcin{(n+1)}{\weakr}{}{\ddGcons{\ddDcons{\wbr{\ddScons{\rig A}}},\lhs{(\Dcons{\emptyset})}}}{
                    \vlcin{n}{\vdot}{}{\ddGcons{\wbrn{\wbr{\ddScons{\rig A}}}{n},\lhs{(\Dcons{\emptyset})}}}{
                      \vlhtr{\DD_1}{\ddGcons{\wbr{\ddScons{\rig A}},\lhs{(\Dcons{\emptyset})}}}}}}{
                  \vlin{\weakr}{}{\Gcons{\Dcons{\wbr{\Scons{\lef A}}},\lhs{(\Dcons{\emptyset})}}}{
                    \vlhtr{\DD_2}{\Gcons{\Dcons{\wbr{\Scons{\lef A}}}}}}}}}}
      }
      \end{equation}
      where we use Proposition~\ref{prop:admiss} to remove the
      $\weakr$- and $\vdot$-steps.  This case is the reason why we
      need the presence of $\vax$ when we have $\fax$ in our logic.
  \item If $\rr$ is a $\sfdot$ or $\sbfdot$, such that the cut-formula $\lef A$ is inside
    $\Dcons\emptyset$, then the situation is similar to case~\ref{c:sb:Delta} above.
  \item If $\rr$ is a $\sbdot$, $\sfbdot$, $\sfdot$, or $\sbfdot$, such that the cut-formula
    $\lef A$ is inside $\Gamma\conhole$, then we proceed as in case~\ref{c:r} above.
  \item Finally, if $\rr$ is another $\cutr$, we can reduce as follows:
    \begin{equation*}
%    %\small
    %\tiny
      \hskip3em
      \cutredcasec[6.5ex]{\cutr}{
        \svlderivation{
          \vliin{\cutr}{}{\Gcons{\Scons{\emptyset},\Dcons{\emptyset}}}{
            \vlhtr{\DD_1}{\ddGcons{\ddScons{\rig A},\lhs{(\Dcons{\emptyset})}}}}{
            \vliin{\cutr}{}{\Gcons{\Scons{\lef A},\Dcons{\emptyset}}}{
              \vlhtr{\DD_2}{\ddGcons{\lhs{(\Scons{\lef A})},\ddDcons{\rig B}}}}{
              \vlhtr{\DD_3}{\Gcons{\Scons{\lef A},\Dcons{\lef B}}}}}}
      }{
        \svlderivation{
          \vliin{\hskip5em\cutr}{}{\Gcons{\Scons{\emptyset},\Dcons{\emptyset}}}{
            \vlhtr{\DD_2'}{\ddGcons{\lhs{(\Scons{\emptyset})},\ddDcons{\rig B}}}}{
            \vliin{\cutr}{}{\Gcons{\Scons{\emptyset},\Dcons{\lef B}}}{
              \vlin{\weakr}{}{\ddGcons{\ddScons{\rig A},\lhs{(\Dcons{\lef B})}}}{
                \vlhtr{\DD_1}{\ddGcons{\ddScons{\rig A},\lhs{(\Dcons{\emptyset})}}}}}{
              \vlhtr{\DD_3}{\Gcons{\Scons{\lef A},\Dcons{\lef B}}}}}}
      }
    \end{equation*}
    where $\DD_2'$ exists because our original derivation is
    left-free, and the $\weakr$-step is removed by
    Proposition~\ref{prop:admiss}. Note that it can happen that
    $\lhs{(\Scons{\emptyset})}=\lhs{(\Scons{\lef A})}$ and/or
    $\lhs{(\Dcons{\emptyset})}=\lhs{(\Dcons{\lef B})}$, depending on
    where the output formula occurs in
    $\Gcons{\Scons{\emptyset},\Dcons{\emptyset}}$. 
  \end{enumerate} 
  Above, we have only shown the cases for $\cutr$. The ones for
  $\wdiacutr$ and $\wboxcutr$ are similar, except the ones when $\rr$
  is one of $\sbdot$, $\sfbdot$, $\sfdot$, or $\sbfdot$. When such a
  rule is on the right above a $\wboxcutr$, we decompose that
  $\wboxcutr$ into a $\svlef$ and a $\cutr$ (using
  Fact~\ref{fact:4cut}) and then apply Lemma~\ref{lem:4-over-b} to
  permute the $\svlef$ over $\rr$, so that we can proceed as described
  above. When the $\cutr$ is permuted over $\rr$, we can compose it
  again with the $\svlef$-instance, so that we can proceed by
  induction hypothesis.  Observe that we make crucial use of the
  left-free property. Without it, the number of origins in $\flg{\DD}$
  would not be stable. Furthermore, note that the cut-cut permutation
  does not affect cuts that are above our current origin. Thus, all
  origins above remain anchored. This is the reason for starting with
  the topmost~one.  %%\qed
\end{proof}

\begin{lemma}\label{lem:one-step-anchored}
  Let $\tuple{\Xax,\Yax}$ be a safe pair of axioms. If there is a proof
  \begin{equation}%\small
    \label{eq:cutlem}
    \svlderivation{
      \vliin{\ast\cutr}{}{\Gcons{\emptyset}}{
        \vlhtr{\DD_1}{\Gamma_1\cons{\rig A}}}{
        \vlhtr{\DD_2}{\Gamma_2\cons{\lef A}}}}
  \end{equation}
  where $\DD_1$ and $\DD_2$ are both in
  $\NCKp+\sXlefrig+\sYdot$ and where $\ast\cutr$ is anchored, then there is a proof
  $\DD'$ of $\Gcons{\emptyset}$ in $\NCKp+\sXlefrig+\sYdot+\Cutr$ in
  which all cuts have a smaller rank.\looseness=-1
\end{lemma}

\begin{proof}%%[Proof of Lemma~\ref{lem:one-step-anchored}]
    We make a case analysis on the cut-formula $A$.
  \begin{enumerate}
  \item If $A=B\cand C$, we reduce the cut rank as follows:
    \begin{equation*}%\small
      \hskip3em
      \cutredcase{\cand}{
        \svlderivation{
          \vliin{\cutr}{}{\Gamma\conempty}{
            \vliin{\rig\cand}{}{\dGamcon{\prig{B\cand C}}}{
              \vlhtr{\DD'_1}{\dGamcon{\prig{B}}}}{
              \vlhtr{\DD''_1}{\dGamcon{\prig{C}}}}}{
            \vlin{\lef\cand}{}{\Gamcon{\plef{B\cand C}}}{
              \vlhtr{\DD'_2}{\Gamcon{\lef B,\lef C}}}}
      }}{
        \svlderivation{
          \vliin{\cutr}{}{\Gamma\conempty}{
            \vlhtr{\DD'_1}{\dGamcon{\prig{B}}}}{
            \vliin{\cutr}{}{\Gamcon{\lef B}}{
              \vlin{\weakr}{}{\dGamcon{\lef B,\prig{C}}}{
                \vlhtr{\DD''_1}{\dGamcon{\prig{C}}}}}{
              \vlhtr{\DD'_2}{\Gamcon{\lef B,\lef C}}}}
      }}
    \end{equation*}
    where $\DD_1'$ and $\DD_1''$ exist since the $\rig\cand$-rule is
    height-preserving invertible (Proposition~\ref{prop:invert}), and
    $\DD_2'$ exists since our cut is anchored. Finally, we remove the
    $\weakr$-step using Proposition~\ref{prop:admiss}.
  \item If $A=B\implies C$, we reduce the cut rank as follows:
    \begin{equation*}%\small
      \hskip3em
      \cutredcase{\implies}{
        \svlderivation{
          \vliin{\cutr}{}{\Gamma\conempty}{
            \vlin{\rig\implies}{}{\dGamcon{\prig{B\implies C}}}{
              \vlhtr{\DD_1'}{\dGamcon{\lef B,\rig C}}}}{
            \vliin{\lef\implies}{}{\Gamcon{\plef{B\implies C}}}{
              \vlhtr{\DD_2'}{\dGamcon{\rig B}}}{
              \vlhtr{\DD_2''}{\Gamcon{\lef C}}}}
      }}{
        \svlderivation{
          \vliin{\cutr}{}{\Gamma\conempty}{
            \vlhtr{\DD_2'}{\dGamcon{\rig B}}}{
            \vliin{\cutr}{}{\Gamma\cons{\lef B}}{
              \vlhtr{\DD_1'}{\dGamcon{\lef B,\rig C}}}{
              \vlhtr{\DD_2''}{\Gamcon{\lef C}}}}
      }}
    \end{equation*}
    where again, $\DD_1'$ exists by invertibility of the
    $\rig\implies$-rule (Proposition~\ref{prop:invert}), and $\DD_2'$
    and $\DD_2''$ exist since our cut is anchored.
  \item If $A=\wbox B$, and the rule on the right above the $\cutr$ is
  a $\lef\wbox$, then we reduce the cut rank as follows: 
    \begin{equation*}%\small
      \hskip3em
      \cutredcasec[8ex]{\wbox}{
        \svlderivation{
          \vliin{\wboxcutr}{}{\Gcons{\Tcons{\wbr{\Delta}}}}{
            \vlin{\rig\wbox}{}{\ddGcons{\rig{\wbox B},\ddown{(\Tcons{\wbr{\Delta}})}}}{
              \vlhtr{\DD_1'}{\ddGcons{\wbr{\rig B},\ddown{(\Tcons{\wbr{\Delta}})}}}}}{
            \vlin{\lef\wbox}{}{\Gcons{\Tcons{\lef{\wbox B},\wbr{\Delta}}}}{
              \vlhtr{\DD_2'}{\Gcons{\Tcons{\wbr{\lef B,\Delta}}}}}}}
      }{
        \svlderivation{
          \vlin{\conr}{}{\Gcons{\Tcons{\wbr{\Delta}}}}{
            \vliin{\hskip10em\cutr}{}{\Gcons{\Tcons{\wbr{\Delta}},\ddown{(\Tcons{\wbr{\Delta}})}}}{
              \vlcin{(n+1)}{\weakr}{}{\ddGcons{\ddTcons{\wbr{\rig B,\ddDelta}},\ddown{(\Tcons{\wbr{\Delta}})}}}{
                \vlin{\weakr}{}{\ddGcons{\wbrn{\wbr{\rig B,\ddDelta}}n,\ddown{(\Tcons{\wbr{\Delta}})}}}{
                  \vlcin{n}{\vdot}{}{\ddGcons{\wbrn{\wbr{\rig B}}n,\ddown{(\Tcons{\wbr{\Delta}})}}}{
                    \vlhtr{\DD_1'}{\ddGcons{\wbr{\rig B},\ddown{(\Tcons{\wbr{\Delta}})}}}}}}}{
              \vlin{\weakr}{}{\Gcons{\Tcons{\wbr{\lef B,\Delta}},\ddown{(\Tcons{\wbr{\Delta}})}}}{
                \vlhtr{\DD_2'}{\Gcons{\Tcons{\wbr{\lef B,\Delta}}}}}}}}
      }
    \end{equation*}
    Where $\DD_1'$ and $\DD_2'$ exist for the same reason as above,
    and we finally apply Proposition~\ref{prop:admiss} to remove the
    $\weakr$- and $\vdot$-steps, where $n$ is the depth of $\ddTcons{\enspace}$.
    \item If $A = \wbox B$ and the rule on the right above $\cutr$ is
     a $\svblef$, then we can reduce to the previous case as follows,
     \begin{equation*}%\small\hskip-3em
     \cutredcaseb[5ex]{\wbox}{
       \svlderivation{
         \vliin{\wboxcutr}{}{ \Gcons{ \Tcons{ \Delta_1 \cons{ \wbr{ \Delta_2 } } } }  }{ 
           \vlhy{\ddGcons{\rig{\wbox B},\ddown{(\Tcons{\Delta_1 \cons{ \wbr{ \Delta_2 } }})}}} }{
     	   \vlin{\svblef}{}{ \Gcons{ \Tcons{ \lef{ \wbox B } ,\Delta_1 \cons{ \wbr{ \Delta_2 } } } } }{ 
             \vlhy{ \Gcons{ \Tcons{ \Delta_1 \cons{ \wbr{\lef B, \Delta_2 } } } } } }
       }}
     }{\hskip13em
       \svlderivation{
     	 \vliin{\wboxcutr}{}{ \Gcons{ \Tcons{ \Delta_1 \cons{ \wbr{ \Delta_2 } } } }  }{ 
           \vlhy{\ddGcons{\rig{\wbox B},\ddown{(\Tcons{\Delta_1 \cons{ \wbr{ \Delta_2 } }})}}} }{
     	   \vlin{\lef\wbox}{}{ \Gcons{ \Tcons{ \Delta_1 \cons{ \lef{ \wbox B } , \wbr{ \Delta_2 } }  } } }{ 
             \vlhy{ \Gcons{ \Tcons{ \Delta_1 \cons{ \wbr{ \lef{ B }, \Delta_2 } } } } } }
       }}
     }
     \end{equation*}
  \item If $A=\wbox B$, and the rule on the right above the $\cutr$ is
    a $\tlef$, then we reduce the cut rank as follows:\vadjust{\vskip1ex}
    \begin{equation*}
      \hskip1em
    %\small
%      \cutredcase
	\cutredcasea
      {\wbox\tax}{
        \svlderivation{
          \vliin{\wboxcutr}{}{\Gcons{\Tcons{\emptyset}}}{
            \vlin{\rig\wbox}{}{\ddGcons{\rig{\wbox B},\ddown{(\Tcons{\emptyset})}}}{
              \vlhtr{\DD_1'}{\ddGcons{\wbr{\rig B},\ddown{(\Tcons{\emptyset})}}}}}{
            \vlin{\tlef}{}{\Gcons{\Tcons{\lef{\wbox B}}}}{
              \vlhtr{\DD_2'}{\Gcons{\Tcons{\lef B}}}}}}
      }{
        \svlderivation{
          \vlin{\conr}{}{\Gcons{\Tcons{\emptyset}}}{
            \vliin{\cutr}{}{\Gcons{\Tcons{\emptyset},\ddown{(\Tcons{\emptyset})}}}{
              \vlcin{(n+1)}{\weakr}{}{\ddGcons{\ddTcons{\rig B},\ddown{(\Tcons{\emptyset})}}}{
                \vlcin{(n-1)}{\vdot/\tdot}{}{\ddGcons{\wbrn{\rig B}n,\ddown{(\Tcons{\emptyset})}}}{
                  \vlhtr{\DD_1'}{\ddGcons{\wbr{\rig B},\ddown{(\Tcons{\emptyset})}}}}}}{
              \vlin{\weakr}{}{\Gcons{\Tcons{\lef B},\ddown{(\Tcons{\emptyset})}}}{
                \vlhtr{\DD_2'}{\Gcons{\Tcons{\lef B}}}}}}}
      }
    \end{equation*}
    where $n$ is the depth of $\ddTcons{\enspace}$, and at the top we
    either have one $\tdot$ step (if $n=0$) or $n-1$ steps of $\vdot$ (if $n\ge1$), which can be removed via
    Proposition~\ref{prop:admiss}. Note that $\vax\in\Xax$ if $n\ge1$.
  \item If $A=a$, then the cut is removed as follows:
    \begin{equation*}%\small
      \cutredcase{\idr}{
        \svlderivation{
          \vliin{\cutr}{}{\Gamcon{\rig a}}{
            \vlhtr{\DD_1}{\Gamcon{\rig a}}}{
            \vlin{\idr}{}{\Gamcon{\lef a,\rig a}}{
              \vlhy{}}}
      }\quad}{\quad
        \svlderivation{
          \vlhtr{\DD_1}{\Gamcon{\rig a}}
      }}
    \end{equation*}
    Note that here $\dGamcon{\rig a}=\Gamcon{\rig a}$. 
  \item If $A=\bot$, the situation is similar: 
    \begin{equation*}%\small
      \cutredcase{\bot}{
        \svlderivation{
          \vliin{\cutr}{}{\Gamcon{\rig\Pi}}{
            \vlhtr{\DD_1}{\Gamcon{\rig\Pi}}}{
            \vlin{\lef\bot}{}{\Gamcon{\lef\bot,\rig\Pi}}{
              \vlhy{}}}
      }\quad}{\quad
        \svlderivation{
          \vlhtr{\DD_1}{\Gamcon{\rig\Pi}}
      }}
    \end{equation*}
  \item If $A=B\cor C$, we proceed by induction on the height of
    $\DD_1$ and make a case analysis on its bottommost rule instance~$\rr$.
    \begin{enumerate}
    \item If $\rr$ is a $\rig\cor$, then it must decompose the cut
      formula, bottom-up, and we can reduce the cut rank as follows:
      \begin{equation*}%\small
        \hskip3em
        \cutredcase{\cor}{
          \svlderivation{
            \vliin{\cutr}{}{\Gamma\conempty}{
              \vlin{\rig\cor}{}{\ddGamcon{\prig{B\cor C}}}{
                \vlhtr{\DD_1'}{\ddGamcon{\prig{B}}}}}{
              \vliin{\lef\cor}{}{\Gamcon{\plef{B\cor C}}}{
                \vlhtr{\DD_2'}{\Gamcon{\lef B}}}{
                \vlhtr{\DD_2''}{\Gamcon{\lef C}}}}
        }\quad}{\quad
          \svlderivation{
            \vliin{\cutr}{}{\Gamma\conempty}{
              \vlhtr{\DD_1'}{\dGamcon{\prig{B}}}}{
              \vlhtr{\DD_2'}{\Gamcon{\lef B}}}
        }}          
      \end{equation*}
      The case when $\rig\cor$ chooses $\rig C$ is symmetric.
    \item If $\rr$ is a $\lef\implies$, then we have
      \vadjust{\vskip1ex}
      \begin{equation*}
        \hskip5em
        \cutredcasec{\lef\implies}{%\hskip-10em
          \svlderivation{
            \vliin{\cutr}{}{\Gamcon{\lefTcons{\plef{D\implies E}},\Dcons\emptyset}}{
              \vliin{\lef\implies}{}{\dGamcon{\lefTcons{\plef{D\implies E}},\ddDcons{\rig A}}}{
                \vlhtr{\DD_1'}{\dGamcon{\lefTcons{\rig D}}}}{
                \vlhtr{\DD_1''}{\dGamcon{\lefTcons{\lef E},\ddDcons{\rig A}}}}}{
              \vlhtr{\DD_2'}{\Gamcon{\lefTcons{\plef{D\implies E}},\Dcons{\lef A}}}}
        }}{
          \svlderivation{
            \vliin{\lef\implies}{}{\Gamcon{\lefTcons{\plef{D\implies E}},\Dcons\emptyset}}{
              \vlin{\weakr}{}{\dGamcon{\lefTcons{\rig D},\lhs{({\Dcons\emptyset})}}}{
                \vlhtr{\DD_1'}{\dGamcon{\lefTcons{\rig D}}}}}{
              \vliin{\cutr}{}{\Gamcon{\lefTcons{\lef E},\Dcons\emptyset}}{
                \vlhtr{\DD_1''}{\dGamcon{\lefTcons{\lef E},\ddDcons{\rig A}}}}{
                \vlin{\Invr}{}{\Gamcon{\lefTcons{\plef{E}},\Dcons{\lef A}}}{
                  \vlhtr{\DD_2'}{\Gamcon{\lefTcons{\plef{D\implies E}},\Dcons{\lef A}}}}}}
        }}
      \end{equation*}
      where the $\Invr$-step is removed by
      Proposition~\ref{prop:invert}, and we can proceed by induction
      hypothesis.
  \item 
    All other cases are standard commutative cases and are shown
    below. They are in fact symmetric to their corresponding cases in
    Lemma~\ref{lem:make-anchored}. This, in particular, concerns the
    case where $\rr$ is $\lef\cor$. Since our $\cutr$ is anchored, the
    output branch of the sequent is next to the $\lef A$ occurrence in the right
    premise of the $\cutr$. Therefore the $\lef\cor$ above the left
    premise of the $\cutr$ can now be permuted under the
    cut:\vadjust{\vskip1ex}
      \begin{equation*}%\small
        \hskip5em
        \cutredcasec{\lef\cor}{
          \svlderivation{
            \vliin{\cutr}{}{\Gamma\cons{\rig\Pi}}{
              \vliin{\lef\cor}{}{\Gamcon{\rig A}}{
                \vlhtr{\DD_1'}{{\Gamma_1}\cons{\rig A}}}{
                \vlhtr{\DD_1''}{{\Gamma_2}\cons{\rig A}}}}{
              \vlhtr{\DD_2}{\Gamma\cons{\lef A,\rig\Pi}}}
        }}{
          \svlderivation{
            \vliin{\lef\cor}{}{\Gamma\cons{\rig\Pi}}{
              \vliin{\hskip10em\cutr}{}{\Gamma_1\cons{\rig\Pi}}{
                \vlhtr{\DD_1'}{{\Gamma_1}\cons{\rig A}}}{
                \vlin{\Invr}{}{\Gamma_1\cons{\lef A,\rig\Pi}}{
                  \vlhtr{\DD_2}{\Gamma\cons{\lef A,\rig\Pi}}}}}{
              \vliin{\cutr}{}{\Gamma_2\cons{\rig\Pi}}{
                \vlhtr{\DD_1''}{{\Gamma_2}\cons{\rig A}}}{
                \vlin{\Invr}{}{\Gamma_2\cons{\lef A,\rig\Pi}}{
                  \vlhtr{\DD_2}{\Gamma\cons{\lef A,\rig\Pi}}}}}
        }}
      \end{equation*}
      The $\Invr$-steps are removed by
      Proposition~\ref{prop:invert}. The other invertible rules are
      handled similarly:
      \begin{equation*}%\small
        \hskip4em
        \cutredcase{\rr}{
          \svlderivation{
            \vliin{\cutr}{}{\Gamma\conempty}{
              \vlin{\rr}{}{\dGamcon{\rig A}}{
                \vlhtr{\DD_1'}{\down{\Gamma_1}\cons{\rig A}}}}{
              \vlhtr{\DD_2}{\Gamma\cons{\lef A}}}
        }\quad}{\quad
          \svlderivation{
            \vlin{\rr}{}{\Gamma\conempty}{
              \vliin{\cutr}{}{\Gamma_1\conempty}{
                \vlhtr{\DD_1'}{\down{\Gamma_1}\cons{\rig A}}}{
                \vlin{\Invr}{}{\Gamma_1\cons{\lef A}}{
                  \vlhtr{\DD_2}{\Gamma\cons{\lef A}}}}}
        }}
      \end{equation*}
      Finally, if the rule on the left above the cut is an axiom
      $\lef\bot$ (note that it cannot be $\idr$ because $\rig A$ is
      not an atom), then we reduce as follows:
      \begin{equation*}
        \cutredcase{\lef\bot}{
          \svlderivation{
            \vliin{\cutr}{}{\Gamma\conempty}{
              \vlin{\lef\bot}{}{\dGamcon{\rig A}}{
                \vlhy{}}}{
              \vlhtr{\DD_2}{\Gamma\cons{\lef A}}}
          }\quad}{\quad
          \svlderivation{
            \vlin{\lef\bot}{}{\Gamma\conempty}{
              \vlhy{}}
          }}
      \end{equation*}
    \end{enumerate}
  \item If $A=\wdia B$ we proceed as in the previous case. 
      The only
    difference occurs when the bottommost rule $\rr$ in $\DD_1$ works on
    $\rig A$. There are three subcases:
    \begin{enumerate}
    \item If $\rr$ is $\svrig$ we have
    \vspace{1ex}
      \begin{equation*}%\small
        \hskip5em
%        \cutredcase
	\cutredcasec
        {\wdia\vax}{
          \svlderivation{
            \vliin{\wdiacutr}{}{\Gcons{\lef\Theta\cons{\lef\Delta\cons{\emptyset}}}}{
              \vlin{\svrig}{}{
                \dGamcon{\lef\Theta\cons{\prig{\wdia B},\lef\Delta\cons{\emptyset}}}}{
                \vlhtr{\DD_1'}{
                  \dGamcon{\lef\Theta\cons{\lef\Delta\cons{\prig{\wdia B}}}}}}}{
              \vlhtr{\DD_2}{
                \Gamma\cons{\plef{\wdia B},\lef\Theta\cons{\lef\Delta\cons{\emptyset}}}}}
        }}{
          \svlderivation{
            \vliin{\hskip15em\wdiacutr}{}{\Gcons{\lef\Theta\cons{\lef\Delta\cons{\emptyset}}}}{
              \vlhtr{\DD_1'}{
                  \dGamcon{\lef\Theta\cons{\lef\Delta\cons{\prig{\wdia B}}}}}}{
              \vlhtr{\DD_2}{
                \Gamma\cons{\plef{\wdia B},\lef\Theta\cons{\lef\Delta\cons{\emptyset}}}}}
        }}
      \end{equation*}
      and can proceed by the induction hypothesis.
    \item If $\rr$ is $\rig\wdia$ we can reduce the cut rank as follows:
      \begin{equation*}
        \hskip5em 
        \cutredcasec{\wdia}{
          \svlderivation{
            \vliin{\wdiacutr}{}{\Gcons{\lef\Theta\cons{\wbr{\lef\Delta}}}}{
              \vlin{\rig\wdia}{}{
                \dGamcon{\lef\Theta\cons{\prig{\wdia B},\wbr{\lef\Delta}}}}{
                \vlhtr{\DD_1'}{
                  \dGamcon{\lef\Theta\cons{\wbr{\rig B,\lef\Delta}}}}}}{
              \vlin{\lef\wdia}{}{
                \Gamma\cons{\plef{\wdia B},\lef\Theta\cons{\wbr{\lef\Delta}}}}{
                \vlhtr{\DD_2'}{
                  \Gamma\cons{\wbr{\lef B},\lef\Theta\cons{\wbr{\lef\Delta}}}}}}
        }}{
          \svlderivation{
            \vlin{\conr}{}{\Gcons{\lef\Theta\cons{\wbr{\lef\Delta}}}}{
              \vliin{\hskip5em\cutr}{}{
                \Gcons{\lef\Theta\cons{\wbr{\lef\Delta}},\lef\Theta\cons{\wbr{\lef\Delta}}}}{
                \vlin{\weakr}{}{
                  \dGamcon{\lef\Theta\cons{\wbr{\rig B,\lef\Delta}},\lef\Theta\cons{\wbr{\lef\Delta}}}}{
                  \vlhtr{\DD_1'}{
                    \dGamcon{\lef\Theta\cons{\wbr{\rig B,\lef\Delta}}}}}}{
                \vlcin{(n+1)}{\weakr}{}{
                  \Gcons{\lef\Theta\cons{\wbr{\lef B,\lef\Delta}},\lef\Theta\cons{\wbr{\lef\Delta}}}}{
                  \vlin{\weakr}{}{
                    \Gamma\cons{\wbrn{\wbr{\lef B,\lef\Delta}}n,\lef\Theta\cons{\wbr{\lef\Delta}}}}{
                    \vlcin{n}{\vdot}{}{
                      \Gamma\cons{\wbrn{\wbr{\lef B}}n,\lef\Theta\cons{\wbr{\lef\Delta}}}}{
                      \vlhtr{\DD_2'}{
                        \Gamma\cons{\wbr{\lef B},\lef\Theta\cons{\wbr{\lef\Delta}}}}}}}}}            
        }}
      \end{equation*}
      where $n$ is the depth of the context $\lef\Theta\conhole$, and
      $\DD_2'$ exists because the instance of $\wdiacutr$ is
      anchored. We use Proposition~\ref{prop:admiss} to remove the
      $\weakr$- and $\vdot$-steps.  We proceed similarly for
      $\svdrig$.
  \item If $\rr$ is $\trig$ the situation is similar and we can reduce the cut rank as follows:
      \begin{equation*}
      \hskip5em 
        \cutredcasec{\wdia\tax}{
          \svlderivation{
            \vliin{\wdiacutr}{}{\Gcons{\lef\Theta\cons{\emptyset}}}{
              \vlin{\trig}{}{
                \dGamcon{\lef\Theta\cons{\prig{\wdia B}}}}{
                \vlhtr{\DD_1'}{
                  \dGamcon{\lef\Theta\cons{\rig B}}}}}{
              \vlin{\lef\wdia}{}{
                \Gamma\cons{\plef{\wdia B},\lef\Theta\cons{\emptyset}}}{
                \vlhtr{\DD_2'}{
                  \Gamma\cons{\wbr{\lef B},\lef\Theta\cons{\emptyset}}}}}
        }}{
          \svlderivation{
            \vlin{\conr}{}{\Gcons{\lef\Theta\cons{\emptyset}}}{
              \vliin{\hskip10em\cutr}{}{
                \Gcons{\lef\Theta\cons{\emptyset},\lef\Theta\cons{\emptyset}}}{
                \vlin{\weakr}{}{
                  \dGamcon{\lef\Theta\cons{\rig B},\lef\Theta\cons{\emptyset}}}{
                  \vlhtr{\DD_1'}{
                    \dGamcon{\lef\Theta\cons{\rig B}}}}}{
                \vlcin{(n+1)}{\weakr}{}{
                  \Gcons{\lef\Theta\cons{\lef B},\lef\Theta\cons{\emptyset}}}{
                  \vlcin{(n-1)}{\vdot/\tdot}{}{
                    \Gamma\cons{\wbrn{\lef B}n,\lef\Theta\cons{\emptyset}}}{
                    \vlhtr{\DD_2'}{
                      \Gamma\cons{\wbr{\lef B},\lef\Theta\cons{\emptyset}}}}}}}            
        }}        
      \end{equation*}
      Again, $n$ is the depth of the context $\lef\Theta\conhole$. If
      $n=0$, there are no brackets, and we use $\tdot$. If $n\ge1$,
      there is at least one bracket nesting (and therefore $\vax\in\Xax$), and we use $n-1$
      instances of $\vdot$, which then are removed by applying
      Proposition~\ref{prop:admiss}. \qedhere
    \end{enumerate}
  \end{enumerate}
\end{proof}

We can now put things together to complete the proof of cut-elimination.

\begin{proof}[Proof of Theorem~\ref{thm:cut-elim}]
      A proof in
    $\NCKp+\Xlefrig+\Ydot+\cutr$ is trivially also a proof in
    $\NCKp+\sXlefrig+\sYdot+\Cutr$.  We proceed by induction on the
    cut-value $\cv{\DD}$ using the well-ordering $\ll$ (defined after
    Definition~\ref{dfn:cut-anchored+value}). In the base case
    $\cv{\DD}$ is empty, and we are done. Otherwise, we pick a topmost
    cut in $\DD$. If this $\ast\cutr$-instance is anchored, then we
    can by Lemma~\ref{lem:one-step-anchored} replace this cut by cuts
    of smaller rank, and thus reduce the overall cut-value of the
    derivation. If our $\ast\cutr$-instance is not anchored, we
    observe that the subderivation rooted at that $\ast\cutr$-instance
    is left-free (because we chose a topmost cut), and therefore we
    can apply Lemma~\ref{lem:make-anchored} to replace that
    subderivation with one in which all cuts are anchored and have the
    same rank. Thus, the overall cut-value of the derivation has
    reduced as well, and we can proceed by the induction
    hypothesis. Finally, we apply Proposition~\ref{prop:super} to eliminate super steps and obtain
    a proof of the same conclusion in $\NCKp+\Xlefrig+\Ydot$.  %%\qed
\end{proof}

From here it is simple to see why Theorem~\ref{thm:complete-cutfree} holds. 

\begin{proof}[Proof of Theorem~\ref{thm:complete-cutfree}]
We have that $\NCKp+\Xlefrig+\Ydot+\cutr$ is complete with respect to
$\HCK + \Xax + \Yax$ by Theorem~\ref{thm:complete}. If
$\dax\notin\Xax$, we can apply Theorem~\ref{thm:cut-elim} and immediately obtain the
completeness of $\NCKp+\Xlefrig+\Ydot$. If $\dax\in\Xax$,
we use Theorem~\ref{thm:cut-elim} to obtain completeness for
$\NCKp+\Xlefrig\setminus\set{\dlef,\drig}+\Ydot+\ddot$, and then use
Proposition~\ref{PropDaxDerivable} to obtain completeness of
$\NCKp+\Xlefrig+\Ydot$.
\end{proof}

%%%%%%%%%%%%%%%%%%%%%%%%%%%%%%%%%%%%%%%%%%%%%%%%%%%%%%%%%%%%%%%%%%%%%%%%
%%%%%%%%%%%%%%%%%%%%%%%%%%%%%%%%%%%%%%%%%%%%%%%%%%%%%%%%%%%%%%%%%%%%%%%%
%%%%%%%%%%%%%%%%%%%%%%%%%%%%%%%%%%%%%%%%%%%%%%%%%%%%%%%%%%%%%%%%%%%%%%%%

\section{Conclusions}       

To the best of our knowledge, our paper is the first attempt to
provide some unified proof-theoretic framework for the constructive modal cube.
Although this work does not show cut-elimination
for every logic in the cube, we conjecture that the systems presented do, in fact, admit cut. More
precisely:

\begin{conjecture}\label{con:complete-cutfree}
  Let $\Xax\subseteq\set{\dax,\tax,\vax}$ and
  $\Yax\subseteq\set{\dax,\bax,\fax}$, such that if
  $\tax\in\Xax$ and $\fax\in\Yax$ then $\bax\in\Yax$.
  Then every theorem
  of $\HCK+\Xax+\Yax$ is provable in
  $\NCKp+\Xlefrig+\Ydot$.
\end{conjecture}
  This would give us a
cut-free system for every logic in the cube.  The reason why we think
Conjecture~\ref{con:complete-cutfree} is true is the observation that
the only place where $\vax$ steps appear in the presence of $\bdot$ or
$\fdot$ is the permutation of $\sbdot$ steps or $\sfdot$ steps under a
cut, as in \eqref{eq:sb-cut}. Instances of $\vdot$ are then
introduced, and in the admissibility proof for $\vdot$ instances of
$\vlef$ or $\vrig$ are only introduced when $\vdot$ steps are
permuted over instances of $\lef\wbox$ or $\rig\wdia$. However, looking
back at \eqref{eq:sb-cut} and~\eqref{eq:sf-cut}, one can see that it
seems possible to permute these instances of $\lef\wbox$ and
$\rig\wdia$ under the whole derivation block, including the
cut. We have not yet managed to incorporate this observation
into the formal cut-elimination argument, and leave this issue for
further research. 

An alternative approach would be to make use of the
observation that $\bax$ implies $\kax3$ and $\kax5$, by generalizing
the $\lef\bot$- and $\lef\cor$-rules to their intuitionistic versions, as used in~\cite{str:fossacs13,marin:str:aiml14}. This would simplify
the cut-elimination argument for logics containing $\bax$ since
we could reuse a lot of the material already appearing in~\cite{str:fossacs13}.

Another path of further research is to give modal logics a
similar uniform treatment as the substructural logics
in~\cite{CGT:lics08,CST:csl09}. For this, it is necessary to first
look at concrete examples of structural rules corresponding to axioms,
as we have shown in Figure~\ref{fig:CXstr}. Since these rules almost
coincide in the classical, the intuitionistic, and the constructive setting,
we hope to eventually discover a general pattern, yielding uniform cut-elimination arguments for a variety of modal logics. %\looseness=-1

Finally, we have observed an apparent dichotomy between the $\bax$ axiom and the `constructiveness' of constructive modal logic, since the former implies $\kax3$ and $\kax5$, for which we do not know of any approach providing some sort of Curry-Howard correspondence. We therefore believe it would be pertinent to develop further outlooks on such logics. Perhaps it would be possible to find weaker formulations of $\bax$ which are equivalent classically, but not constructively, and which do not entail $\kax3$ and $\kax5$. Such an endeavor might yield new insights for extending the scope of the Curry-Howard correspondence to modal logics.

%%%%%%%%%%%%%%%%%%%%%%%%%%%%%%%%%%%%%%%%%%%%%%%%%%%%%%%%%%%%%%%%%%%%%%%%
%%%%%%%%%%%%%%%%%%%%%%%%%%%%%%%%%%%%%%%%%%%%%%%%%%%%%%%%%%%%%%%%%%%%%%%%
%%%%%%%%%%%%%%%%%%%%%%%%%%%%%%%%%%%%%%%%%%%%%%%%%%%%%%%%%%%%%%%%%%%%%%%%
%%%%%%%%%%%%%%%%%%%%%%%%%%%%%%%%%%%%%%%%%%%%%%%%%%%%%%%%%%%%%%%%%%%%%%%%

\bibliography{constructive-refs}

\begin{thebibliography}{AMdPR01}

\bibitem[aBMTS99]{Zine99}
Zine~El abidine Benaissa, Eugenio Moggi, Walid Taha, and Tim Sheard.
\newblock Logical {Modalities} and {Multi-Stage} {Programming} ({Research}
  {Report}), 1999.

\bibitem[AMdPR01]{alechina:etal:01}
Natasha Alechina, Michael Mendler, Valeria de~Paiva, and Eike Ritter.
\newblock Categorical and {K}ripke semantics for constructive {S4} modal logic.
\newblock In Laurent Fribourg, editor, {\em CSL'01}, volume 2142 of {\em
  Lecture Notes in Computer Science}, pages 292--307. Springer, 2001.

\bibitem[BdP00]{bierman:paiva:00}
Gavin~M. Bierman and Valeria de~Paiva.
\newblock On an intuitionistic modal logic.
\newblock {\em Studia Logica}, 65(3):383--416, 2000.

\bibitem[Br{\"u}09]{brunnler:deepseq}
Kai Br{\"u}nnler.
\newblock Deep sequent systems for modal logic.
\newblock {\em Archive for Mathematical Logic}, 48(6):551--577, 2009.

\bibitem[Bus91]{buss:91}
Samuel~R. Buss.
\newblock The undecidability of $k$-provability.
\newblock {\em Annals of Pure and Applied Logic}, 53:72--102, 1991.

\bibitem[Bus98]{buss:98}
Samuel~R. Buss.
\newblock An introduction to proof theory.
\newblock In Samuel~R. Buss, editor, {\em Handbook of Proof Theory}. Elsevier,
  1998.

\bibitem[CGT08]{CGT:lics08}
Agata Ciabattoni, Nikolaos Galatos, and Kazushige Terui.
\newblock From axioms to analytic rules in nonclassical logics.
\newblock In {\em LICS}, pages 229--240, 2008.

\bibitem[CST09]{CST:csl09}
Agata Ciabattoni, Lutz Stra{\ss}burger, and Kazushige Terui.
\newblock Expanding the realm of systematic proof theory.
\newblock In Erich Gr{\"a}del and Reinhard Kahle, editors, {\em Computer
  Science Logic, CSL'09}, volume 5771 of {\em Lecture Notes in Computer
  Science}, pages 163--178. Springer, 2009.

\bibitem[DM79]{derschowitz:nachum:multisets}
Nachum Dershowitz and Zohar Manna.
\newblock Proving termination with multiset orderings.
\newblock {\em Communications of the ACM}, 22(8):465--476, 1979.

\bibitem[DP96]{Davies00}
Rowan Davies and Frank Pfenning.
\newblock A {Modal} {Analysis} of {Staged} {Computation}.
\newblock In {\em POPL}, pages 258--270, 1996.

\bibitem[Fit48]{fitch:48}
Frederic~B. Fitch.
\newblock Intuitionistic modal logic with quantifiers.
\newblock {\em Portugaliae Mathematica}, 7(2):113--118, 1948.

\bibitem[Fit12]{Fitting12}
Melvin Fitting.
\newblock {Prefixed tableaus and nested sequents}.
\newblock {\em Annals of Pure and Applied Logic}, 163:291--313, 2012.

\bibitem[Fit14]{Fitting14}
Melvin Fitting.
\newblock Nested {Sequents} for {Intuitionistic} {Logics}.
\newblock {\em Notre Dame Journal of Formal Logic}, 55(1):41--61, 2014.

\bibitem[FM97]{Fairtlough97}
Matt Fairtlough and Michael Mendler.
\newblock Propositional {Lax} {Logic}.
\newblock {\em Information and Computation}, 137(1):1--33, 1997.

\bibitem[Gar08]{garson:stanford}
Jim Garson.
\newblock Modal logic.
\newblock In Edward~N. Zalta, editor, {\em The Stanford Encyclopedia of
  Philosophy}. Stanford University, 2008.

\bibitem[GPT09]{GorePT09}
Rajeev Gor{\'e}, Linda Postniece, and Alwen Tiu.
\newblock Taming displayed tense logics using nested sequents with deep
  inference.
\newblock In {\em TABLEAUX}, volume 5607 of {\em Lecture Notes in Computer
  Science}, pages 189--204. Springer, 2009.

\bibitem[GS10]{galmiche:salhi:10}
Didier Galmiche and Yakoub Salhi.
\newblock Label-free natural deduction systems for intuitionistic and classical
  modal logics.
\newblock {\em Journal of Applied Non-Classical Logics}, 20(4):373--421, 2010.

\bibitem[HP07]{heilala:pientka:07}
Samuli Heilala and Brigitte Pientka.
\newblock Bidirectional decision procedures for the intuitionistic
  propositional modal logic {IS4}.
\newblock In Frank Pfenning, editor, {\em CADE-21}, volume 4603 of {\em Lecture
  Notes in Computer Science}, pages 116--131. Springer, 2007.

\bibitem[Kas94]{kashima:nested}
Ryo Kashima.
\newblock Cut-free sequent calculi for some tense logics.
\newblock {\em Studia Logica}, 53(1):119--136, 1994.

\bibitem[Koj12]{Kojima12}
Kensuke Kojima.
\newblock {\em Semantical study of intuitionistic modal logics}.
\newblock PhD thesis, Kyoto University, 2012.

\bibitem[Mas92]{Masini92:2sequents:classical}
Andrea Masini.
\newblock 2-sequent calculus: a proof theory of modalities.
\newblock {\em Annals of Pure and Applied Logic}, 58(3):229 -- 246, 1992.

\bibitem[Mas93]{Masini93:2sequents:intuitionistic}
Andrea Masini.
\newblock 2-sequent calculus: Intuitionism and natural deduction.
\newblock {\em J. Log. Comput.}, 3(5):533--562, 1993.

\bibitem[MdP05]{Mendler05}
Michael Mendler and Valeria de~Paiva.
\newblock Constructive {CK} for {Contexts}.
\newblock In {\em CONTEXT}, 2005.

\bibitem[MS11]{Mendler11}
Michael Mendler and Stephan Scheele.
\newblock Cut-free gentzen calculus for multimodal {CK}.
\newblock {\em Information and Computation}, 209(12):1465--1490, 2011.

\bibitem[MS14a]{marin:str:aiml14}
Sonia Marin and Lutz Straßburger.
\newblock {Label-free Modular Systems for Classical and Intuitionistic Modal
  Logics}.
\newblock In {\em {Advances in Modal Logic 10}}, Groningen, Netherlands, August
  2014.

\bibitem[MS14b]{Mendler14}
Michael Mendler and Stephan Scheele.
\newblock On the {Computational} {Interpretation} of {CKn} for {Contextual}
  {Information} {Processing}.
\newblock {\em Fundamenta Informaticae}, 130(1):125--162, 2014.

\bibitem[PD01]{pfenning:davies:01}
Frank Pfenning and Rowan Davies.
\newblock A judgmental reconstruction of modal logic.
\newblock {\em Mathematical Structures in Computer Science}, 11(4):511--540,
  2001.

\bibitem[Pra65]{prawitz:65}
Dag Prawitz.
\newblock {\em Natural Deduction, A Proof-Theoretical Study}.
\newblock Almquist and Wiksell, 1965.

\bibitem[PS86]{plotkin:stirling:86}
Gordon~D. Plotkin and Colin~P. Stirling.
\newblock A framework for intuitionistic modal logic.
\newblock In J.~Y. Halpern, editor, {\em Theoretical Aspects of Reasoning About
  Knowledge}, 1986.

\bibitem[Ser84]{fischer-servi:84}
Gis{\`e}le.~Fischer Servi.
\newblock Axiomatizations for some intuitionistic modal logics.
\newblock {\em Rend.\ Sem.\ Mat.\ Univers.\ Politecn.\ Torino}, 42(3):179--194,
  1984.

\bibitem[Sim94]{simpson:phd}
Alex Simpson.
\newblock {\em The Proof Theory and Semantics of Intuitionistic Modal Logic}.
\newblock PhD thesis, University of Edinburgh, 1994.

\bibitem[Str13]{str:fossacs13}
Lutz Stra{\ss}burger.
\newblock Cut elimination in nested sequents for intuitionistic modal logics.
\newblock In Frank Pfenning, editor, {\em FoSSaCS'13}, volume 7794 of {\em
  LNCS}, pages 209--224. Springer, 2013.

\bibitem[TS00]{troelstra:schwichtenberg:00}
Anne~Sjerp Troelstra and Helmut Schwichtenberg.
\newblock {\em Basic Proof Theory}.
\newblock Cambridge University Press, second edition, 2000.

\bibitem[Wan94]{Wansing94}
Heinrich Wansing.
\newblock Sequent calculi for normal modal propositional logics.
\newblock {\em Journal of Logic and Computation}, 4(2):125--142, 1994.

\bibitem[Wij90]{wijesekera:90}
Duminda Wijesekera.
\newblock Constructive modal logics {I}.
\newblock {\em Ann. Pure Appl. Logic}, 50(3):271--301, 1990.

\end{thebibliography}
\bibliographystyle{alpha} 

%%%%%%%%%%%%%%%%%%%%%%%%%%%%%%%%%%%%%%%%%%%%%%%%%%%%%%%%%%%%%%%%%%%%%%%%
%%%%%%%%%%%%%%%%%%%%%%%%%%%%%%%%%%%%%%%%%%%%%%%%%%%%%%%%%%%%%%%%%%%%%%%%
%%%%%%%%%%%%%%%%%%%%%%%%%%%%%%%%%%%%%%%%%%%%%%%%%%%%%%%%%%%%%%%%%%%%%%%%
%%%%%%%%%%%%%%%%%%%%%%%%%%%%%%%%%%%%%%%%%%%%%%%%%%%%%%%%%%%%%%%%%%%%%%%%

\end{document}